\documentclass[a4paper,twocolumn,11pt,accepted=2022-08-31]{quantumarticle}
\pdfoutput=1
\usepackage[utf8]{inputenc}
\usepackage[english]{babel}
\usepackage[T1]{fontenc}
\usepackage{tikz}
\usepackage{amsmath}
\usepackage{amsfonts}
\usepackage{csquotes}
\usepackage{cite}
\usepackage{amssymb}
\usepackage{mathtools}
\usepackage{amsthm}
\usepackage{enumerate}
\usepackage{dsfont}
\usepackage{graphicx}
\usepackage{authblk}
\usepackage{hyperref}
\usepackage{epstopdf}
\usepackage{theoremref}
\usepackage{subfig}
\usepackage{bbm}
\usepackage{bm}
\usepackage{qcircuit}
\usepackage{diagbox}
\usepackage{ltxgrid}
\usetikzlibrary{arrows,decorations.pathmorphing,backgrounds,positioning,fit}
\usetikzlibrary{positioning,chains,fit,shapes,calc,math,arrows.meta}
\tikzstyle{vertex}=[circle, draw, inner sep=0pt, minimum size=6pt]

\usepackage{algorithm}
\usepackage{algpseudocode}
\usepackage{wrapfig}

\usepackage{blkarray, bigstrut}
\usepackage{booktabs}
\usepackage[stable]{footmisc}

\newtheorem{theorem}{Theorem}
\newtheorem{definition}[theorem]{Definition}
\newtheorem{lemma}[theorem]{Lemma}

\newtheorem{remark}[theorem]{Remark}
\newtheorem{example}{Example}

\newcommand{\llbr}{[\![}
\newcommand{\rrbr}{]\!]}

\newcommand{\C}{\mathbb C}

\newcommand{\Z}{\mathbb Z}
\newcommand{\F}{\mathbb F}

\newcommand\numberthis{\addtocounter{equation}{1}\tag{\theequation}}
\newcommand{\ket}[1]{|{#1}\rangle}
\newcommand{\bra}[1]{\langle{#1}|}

\begin{document}

\title{Designing the Quantum Channels Induced by Diagonal Gates}

\author{Jingzhen Hu}
\email{jingzhen.hu@duke.edu}
\orcid{0000-0002-2699-3966}
\affiliation{Department of Mathematics, Duke University, Durham, NC 27708, USA}
\thanks{}
\author{Qingzhong Liang}
\email{qingzhong.liang@duke.edu}
\orcid{0000-0002-5073-9431}
\affiliation{Department of Mathematics, Duke University, Durham, NC 27708, USA}
\thanks{}
\author{Robert Calderbank}
\email{robert.calderbank@duke.edu}
\orcid{0000-0003-2084-9717}
\affiliation{Department of Mathematics, Duke University, Durham, NC 27708, USA}
\affiliation{Department of Electrical and Computer Engineering, Department of Computer Science, Duke University, NC 27708, USA}
\thanks{\newline J.H and Q.L contributed equally to this work.}
% \thanks{You can use the \texttt{\textbackslash{}email}, \texttt{\textbackslash{}homepage}, and \texttt{\textbackslash{}thanks} commands to add additional information for the preceding \texttt{\textbackslash{}author}. If applicable, this can also be used to indicate that a work has previously been published in conference proceedings.}
\maketitle

\begin{abstract}
% \blfootnote{The first two authors contributed equally to this work.}
  	The challenge of quantum computing is to combine error resilience with universal computation. Diagonal gates such as the transversal $T$ gate play an important role in implementing a universal set of quantum operations. This paper introduces a framework that describes the process of preparing a code state, applying a diagonal physical gate, measuring a code syndrome, and applying a Pauli correction that may depend on the measured syndrome (the average logical channel induced by an arbitrary diagonal gate). It focuses on CSS codes, and describes the interaction of code states and physical gates in terms of generator coefficients determined by the induced logical operator. The interaction of code states and diagonal gates depends very strongly on the signs of $Z$-stabilizers in the CSS code, and the proposed generator coefficient framework explicitly includes this degree of freedom. The paper derives necessary and sufficient conditions for an arbitrary diagonal gate to preserve the code space of a stabilizer code, and provides an explicit expression of the induced logical operator. When the diagonal gate is a quadratic form diagonal gate (introduced by Rengaswamy et al.), the conditions can be expressed in terms of divisibility of weights in the two classical codes that determine the CSS code. These codes find application in magic state distillation and elsewhere. When all the signs are positive, the paper characterizes all possible CSS codes, invariant under transversal $Z$-rotation through $\pi/2^l$, that are constructed from classical Reed-Muller codes by deriving the necessary and sufficient constraints on $l$. The generator coefficient framework extends to arbitrary stabilizer codes but there is nothing to be gained by considering the more general class of non-degenerate stabilizer codes.
\end{abstract}

	\section{Introduction and Review\footnote{Section \ref{sec:prelims} introduces notation and provides technical background for the results described in this section.}} 
	\label{sec:Intro}
% 	Quantum information science is distinguished by the interplay between error correction and logical computation.  
	    	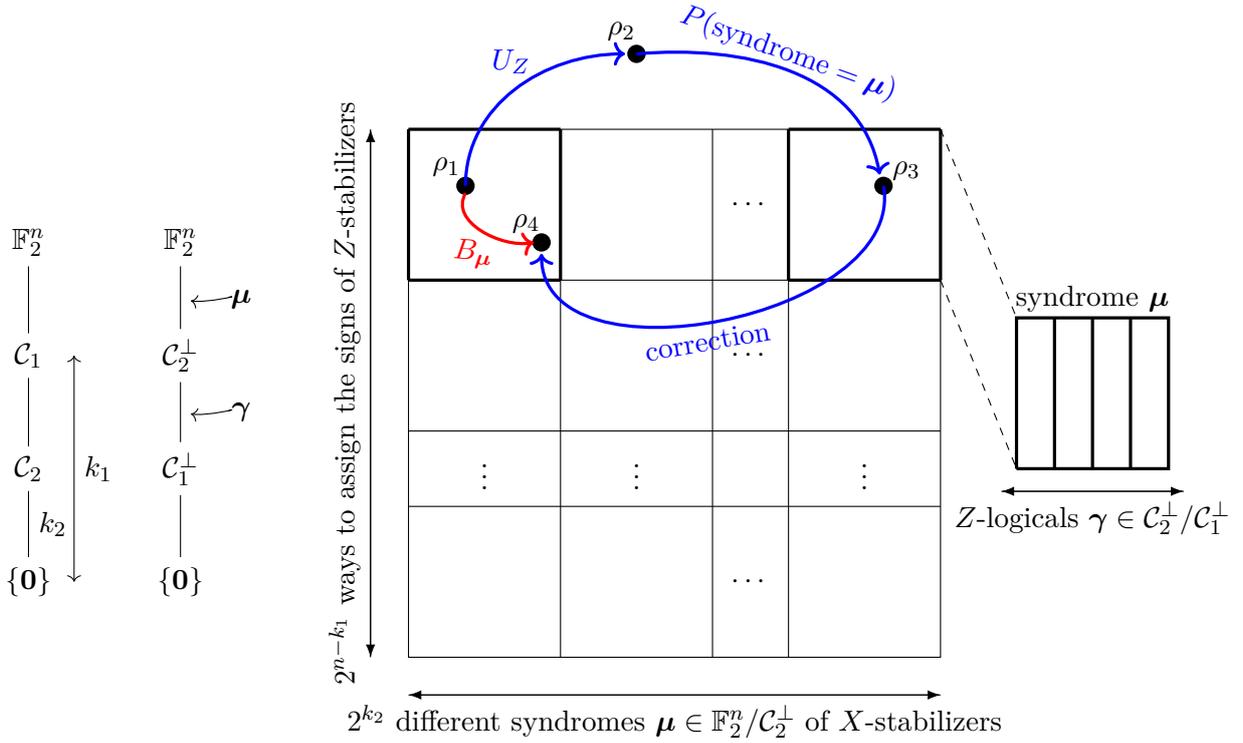
\begin{figure*}[t]
				\centering
			\begin{tikzpicture}
			\node (Z) at (0,0) {$\{ \bm{0} \}$};
			\node (C2) at (0,1.5) {$\mathcal{C}_2$};
			\node (C1) at (0,3) {$\mathcal{C}_1$};
			\node (F2m) at (0,4.5) {$\mathbb{F}_2^{n}$};
% 			\node[text width=3.5cm] at (0.6,1) {$S_X$~= };
			\path[draw] (Z) -- (C2) node[midway,right] {{$k_2$}} -- (C1) node[midway,left] {{}} -- (F2m) node[midway,left] {{}};
			\path[draw,<->,black] (0.6,0) -- (0.6,3) node [midway,right] {$k_1$}; 
			%\path[draw,<->,black] (0.4,0) -- (0.4,1.5) node [midway,right] {$k$}; 
			\node (Zp) at (2,0) {$\{ \bm{0} \}$};
			\node (C1p) at (2,1.5) {$\mathcal{C}_1^{\perp}$};
			\node (C2p) at (2,3) {$\mathcal{C}_2^{\perp}$};
			\node (F2m) at (2,4.5) {$\mathbb{F}_2^{n}$};
			\node (mu) at (2.8,3.75) {$\bm{\mu}$};
			\node (frommu) at (2.5,3.75) {};
			\node (tomu) at (2,3.75) {};
			\node (gamma) at (2.8,2.25) {$\bm{\gamma}$};
			\node (fromgamma) at (2.5,2.25) {};
			\node (togamma) at (2,2.25) {};
% 			\node[text width=3.5cm] at (4.1,1) {=~$S_Z$};
			\draw[->,black] (frommu) to [out=10,in=-10] (tomu);
			\draw[->,black] (fromgamma) to [out=10,in=-10] (togamma);
			
			\path[draw] (Zp) -- (C1p) node[midway,left] {{}} -- (C2p) node[midway,left] {{}} -- (F2m) node[midway,left] {{}};
			%%%%%%%%%%%%%%%%%%%%%%%%%%%%%%%%%%%%%%%%%%%%%%%%%%%%%%%%
			\tikzmath{\x1 = 5; \y1 =-1; 
            % Computations are also possible
            \x2 = \x1 + 7; \y2 =\y1 +7; } 

			\draw[] (\x1,\y1) -- (\x2,\y1) -- (\x2,\y2) -- (\x1,\y2) -- (\x1,\y1);
		    \draw[] (\x1,\y1+2) -- (\x2,\y1+2);
		    \draw[] (\x1,\y1+3) -- (\x2,\y1+3);
		    \draw[] (\x1,\y1+5) -- (\x2,\y1+5);
		    \draw[] (\x1+2,\y1) -- (\x1+2,\y2);
		    \draw[] (\x1+4,\y1) -- (\x1+4,\y2);
		    \draw[] (\x1+5,\y1) -- (\x1+5,\y2);
		  %  \path[draw,<->] (\x2+0.5,\y1) -- (\x2+0.5,\y2) node [midway,right] {\rotatebox{90}{$2^{n-k_1}$ different signs of $Z$-stabilizers}}; 
		  %  \draw[-Latex] (\x2+0.5,\y1) -- (\x2+0.5,\y2) node [midway,right] {\rotatebox{90}{$2^{n-k_1}$ different signs of $Z$-stabilizers ($\F_2^n/\mathcal{C}_1$)}}; 
		  %  \draw[Latex-] (\x2+0.5,\y1) -- (\x2+0.5,\y2) node [midway,right] {};
		    \draw[-Latex] (\x1-0.5,\y1) -- (\x1-0.5,\y2) node [midway,left] {\rotatebox{90}{$2^{n-k_1}$ ways to assign the signs of $Z$-stabilizers}}; %($\F_2^n/\mathcal{C}_1$) 
		    \draw[Latex-] (\x1-0.5,\y1) -- (\x1-0.5,\y2) node [midway,left] {};
		    
		    \draw[-Latex] (\x1,\y1-0.5) -- (\x2,\y1-0.5) node [midway,below] {$2^{k_2}$ different syndromes $\bm{\mu} \in \F_2^n/\mathcal{C}_2^\perp$ of $X$-stabilizers}; 
		    \draw[Latex-] (\x1,\y1-0.5) -- (\x2,\y1-0.5) node [midway,below] {}; 
		    
		    \node (vdots1) at (\x1+1,\y1+2.5) {$\vdots$};
		    \node (vdots2) at (\x1+3,\y1+2.5) {$\vdots$};
		    \node (vdots2) at (\x1+6,\y1+2.5) {$\vdots$};
		    \node (cdots1) at (\x2-2.5,\y1+1) {$\cdots$};
		    \node (cdots2) at (\x2-2.5,\y1+4) {$\cdots$};
		    \node (cdots2) at (\x2-2.5,\y1+6) {$\cdots$};
            %%%%%%%%%%%%%%%%%%%%%%%%%%%%%%%%%%%%%%%%%%%%%%%%%%%%%%%%
            \draw[very thick] (\x1,\y2) -- (\x1+2,\y2);
            \draw[very thick] (\x1,\y2) -- (\x1,\y2-2);
            \draw[very thick] (\x1+2,\y2-2) -- (\x1+2,\y2);
            \draw[very thick] (\x1+2,\y2-2) -- (\x1,\y2-2);
            
            \draw[very thick] (\x2,\y2) -- (\x2-2,\y2);
            \draw[very thick] (\x2,\y2) -- (\x2,\y2-2);
            \draw[very thick] (\x2-2,\y2-2) -- (\x2-2,\y2);
            \draw[very thick] (\x2-2,\y2-2) -- (\x2,\y2-2);
            %%%%%%%%%%%%%%%%%%%%%%%%%%%%%%%%%%%%%%%%%%%%%%%%%%%%%%%%
            \filldraw[very thick](\x1+0.75,\y2-0.75) circle (0.1);
            \node (rho1) at (\x1+0.5,\y2-0.5) {$\rho_1$};
            
            \filldraw[very thick](\x1+3,\y2+1) circle (0.1);
            \node (rho2) at (\x1+2.8,\y2+1.3) {$\rho_2$};
            
            \filldraw[very thick](\x2-0.75,\y2-0.75) circle (0.1);
            \node (rho3) at (\x2-0.45,\y2-0.55) {$\rho_3$};
            
            \filldraw[very thick](\x1+1.75,\y2-1.5) circle (0.1);
            \node (rho4) at (\x1+1.55,\y2-1.25) {$\rho_4$};
            
            \draw[->,blue,very thick] (\x1+0.75,\y2-0.75) to [out=90,in=180] (\x1+2.85,\y2+1);
            \node[blue] (U_Z) at (\x1+1.35,\y2+0.9) {$U_Z$};
            
            \draw[->,blue,very thick] (\x1+3,\y2+1) to [out=5,in=100] (\x2-0.8,\y2-0.6);
            \node[blue] (msmt) at (\x2-2,\y2+1) {\rotatebox{-20}{$P(\text{syndrome}=\bm{\mu})$}};
            % {prob. $p_{\bm{\mu}}$ after measurement}};
            
            \draw[->,blue,very thick] (\x2-0.75,\y2-0.75) to [out=-80,in=-90] (\x1+1.75,\y2-1.65);
            \node[blue] (msmt) at (\x2-3,\y2-2.8) {\rotatebox{10}{correction %$\epsilon_{\bm{\gamma_{\bm{\mu}}}\oplus \bm{\mu}}E(\bm{0},\bm{\gamma_{\bm{\mu}}}\oplus \bm{\mu})$
            }};
            
            \draw[->,red,very thick] (\x1+0.75,\y2-0.85) to [out=-120,in=-170] (\x1+1.65,\y2-1.5);
            \node[red] (ILO) at (\x1+0.85,\y2-1.65) {\rotatebox{-5}{$B_{\bm{\mu}}$}};
            % \draw[-Latex,,blue,thick,out=10,in=180] (\x1+0.85,\y2-0.75) -- (\x1+2.85,\y2+1) node [midway,below] {Diagonal Gate $U_Z$};
		    %%%%%%%%%%%%%%%%%%%%%%%%%%%%%%%%%%%%%%%%%%%%%%%%%%%%%%%%
			\draw[very thick] (\x2+1,\y2-4.5) -- (\x2+3,\y2-4.5) -- (\x2+3,\y2-2.5) -- (\x2+1,\y2-2.5) -- (\x2+1,\y2-4.5);
			
			\draw[dashed] (\x2,\y2) to (\x2+1,\y2-2.5);
			\draw[dashed] (\x2,\y2-2) to (\x2+1,\y2-4.5);
			
			\draw[very thick] (\x2+1.5,\y2-4.5) -- (\x2+1.5,\y2-2.5);
			\draw[very thick] (\x2+2,\y2-4.5) -- (\x2+2,\y2-2.5);
			\draw[very thick] (\x2+2.5,\y2-4.5) -- (\x2+2.5,\y2-2.5);
			\node (mu) at (\x2+2,\y2-2.25) {syndrome $\bm{\mu}$};
			
			 \draw[-Latex] (\x2+3.2,\y2-4.8) -- (\x2+0.8,\y2-4.8) node [midway,below] {$Z$-logicals $\bm{\gamma}\in \mathcal{C}_2^\perp /\mathcal{C}_1^\perp$}; 
		    \draw[Latex-] (\x2+3.2,\y2-4.8) -- (\x2+0.8,\y2-4.8) node [midway,below] {}; 
		    
% 			\node (gamma) at (\x2+2,\y2-4.8) {$Z$-logicals $\bm{\gamma}\in \mathcal{C}_2^\perp /\mathcal{C}_1^\perp$};
			%%%%%%%%%%%%%%%%%%%%%%%%%%%%%%%%%%%%%%%%%%%%%%%%%%%%%%%%
			\end{tikzpicture}
			\caption{
			The $2^{n-k_1}$ rows of the array are indexed by the $\llbr n,k_1-k_2,d \rrbr $ CSS codes corresponding to all possible signings of the $Z$-stabilizer group. The $2^{k_2}$ columns of the array are indexed by all possible $X$-syndromes $\bm{\mu}$. The logical operator $B_{\bm{\mu}}$ is induced by (1) preparing any code state $\rho_1$; (2) applying a diagonal physical gate $U_Z$ to obtain $\rho_2$; (3) using $X$-stabilizers to measure $\rho_2$, obtaining the syndrome $\bm{\mu}$ with probability $p_{\bm{\mu}}$, and the post-measurement state $\rho_3$; (4) applying a Pauli correction to $\rho_3$, obtaining $\rho_4$. The generator coefficients $A_{\bm{\mu},\bm{\gamma}}$ are obtained by expanding the logical operator $B_{\bm{\mu}}$ in terms of $Z$-logical Pauli operators $\epsilon_{(\bm{0},\bm{\gamma})}E(\bm{0},\bm{\gamma})$, where $\epsilon_{(\bm{0},\bm{\gamma})}\in\{\pm 1\}$. 
			}
			\label{fig:GC_frame} 
		\end{figure*}
	We approach quantum computing through fault tolerant implementation of a universal set of gates. There are many finite sets of gates that are universal% for quantum computing
	, and a standard choice is to augment the set of Clifford gates by a non-Clifford unitary \cite{boykin1999universal} such as the $T$ gate ($\pi/8$ rotation). Gottesman and Chuang \cite{gottesman1999demonstrating} introduced the \emph{Clifford hierarchy} of unitary operators. The first level is the \emph{Pauli group}. The second level is the \emph{Clifford group}, which consists of unitary operators that normalize the Pauli group. The $l$-th level consists of unitary operators that map Pauli operators to the $(l-1)$-th level under conjugation. The teleportation model of quantum computation introduced in \cite{gottesman1999demonstrating} is closely related to the structure of the Clifford hierarchy (for details, see \cite{zeng2008semi,beigi2008c3,bengtsson2014order,Anderson2016Classification, cui2017diagonal,rengaswamy2019unifying,pllaha2020weyl}). The diagonal gates in the Clifford hierarchy form a group \cite{zeng2008semi,cui2017diagonal}, and the diagonal entries are $2^l$-th roots of unity raised to some polynomial function of the qubit state. Cui et al. \cite{cui2017diagonal} determined the level of a diagonal gate in the Clifford hierarchy in terms of $l$ and the degree of the polynomial function. \emph{Quadratic form diagonal} (QFD) gates are a family of diagonal gates associated with quadratic forms. The class of QFD gates includes transversal $Z$-rotations through $\pi/2^l$, and encompasses all $2$-local gates in the hierarchy \cite{rengaswamy2019unifying}. 
		
	Quantum error-correcting codes (QECCs) protect information as it is transformed by logical gates. 
% 	are designed to control the overall noise level of realizing the logical gates on the protected information. 
    In general, a logical non-Clifford gate is more difficult to implement than a logical Clifford gate \cite{gottesman1998heisenberg}. Any non-Clifford operation on the $k$ logical qubits of an $\llbr  n,k,d  \rrbr $ QECC must be induced by a non-Clifford operation on the $n$ physical qubits \cite{cui2017diagonal}. We derive a \emph{global} necessary and sufficient condition for any diagonal physical gate to preserve the code space of a stabilizer code \cite{gottesman1997stabilizer,calderbank1998quantum}. A \emph{transversal} gate \cite{gottesman1997stabilizer} is a tensor product of unitaries on individual code blocks. In the case of transversal $Z$-rotation through $\pi/2^l$, we show that this global condition is equivalent to the \emph{local} trigonometric conditions derived by Rengaswamy et al. \cite{rengaswamy2020optimality}. Our approach has the advantage of providing insight into the induced logical operator.
		
		It is essential that a set of gates be both universal and fault-tolerant. Fault-tolerance of transversal gates follows from the observation that uncorrelated errors remain uncorrelated in code blocks. The Eastin-Knill Theorem \cite{eastin2009restrictions} reveals that we cannot implement a universal set of logical operations on a QECC using transversal operations alone. Magic state distillation (MSD) combines transversal gates with an ancillary magic state to circumvent this restriction \cite{bravyi2005universal,reichardt2005quantum,bravyi2012magic,anwar2012qutrit,campbell2012magic,landahl2013complex,campbell2017unified,haah2018codes,krishna2019towards,vuillot2019quantum}. If the initial fidelity of magic state exceeds a certain threshold, then it can be purified by successive application of the quantum teleportation protocol on stabilizer codes that are able to realize a logical non-Clifford gate. \emph{(Generalized) triorthogonal codes} \cite{bravyi2012magic,haah2018codes} are \emph{Calderbank-Shor-Steane} (CSS) codes \cite{Calderbank-physreva96,Steane-physreva96} designed to implement a non-Clifford logical gate (up to some diagonal Clifford logical gates). Hamming weights in the classical codes that determine the CSS codes are required to satisfy certain divisibility properties \cite{campbell2012magic,landahl2013complex,haah2018towers,vuillot2019quantum,nezami2021classification}. Many examples employ Reed-Muller (RM) codes. In Section 5 we characterize CSS codes constructed from classical RM codes that are fixed by transversal $Z$-rotation through $\pi/2^l$.
		
		MSD provides a path to universal fault tolerant computation, where success depends on engineering the interaction of code states and physical gates. Here we consider the interaction of a diagonal physical gate $U_Z$ with the code states of a stabilizer code, as shown in Figure 1. We prepare an initial code state, apply a physical gate, then measure a code syndrome $\bm{\mu}$, and finally apply a correction based on $\bm{\mu}$. For each syndrome, we expand the induced logical operator in the Pauli basis to obtain the \emph{generator coefficients} that capture state evolution. Intuitively, the diagonal physical gate preserves the code space if and only if the induced logical operator corresponding to the trivial syndrome is unitary.  

   The effectiveness of magic state distillation (MSD) depends on the probability of observing a given syndrome, and it is possible to combine syndrome measurement with a decoder (see Krishna and Tillich \cite{krishna2019towards} for example). Generator coefficients provide a framework for investigating the effectiveness and the threshold of distillation. We describe the design space that is available through a running example.
   
   \begin{example}[The $\llbr 7,1,3 \rrbr $ Steane code]
   \normalfont
   \label{examp1_in_intro}
   Reichardt \cite{reichardt2005quantum} demonstrated that it is possible to distill the magic state $\ket{A}=(\ket{0}+e^{\imath \pi/4}\ket{1})/\sqrt{2}$ by post-selecting on the trivial syndrome, even though the Steane code is not perfectly preserved by the transversal T gate. He also demonstrated the distillation threshold is optimal for $\ket{A}$. In Section \ref{sec:avg_log_chl}, we use generator coefficients to describe the average-logical channel induced by the transversal $T$ gate on the Steane code. When we observe the trivial syndrome, the induced logical operator is $T^\dagger$. Otherwise it is a logical Pauli $Z$ followed by a logical $T^\dagger$. The induced logical operator becomes $T^\dagger$ for all syndromes after applying a logical Pauli $Z$ correction to all non-trivial syndromes. However, the distillation protocol no longer converges, despite the higher probability of success\footnote{See Appendix \ref{sec:MSD_Steane}}. Generator coefficients encode the probabilities of observing different syndromes, which can be used to analyze variants of the Steane protocol (such as applying a decoder to subsets of syndromes), as well as MSD protocols that use different codes (such as the $\llbr 15,1,3 \rrbr $ code). 
   \end{example}
    
    The introduction of magic state distillation by Bravyi and Kitaev \cite{bravyi2005universal} led to the construction of CSS codes where the code space is preserved by a transversal $Z$-rotation of the underlying physical space \cite{bravyi2005universal,reichardt2005quantum,bravyi2012magic,campbell2012magic,landahl2013complex,campbell2017unified,haah2018codes,vuillot2019quantum}. The approach taken in each paper is to examine the action of a transversal $Z$-rotation on the basis states of a CSS code. This approach results in \emph{sufficient} conditions for a transversal $Z$-rotation to realize a logical operation on the code space. In contrast we derive \emph{necessary and sufficient} conditions by analyzing the action of a transversal diagonal gate on the stabilizer group that determines the code. In effect, we study the code space by studying symmetries of the codespace.
    
    The interaction of transversal physical operators and code states depends very strongly on the signs of stabilizers \cite{coherent_noise,debroy2021optimizing}. Consider for example, the design of CSS codes that are oblivious to coherent noise. We can model the effective error as a uniform $Z$-rotation on each qubit through some (small) angle $\theta$. We require the noise to preserve the code space and to act trivially (as the logical identity operator). It is possible to demonstrate the existence of weight-$2$ $Z$-stabilizers, and to show that their signs must be balanced \cite{coherent_noise}. Our generator coefficient framework includes the freedom to choose signs and this degree of freedom is relatively unexplored. We describe the design space that is available through a running example. 
    
    \begin{example}[The $\llbr 4,2,2 \rrbr $ code]
    \normalfont
   \label{examp2_in_intro}
   
   Generator coefficients encode correlation between the initial code state and syndrome measurement, which may result in loss of logical information. The $\llbr 4,2,2 \rrbr $ code shows that correlation can depend very strongly on the signs of $Z$-stabilizers. The stabilizer group is $\mathcal{S}=\langle X^{\otimes 4}, Z^{\otimes 4} \rangle$. In Section \ref{sec:avg_log_chl} we show that if $Z^{\otimes 4}$ has a positive sign, then there is an embedded decoherence free subspace spanned by the three encoded basis states $\ket{\overline{01}}$, $\ket{\overline{10}}$, and $\ket{\overline{11}}$. We also show that syndrome measurement collapses logical information. If $Z^{\otimes 4}$ has a negative sign, then we show that logical information does not collapse, but the embedded decoherence free subspace disappears. Generator coefficients encode the different ways that code states can evolve. 
   \end{example}
	
	We now summarize our main technical contributions.
	\begin{enumerate}[1)]
	    \item We derive an explicit expression for the logical channel induced by a diagonal physical gate (Section \ref{sec:avg_log_chl}, \eqref{eqn:kraus_ops} describes the induced logical operator for each syndrome $\bm{\mu}$ and \eqref{eqn:prob_mu} describes the probability of observing $\bm{\mu}$). We quantify the correlation between initial code state and measured syndrome by separating the probability of observing a given syndrome into two components, one depending on the generator coefficients, the other on the choice of initial state (Section \ref{subsec:prob}). We analyze the $\llbr 4,2,2 \rrbr $ code (Example \ref{examp2_in_intro}) to show that each component depends strongly on the choice of signs in the stabilizer code, and that we can choose signs to create a embedded decoherence free subspace. 
	    
	    \item We derive \emph{necessary and sufficient conditions} for an arbitrary diagonal physical gate to preserve the codespace of a CSS code with arbitrary signs (Section \ref{sec:nece_suff_cond_invariant}, Theorem \ref{thm:preserved_by_Uz}), and describe the logical operator that results (Section \ref{sec:nece_suff_cond_invariant}, Remark \ref{rem:log_op_U_Z}). These conditions generalize earlier conditions found by Rengaswamy et al \cite{rengaswamy2020optimality} for transversal $Z$-rotation through $\pi/2^l$.
	    
	    \item We further simplify the \emph{necessary and sufficient conditions} for a QFD gate to preserve the code space of a CSS code (Section \ref{sec:nece_suff_cond_invariant}, Theorem \ref{thm:div_cond_qfd}). These conditions govern divisibility of Hamming weights in the classical codes that determine the CSS codes. In the case of transversal $Z$-rotation through $\pi/2^l$ applied to CSS codes with positive signs, we show the necessity of divisibility conditions derived in \cite{landahl2013complex,vuillot2019quantum}.  
	    
	    \item We characterize \emph{all} CSS codes with positive signs, invariant under transversal $Z$-rotation through $\pi/2^l$, that are constructed from classical Reed-Muller (RM) codes (and their derivatives obtained by puncturing or removing the first coordinate). We derive \emph{necessary and sufficient} conditions that relate $l$ to the parameters of the component RM codes (Section \ref{sec:nece_suff_cond_invariant}, Theorem \ref{thm:RM_cosntruction} and Remark \ref{rem:elemtary_op}). 
	    
	    \item We extend the generator coefficient framework to stabilizer codes (Appendix \ref{sec:stab_gcf}). This extension shows that given an $\llbr n, k, d \rrbr $ \emph{non-degenerate} stabilizer code preserved by a diagonal gate $U_Z$, we can construct an $\llbr n, k, d_Z \ge d \rrbr $ CSS code preserved by $U_Z$ with the same induced logical operator. Note that $d_Z$ (the minimum weight of any nontrivial $Z$-logical Pauli operator) is the relevant distance for MSD. Recall that an $\llbr n, k, d \rrbr $ stabilizer code is non-degenerate if the weight of every stabilizer element is at least $d$. 
	\end{enumerate}
    
	The rest of the paper is organized as follows. Section \ref{sec:prelims} introduces notation and provides the necessary background. 
% 	including the Reed-Muller codes and MacWilliams Identities from classical coding theory. 
    Our review of stabilizer codes takes account of the freedom to choose signs in the stabilizer group, and provides the general encoding map and logical Pauli operators for CSS codes with arbitrary signs. % in subsections \ref{subsec:CSS_general_Enc} and \ref{subsec:general_logical_op_CSS}. 
    Section \ref{sec:intro_gcs} introduces the generator coefficients that describe how a diagonal gate acts on a CSS code. Section \ref{sec:avg_log_chl} describes how generator coefficient govern the average logical channel. Section \ref{sec:nece_suff_cond_invariant} establishes necessary and sufficient conditions for a CSS code to support a diagonal physical gate, and derives the induced logical operator. We then derive the divisibility conditions and introduce RM constructions. Section \ref{sec:concln} concludes the paper and discusses future directions. 
	In Appendix \ref{sec:stab_gcf}, we extends the generator coefficient framework to general stabilizer codes and show that CSS codes perform at the least as well as non-degenerate stabilizer codes for diagonal gates. % by deriving the Kraus operators of a $Z$-unitary acting inside a CSS codespace and the probability of observing different syndromes.
	
	\section{Preliminaries and Notation}
	\label{sec:prelims}
	
	\subsection{Classical Reed-Muller Codes}
	\label{subsec:prem_RM_codes}
	 Let $\F_2 = \{0,1\}$ denote the binary field. Let $m\ge1$, and let  $x_1$, $x_2$, $\dots$, $x_m$ be binary variables (monomials of degree $1$). Monomials of degree $r$ can be written as $x_{i_1}x_{i_2}\cdots x_{i_r}$ where $i_j \in \{1,2,\dots,m\}$ are distinct. A boolean function with degree $r$ is a binary linear combination of monomials with degrees at most $r$. There is a one-to-one correspondence between boolean functions $h$ and evaluation vectors $\bm{h}=[h(x_1,x_2,\cdots,x_m)]_{(x_1,x_2,\dots,x_m)\in\F_2^{m}}$. The degree $0$ boolean function corresponds to the constant evaluation vector $\bm{1}\in \F_2^{2^m}$.
    
    For $0\le r\le m$, the Reed-Muller code RM$(r,m)$ is the set of all evaluation vectors $\bm{h}$ associated with boolean functions $h(x_1,x_2,\cdots,x_m)$ of degree at most $r$, 
    $\mathrm{RM}(r,m) \coloneqq \{\bm{h}\in \F_2^{2^m}\mid h\in\F_2[x_1,x_2,\cdots,x_m],$ $\deg(h)\le r\}.$ 
    The length of the RM$(r,m)$ code is $2^m$, the dimension is given by $k=\sum_{j=0}^{r}\binom{m}{j}$, and the minimal distance is $2^{m-r}$. The dual of RM$(r,m)$ is RM$(m-r-1,m)$, and we can construct the RM codes by a recursively observing RM$(r,m+1) = \{(\bm{u},\bm{u}+\bm{v})\mid \bm{u}\in \mathrm{RM}(r,m), \bm{v}\in \mathrm{RM}(r-1,m)\}$ \cite{macwilliams1977theory}. Note that all weights in RM($r,m$) are multiples of $2^{\left\lfloor(m-1)/r\right\rfloor}$ \cite{ax1964zeroes,mceliece1971periodic,macwilliams1977theory}, and the highest power of $2$ that divides all weights of codewords in RM($r,m$) is exactly $2^{\left\lfloor(m-1)/r\right\rfloor}$ \cite{borissov2013mceliece}. % that is, there exists a codeword $\bm{x} \in \mathrm{RM}(r,m)$ such that $2^{\left\lfloor(m-1)/r\right\rfloor+1} \nmid w_H(\bm{x})$.}
	
	\subsection{The MacWilliams Identities}
    \label{subsec:prem_Mac_RM}
    Let $\imath\coloneqq \sqrt{-1}$ be the imaginary unit. We denote the Hamming weight of a binary vector $\bm{v}$ by $w_H(\bm{v})$. 
    %$\mathop{\rm wgt}(\bm{v})$. 
    The weight enumerator of a binary linear code $\mathcal{C} \subset \F_2^m$ is the polynomial
    \begin{equation}
    P_{\mathcal{C}}(x,y) = \sum_{\bm{v}\in \mathcal{C}} x^{m-w_H\left(\bm{v}\right)}y^{w_H\left(\bm{v}\right)}.
    % P_{\mathcal{C}}(x,y) = \sum_{\bm{v}\in \mathcal{C}} x^{m-\mathop{\rm wgt}\left(\bm{v}\right)}y^{\mathop{\rm wgt}\left(\bm{v}\right)}.
    \end{equation}
    The MacWilliams Identities \cite{Mac} relate the weight enumerator of a code $\mathcal{C}$ to that of the dual code $\mathcal{C}^\perp$, and are given by
    \begin{equation}
    P_{\mathcal{C}}(x,y) = \frac{1}{|\mathcal{C}^\perp|} P_{\mathcal{C}^\perp}(x+y,x-y).
    \end{equation}
    Given an angle $\theta \in (0,2\pi)$, we make the substitution $x=\cos\frac{\theta}{2}$ and $y=-\imath\sin\frac{\theta}{2}$, and define
    \begin{align}
    P_{\theta}[\mathcal{C}] 
    &\coloneqq P_{\mathcal{C}}\left(\cos\frac{\theta}{2},-\imath\sin\frac{\theta}{2}\right) \\
    &= \sum_{\bm{v}\in \mathcal{C}} \left(\cos\frac{\theta}{2}\right)^{m-w_H(\bm{v})}\left(-\imath\sin\frac{\theta}{2}\right)^{w_H(\bm{v})}. \label{eqn:key_sub_Mac}
    %&= \sum_{\bm{v}\in \mathcal{C}} \left(\cos\frac{2\pi}{2^l}\right)^{m-\mathop{\rm wgt}(\bm{v})}\left(\imath\sin\frac{2\pi}{2^l}\right)^{\mathop{\rm wgt}(\bm{v})}. 
    \end{align}

    \subsection{The Pauli Group}
    \label{subsec:prem_paulis}
    Any $2\times 2$ Hermitian matrix can be uniquely expressed as a real linear combination of the four single qubit Pauli matrices/operators
    \begin{align}
    I_2 \coloneqq \begin{bmatrix}
    1 & 0\\
    0 & 1
    \end{bmatrix},~  
    X \coloneqq \begin{bmatrix} 
    0 & 1\\
    1 & 0
    \end{bmatrix},~  
    Z \coloneqq  \begin{bmatrix} 
    1 & 0\\
    0 & -1
    \end{bmatrix}, 
    % ~ \text{and }  Y \coloneqq  \left[\begin{array}{cc} % sigma_Z
    % 0 & -\imath\\
    % \imath & 0
    % \end{array} \right],
    \end{align}
    and $Y\coloneqq \imath XZ$.
    The operators satisfy 
    $
    X^2= Y^2= Z^2=I_2,~  X Y=- Y X,~  X Z=- Z X,~ \text{ and }  Y Z=- Z Y.
    $
    
    Let $A \otimes B$ denote the Kronecker product (tensor product) of two matrices $A$ and $B$. Let $n\ge 1$ and $N=2^n$. Given binary vectors $\bm{a}=[a_1,a_2,\dots,a_n]$ and $\bm{b}=[b_1,b_2,\dots,b_n]$ with $a_i,b_j =0$ or $1$, we define the operators
    \begin{align}
    % D(\bm{a},\bm{b})&\coloneqq \sigma_X^{a_1} \sigma_Z^{\beta_1}\otimes  \sigma_X^{a_2} \sigma_Z^{\beta_2}\otimes \cdots \otimes  \sigma_X^{a_n} \sigma_Z^{\beta_n},\\
    D(\bm{a},\bm{b})&\coloneqq X^{a_1} Z^{b_1}\otimes \cdots \otimes  X^{a_n} Z^{b_n},\\
    E(\bm{a},\bm{b}) &\coloneqq\imath^{\bm{a}\bm{b}^T \bmod 4}D(\bm{a},\bm{b}).
    %E(\bm{a},\bm{b}) &\coloneqq\left(\imath^{a_1 b_1} \sigma_X^{a_1} \sigma_Z^{b_1}\right)\otimes \cdots\otimes \left(\imath^{a_n b_n} \sigma_X^{a_n} \sigma_Z^{b_n}\right)
    %	\\&=\imath^{\bm{ab}^T \bmod 4}D(\bm{a},\bm{b}).
    \end{align}
    
    We often abuse notation and write $\bm{a}, \bm{b} \in \F_2^n$, though entries of vectors are sometimes interpreted in $\mathbb{Z}_4 = \{ 0,1,2,3 \}$. Note that $D(\bm{a},\bm{b})$ can have order $1,2$ or $4$, but $E(\bm{a},\bm{b})^2=\imath^{2\bm{a}\bm{b}^T}D(\bm{a},\bm{b})^2=\imath^{2ab^T}( \imath^{2\bm{a}\bm{b}^T} I_N)=I_N$. The $n$-qubit \textit{Pauli group} is defined as
    \begin{equation}
    \mathcal{HW}_N \coloneqq\{\imath^\kappa D(\bm{a},\bm{b}): \bm{a},\bm{b}\in \F_2^n, \kappa\in \Z_{4} \},
    \end{equation}
    where $\Z_{2^l} = \{0,1,\dots,2^l-1\}$.
    The $n$-qubit Pauli matrices form an orthonormal basis for the vector space of $N\times N$ complex matrices ($\C^{N\times N}$) under the normalized Hilbert-Schmidt inner product $\langle A,B\rangle \coloneqq \mathrm{Tr}(A^\dagger B)/N$ \cite{gottesman1997stabilizer}. %$\C^{N\times N}$
    
    We use the \textit{Dirac notation}, $|\cdot \rangle$ to represent the basis states of a single qubit in $\C^2$. For any $\bm{v}=[v_1,v_2,\cdots, v_n]\in \F_2^n$, we define $|\bm{v}\rangle=|v_1\rangle\otimes|v_2\rangle\otimes\cdots\otimes|v_n\rangle$, the standard basis vector in $\C^N$ with $1$ in the position indexed by $\bm{v}$ and $0$ elsewhere. 
    We write the Hermitian transpose of $|\bm{v}\rangle$ as $\langle \bm{v}|=|\bm{v}\rangle^\dagger$. 
    We may write an arbitrary $n$-qubit quantum state as $|\psi\rangle=\sum_{\bm{v}\in \F_2^n} \alpha_{\bm{v}} |\bm{v}\rangle \in \C^N$, where $\alpha_{\bm{v}}\in \C$ and $\sum_{\bm{v}\in\F_2^n}|\alpha_{\bm{v}}|^2=1$. The Pauli matrices act on a single qubit as
    % \begin{equation}
    $X\ket{0}=\ket{1},  X\ket{1}=\ket{0},  Z\ket{0}=\ket{0}, \text{ and }  Z\ket{1}=-\ket{1}.$
    % \end{equation}
    
    The symplectic inner product is $\langle [\bm{a},\bm{b}],[\bm{c},\bm{d}]\rangle_S=\bm{a}\bm{d}^T+\bm{b}\bm{c}^T \bmod 2$. Since $ X Z=- Z X$, we have %(see \cite{8437652})
    \begin{equation}
    E(\bm{a},\bm{b})E(\bm{c},\bm{d})=(-1)^{\langle [\bm{a},\bm{b}],[\bm{c},\bm{d}]\rangle_S}E(\bm{c},\bm{d})E(\bm{a},\bm{b}).
    \end{equation}

    \subsection{The Clifford Hierarchy}
    \label{subsec:prem_Cliff}
    
    The \textit{Clifford hierarchy} of unitary operators was introduced in \cite{gottesman1999demonstrating}. The first level of the hierarchy is defined to be the Pauli group $\mathcal{C}^{(1)}=\mathcal{HW}_N$. For $l\ge 2$, the levels $l$ are defined recursively as 
    % \begin{equation}\label{eqn:def_Cliff_hierarchy}
    % \mathcal{C}^{(l)}:=\{U\in \mathbb{U}_N: UE(\bm{a},\bm{b})U^\dagger\in \mathcal{C}^{(l-1)}, ~\forall~ E(\bm{a},\bm{b})\in \mathcal{HW}_N \},
    % \end{equation}
    \begin{equation}\label{eqn:def_Cliff_hierarchy}
    \mathcal{C}^{(l)}:=\{U\in \mathbb{U}_N: U \mathcal{HW}_N U^\dagger\subset \mathcal{C}^{(l-1)}\},
    \end{equation}
    where $\mathbb{U}_N$ is the group of $N\times N$ unitary matrices. The second level is the Clifford Group, $\mathcal{C}^{(2)}$, which can be generated (up to overall phases) using the ``elementary" unitaries \textit{Hadamard}, \textit{Phase}, and either of \textit{Controlled-NOT} (C$X$) or \textit{Controlled-$Z$} (C$Z$) defined respectively as
    \begin{equation}
    H\coloneqq \frac{1}{\sqrt{2}}\begin{bmatrix}
    1 & 1\\
    1 & -1
    \end{bmatrix} , 
    P\coloneqq\begin{bmatrix}
    1 & 0\\
    0 & \imath
    \end{bmatrix} , 
    % CNOT \coloneqq\left[\begin{array}{cc}
    % I_2 & 0\\
    % 0 & \sigma_X
    % \end{array} \right].
    \end{equation}
    \begin{align}
    \text{C}Z_{ab}  &\coloneqq \ket{0}\bra{0}_a \otimes (I_2)_b + \ket{1}\bra{1}_a \otimes Z_b,\\~
    \text{C}X_{a \rightarrow b}  &\coloneqq \ket{0}\bra{0}_a \otimes (I_2)_b + \ket{1}\bra{1}_a \otimes X_b.
    \end{align}
    
    Note that Clifford unitaries in combination with \emph{any} unitary from a higher level can be used to approximate any unitary operator arbitrarily well~\cite{boykin1999universal}. 
    Hence, they form a universal set for quantum computation. A widely used choice for the non-Clifford unitary is the $T$ gate in the third level defined by 
    \begin{equation}
    T:=\begin{bmatrix}
    1 & 0\\
    0 & e^{\frac{\imath\pi}{4}}
    \end{bmatrix}=
    % \sqrt{P}=
    Z^{\frac{1}{4}}\equiv 
    \begin{bmatrix}
    e^{-\frac{\imath\pi}{8}} & 0\\
    0 & e^{\frac{\imath\pi}{8}}
    \end{bmatrix}
    =e^{-\frac{\imath\pi}{8} Z}.
    \end{equation}
    
    \subsection{Stabilizer Codes}
    \label{subsec:prem_stab}
    We define a stabilizer group $\mathcal{S}$ to be a commutative subgroup of the Pauli group $\mathcal{HW}_N$, where every group element is Hermitian and no group element is $-I_N$. We say $\mathcal{S}$ has dimension $r$ if it can be generated by $r$ independent elements as $\mathcal{S}=\langle \nu_i E(\bm{c_i},\bm{d_i}): i=1,2,\dots, r \rangle$, where $\nu_i\in\{\pm1\}$ and $\bm{c_i},\bm{d_i}\in \F_2^n$. Since $\mathcal{S}$ is commutative, we must have $\langle [\bm{c_i},\bm{d_i}],[\bm{c_j},\bm{d_j}]\rangle_S=\bm{c_i}\bm{d_j}^T+\bm{d_i}\bm{c_j}^T=0\bmod 2$.
    
    Given a stabilizer group $\mathcal{S}$, the corresponding \textit{stabilizer code} is the fixed subspace $\mathcal{V}(\mathcal{S)}:=\{|\psi\rangle \in \C^N: g|\psi\rangle=|\psi\rangle \text{ for all } g\in \mathcal{S} \}$. 
    %spanned by all eigenvectors in the common eigenbasis of $\mathcal{S}$ that have eigenvalue $+1$. 
    We refer to the subspace $\mathcal{V}(\mathcal{S})$ as an $\left[\left[n,k,d\right]\right]$ stabilizer code because it encodes $k:=n-r$ \text{logical} qubits into $n$ \textit{physical} qubits. The minimum distance $d$ is defined to be the minimum weight of any operator in $\mathcal{N}_{\mathcal{HW}_N}\left(\mathcal{S}\right)\setminus \mathcal{S}$. Here, the weight of a Pauli operator is the number of qubits on which it acts non-trivially (i.e., as $ X,~Y$ or $ Z$), and $\mathcal{N}_{\mathcal{HW}_N}\left(\mathcal{S}\right)$ denotes the normalizer of $\mathcal{S}$ in $\mathcal{HW}_N$ defined by 
    % \begin{align}
    %     \mathcal{N}_{\mathcal{HW}_N}\left(\mathcal{S}\right) & \coloneqq \{\imath^\kappa E\left(\bm{a},\bm{b}\right)\in \mathcal{HW}_N: E\left(\bm{a},\bm{b}\right)E\left(\bm{c},\bm{d}\right)E\left(\bm{a},\bm{b}\right)=E\left(\bm{c}',\bm{d}' \right) \in \mathcal{S} \text{ for all } E\left(\bm{c},\bm{d}\right)\in \mathcal{S}, \kappa\in \{0,1,2,3\} \} \nonumber\\
    %   & = \{\imath^\kappa E\left(\bm{a},\bm{b}\right)\in \mathcal{HW}_N: E\left(\bm{a},\bm{b}\right)E\left(\bm{c},\bm{d}\right)E\left(\bm{a},\bm{b}\right)=E\left(\bm{c},\bm{d}\right) \text{ for all } E\left(\bm{c},\bm{d}\right)\in \mathcal{S}, \kappa\in \{0,1,2,3\} \}.
    % \end{align}
        \begin{align}
        \mathcal{N}_{\mathcal{HW}_N}\left(\mathcal{S}\right) & \coloneqq \{\imath^\kappa E\left(\bm{a},\bm{b}\right)\in \mathcal{HW}_N: \nonumber\\ &~~~~~~E\left(\bm{a},\bm{b}\right) \mathcal{S} E\left(\bm{a},\bm{b}\right)= \mathcal{S}, \kappa\in \Z_4 \} \nonumber\\
       & = \{\imath^\kappa E\left(\bm{a},\bm{b}\right)\in \mathcal{HW}_N:\nonumber\\ &~~~~~E\left(\bm{a},\bm{b}\right)E\left(\bm{c},\bm{d}\right)E\left(\bm{a},\bm{b}\right)=E\left(\bm{c},\bm{d}\right) \nonumber\\
       &~~~~\text{ for all } E\left(\bm{c},\bm{d}\right)\in \mathcal{S}, \kappa\in \Z_4 \}.
    \end{align}
    Note that the second equality defines the centralizer of $\mathcal{S}$ in $\mathcal{HW}_N$, and it follows from the first since Pauli matrices commute or anti-commute.
    
    For any Hermitian Pauli matrix $E\left(\bm{c},\bm{d}\right)$ and $\nu\in\{\pm 1\}$, the operator $\frac{I_N+\nu E\left(\bm{c},\bm{d}\right)}{2}$ projects onto the $\nu$-eigenspace of $E\left(\bm{c},\bm{d}\right)$. Thus, the projector onto the codespace $\mathcal{V}(\mathcal{S})$ of the stabilizer code defined by $\mathcal{S}=\langle \nu_i E\left(\bm{c_i},\bm{d_i}\right): i=1,2,\dots, r \rangle$ is 
    \begin{align}
    \Pi_{\mathcal{S}}
    &=\prod_{i=1}^{r}\frac{\left(I_N+\nu_iE\left(\bm{c_i},\bm{d_i}\right)\right)}{2} \nonumber\\
    &=\frac{1}{2^r}\sum_{j=1}^{2^r}\epsilon_jE\left(\bm{a_j},\bm{b_j}\right),
    \end{align}
    where $\epsilon_j\in \{\pm 1 \}$ is a character of the group $\mathcal{S}$, and is determined by the signs of the generators that produce $E(\bm{a_j},\bm{b_j})$: $\epsilon_jE\left(\bm{a_j},\bm{b_j}\right)=\prod_{t\in J\subset \{1,2,\dots,r\} } \nu_t E\left(\bm{c_t},\bm{d_t}\right)$ for a unique $J$.
    
    Let $\ket{\bm{\alpha}}_L$, $\bm{\alpha}\in \F_2^k$ be the protected logical state. We define the generating set $\{X^L_j, Z^L_j \in \mathcal{HW}_{2^k} : j = 1, \dots k=k_1-k_2\}$ for the logical Pauli operators by the actions 
	\begin{align}
	X_j^L \ket{\bm{\alpha}}_L = \ket{\bm{\alpha'}}_L, 
	\end{align}
	where
	\begin{align}\alpha'_i = 
	\left\{ \begin{array}{lc}
	\alpha_i,  &\text{ if } i \neq j, \\
	\alpha_i \oplus 1, & \text{ if } i = j,
	\end{array}\right.
	\end{align}
	and $Z_j^L \ket{\bm{\alpha}}_L = (-1)^{\alpha_j}\ket{\bm{\alpha}}_L.$
	Let $\bar{X}_j, \bar{Z}_j$ be the $n$-qubit operators which are physical representatives of $X_j^L, Z^L_j$ for $j= 1,\dots,k $. Then $\bar{X}_j, \bar{Z}_j$ commute with the stabilizer group $S$ and satisfy
	\begin{align}\label{eqn:logical_condition}
	\bar{X}_i \bar{Z}_j = 	
	\left\{ 
	\begin{array}{lc}
	\bar{Z}_j \bar{X}_i,  &\text{ if } i \neq j, \\
	-\bar{Z}_j \bar{X}_i, &\text{ if } i = j.
	\end{array}\right. 
	\end{align}
   \begin{remark}
   \normalfont
   \label{rem:signs_matter}
    A stabilizer code determines a resolution of the identity with the different subspaces fixed by different signings of the stabilizer generators. When we correct stochastic and independent Pauli errors, different signings of stabilizer generators lead to quantum codes with identical performance. 
	%The signs of stabilizers do not make difference when we correct stochastic and independent Pauli errors. 
	%One way to understand it is to partition the Clifford group into orbits by Pauli group. Then, changing the signs of stabilizers does not cross orbits.
	However, when we consider correlated errors such as the coherent errors (rotations of $Z$ axis for any angle $\theta$), the signs of stabilizers play an important role  \cite{coherent_noise,debroy2021optimizing}. %Hence, we need to work with the encoding map and logical Pauli operators of CSS codes with signs. We begin with an introductory example.
	\end{remark}
	\begin{example}[$3$-qubit bit flip code with negative signs]
	\label{examp1}
	    \normalfont
		Consider the stabilizer code defined by the group 
		$\mathcal{S}=\langle -Z_1Z_2, Z_2Z_3\rangle $,
		%$\mathcal{S}=\langle -E(000,110), E(000,011)\rangle $,
		which differs from the stabilizer group of the $3$-qubit bit flip code, 
		$\mathcal{S'}=\langle Z_1Z_2, Z_2Z_3\rangle $,
		%$\mathcal{S'}=\langle E(000,110), E(000,011)\rangle$, 
		just by the sign of $Z_1Z_2$. The encoding circuit of 
		%the $3$-qubit bit flip code with all positive signs 
		$\mathcal{V}(\mathcal{S'})$ consist of C$X_{1\to 2}$ and C$X_{1\to 3}$ gates, which maps $\ket{0}_L$ to $ \ket{000}$ and $\ket{1}_L$ to $ \ket{111}$. Since $XZX^\dagger = -Z$, the encoding circuit of $\mathcal{V}(\mathcal{S})$ has an extra $X$ gate on the first qubit, which has $\ket{\bar{0}} = \ket{100}$ and $\ket{\bar{1}} = \ket{011}$. Moreover, the physical representation of logical Pauli $X$ and $Z$ for $\mathcal{S}$ is
		$X_1X_2X_3$ and $Z_1$ 
		%$E(111,000)$ and $E(000,100)$ 
		respectively, i.e., $\bar{X} = X_1X_2X_3, ~\bar{Z} = { -}Z_1$.
		%The character vector for $X$-stabilziers is trivial while that for $Z$-stabilizers, $y=[1,0,0]$.
% 		\begin{minipage}{.6\textwidth}
% 			\begin{itemize}
% 				\vspace{10pt}
% 				\item Encoded States:	$\ket{\overline{0}} =  \ket{{100}},
% 				~ \ket{\overline{1}} = \ket{{011}} $
% 				\item Logical Pauli ops:	$\overline{X} = X_1X_2X_3,
% 				~\overline{Z} = { -}Z_1$
% 			\end{itemize}
% 		\end{minipage}
% 	\begin{minipage}{.4\textwidth}
% 			%\vspace{-10pt}
% 		\Qcircuit @C=1em @R=.7em {
% 			\lstick{\ket{\psi}} & \ctrl{1} &  \ctrl{2} & \gate{X} &\qw \\
% 			\lstick{\ket{0}} & \targ & \qw & \qw & \qw \\
% 			\lstick{\ket{0}} & \qw & \targ & \qw & \qw 
% 		}
% 	\end{minipage}
	\end{example}
    \subsection{CSS Codes}
    \label{subsec:prem_CSS}
    A \textit{CSS (Calderbank-Shor-Steane) code} is a particular type of stabilizer code with generators that can be separated into strictly $X$-type and strictly $Z$-type operators. Consider two classical binary codes $\mathcal{C}_1,\mathcal{C}_2$ such that $\mathcal{C}_2\subset \mathcal{C}_1$, and let $\mathcal{C}_1^\perp$, $\mathcal{C}_2^\perp$ denote the dual codes. Note that $\mathcal{C}_1^\perp\subset \mathcal{C}_2^\perp$. Suppose that $\mathcal{C}_2 = \langle \bm{c_1},\bm{c_2},\dots,\bm{c_{k_2}} \rangle$ is an $[n,k_2]$ code and $\mathcal{C}_1^\perp =\langle \bm{d_1},\bm{d_2}\dots,\bm{d_{n-k_1}}\rangle$ is an $[n,n-k_1]$ code. Then, the corresponding CSS code has the stabilizer group 
	\begin{align*}
	\mathcal{S} 
	&=\langle \nu_{(\bm{c_i},\bm{0})} E\left(\bm{c_i},\bm{0}\right), \nu_{(\bm{0},\bm{d_j})} E\left(\bm{0},\bm{d_j}\right)\rangle_{i=1;~j=1}^{i=k_2;~j=n-k_1} \nonumber\\
	&=\{\epsilon_{(\bm{a},\bm{0})} \epsilon_{(\bm{0},
	\bm{b})} E\left(\bm{a},\bm{0}\right)E\left(\bm{0},\bm{b}\right): \bm{a}\in \mathcal{C}_2, \bm{b}\in \mathcal{C}_1^\perp\}, %\imath^{-ab^T\bmod 4}
	\end{align*}
	where $\nu_{(\bm{c_i},\bm{0})},\nu_{(\bm{0},\bm{d_j})},\epsilon_{(\bm{a},\bm{0})},\epsilon_{(\bm{0},
	\bm{b})}  \in\{\pm 1 \}$.
    The CSS code projector can be written as the product:
	\begin{align}
	\Pi_{\mathcal{S}} 
	%= \prod_{i=1}^{k_2} \frac{(I_N+\nu_{(\bm{c_i},\bm{0})} E(\bm{c_i},\bm{0}))}{2} \prod_{j=1}^{n-k_1} \frac{(I_N + \nu_{(\bm{0},\bm{d_j})} E(\bm{0},\bm{d_j}))}{2} \eqqcolon 
	=\Pi_{\mathcal{S}_X}\Pi_{\mathcal{S}_Z},
	\end{align}
	where 
	\begin{align}
	\Pi_{\mathcal{S}_X} &\coloneqq \prod_{i=1}^{k_2} \frac{(I_N+\nu_{(\bm{c_i},\bm{0})} E(\bm{c_i},\bm{0}))}{2} \nonumber\\
	&= \frac{\sum_{\bm{a}\in \mathcal{C}_2} \epsilon_{(\bm{a},\bm{0})} E(\bm{a},\bm{0})}{|\mathcal{C}_2|},
	\end{align}
	and 
	\begin{align}
	\Pi_{\mathcal{S}_Z} &\coloneqq \prod_{j=1}^{n-k_1} \frac{(I_N + \nu_{(\bm{0},\bm{d_j})} E(\bm{0},\bm{d_j}))}{2} \nonumber \\
	&= \frac{\sum_{\bm{b}\in \mathcal{C}_1^\perp} \epsilon_{(\bm{0},
	\bm{b})}  E(\bm{0},\bm{b})}{|\mathcal{C}_1^\perp|}.
	\end{align}
	Each projector defines a resolution of the identity, and we focus on $\Pi_{\mathcal{S}_X}$ since we consider diagonal gates.  
	% Moreover, it defines three types of resolution of the identity, which helps to detect and correct Pauli $X$, $Z$, and $Y$ errors. Since our work focus on $Z$ errors, we state the one for $Z$ type below and the rest two types can be defined in similar ways.
	Note that any $n$-qubit Pauli $Z$ operator can be expressed as $E(\bm{0},\bm{b})E(\bm{0},\bm{\gamma})E(\bm{0},\bm{\mu})$ for a $Z$-stabilizer representation $\bm{b}\in\mathcal{C}_1^\perp$, a $Z$-logical representation $\bm{\gamma}\in \mathcal{C}_2^\perp/\mathcal{C}_1^\perp$, and a $X$-syndrome representation $\bm{\mu} \in\F_2^n/\mathcal{C}_2^\perp$. For $\bm{\mu}\in \F_2^n / \mathcal{C}_2^\perp$, we define 
	\begin{align} \label{eqn:Pi_SX_mu}
	\mathcal{S}_X(\bm{\mu})&\coloneqq\left\{ (-1)^{\bm{a}\bm{\mu}^T}\epsilon_{(\bm{a},0)} E(\bm{a},\bm{0}) : \bm{a}\in \mathcal{C}_2 \right\},\\
	\Pi_{\mathcal{S}_X(\bm{\mu})} &\coloneqq \frac{1}{|\mathcal{C}_2|}\sum_{\bm{a}\in \mathcal{C}_2} (-1)^{\bm{a}\bm{\mu}^T} \epsilon_{(\bm{a},\bm{0})} E(\bm{a},\bm{0}).
	\end{align}
	Then, we have 
		\begin{align}
		&\Pi_{\mathcal{S}_X(\bm{\mu})} \Pi_{\mathcal{S}_X(\bm{\mu}')} = 
		\left\{ \begin{array}{lc}
		\Pi_{\mathcal{S}_X(\bm{\mu})},  & \text{if}\ \bm{\mu} = \bm{\mu}', \\
		0, & \text{if}\ \bm{\mu} \neq \bm{\mu}',
		\end{array}
		\right.
		\\
	    & \text{and } \sum_{\bm{\mu} \in \F_2^n/\mathcal{C}_2^\perp} \Pi_{S_X(\bm{\mu})} = I_{2^n}.
		\end{align}

	If $\mathcal{C}_1$ and $\mathcal{C}_2^\perp$ can correct up to $t$ errors, then $S$ defines an $\left[\left[n, k_1-k_2, d\right]\right]$ CSS code with $d\ge 2t+1$, which we will represent as CSS($X,\mathcal{C}_2;Z,\mathcal{C}_1^\perp$). If $G_2$ and $G_1^\perp$ are the generator matrices for $\mathcal{C}_2$ and $\mathcal{C}_1^\perp$ respectively, then the $(n-k_1+k_2)\times (2n)$ matrix
	\begin{equation}
	G_{\mathcal{S}}=\left[\begin{array}{c|c}
	G_2 &  \\ \hline
	&  G_1^\perp
	\end{array} \right]
	\end{equation}
	generates $\mathcal{S}$. 
% 	The codespace defined by the stabilizer group $\mathcal{S}$ is $ \mathcal{V}(\mathcal{S}):=\{\ket{\psi}\in \C^N : g\ket{\psi} = \ket{\psi} \text{ for all } g\in \mathcal{S} \}$.

    \subsection{General Encoding Map for CSS codes}
    \label{subsec:CSS_general_Enc}
    Given an $\llbr n,k,d \rrbr $ CSS($X,\mathcal{C}_2;Z,\mathcal{C}_1^\perp$) code with all positive signs, let $G_{\mathcal{C}_1/\mathcal{C}_2}$ be the generator matrix for all coset representatives for $\mathcal{C}_2$ in $\mathcal{C}_1$ (note that the choice of coset representatives is not unique). The canonical encoding map $e:\mathbb{F}_2^k \to \mathcal{V}(\mathcal{S})$ is given by 
	$
	e(\ket{\bm{\alpha}}_L) 
	\coloneqq \frac{1}{\sqrt{|\mathcal{C}_2|}} \sum_{\bm{x}\in \mathcal{C}_2}\ket{\bm{\alpha}  G_{\mathcal{C}_1/\mathcal{C}_2} \oplus \bm{x}}$. 
	Note that the signs of stabilizers change the fixed subspace by changing the eigenspaces that enter into the intersection. Thus, the encoding map needs to include information about nontrivial signs. 
	\begin{center}
		\begin{tikzpicture}
		%\hspace{5pt}
		\node (B) at (1,1) {$\mathcal{B}\coloneqq\{\bm{z}\in \mathcal{C}_1^\perp|\epsilon_{\bm{z}}=1\}$};
		\node (C) at (1,2) {$\mathcal{C}_1^\perp$};
		\path[draw] (B) -- (C);
% 		\path[draw] (B) -- (C) node[midway,left] {$2$};
% 		\node (Cp) at (3.5,1) {$\mathcal{C}_1$};
% 		\node (Bp) at (3.5,2) {$\mathcal{B}^{\perp}$};
% 		\path[draw] (Cp) -- (Bp);
% 		\path[draw] (Cp) -- (Bp) node[midway,left] {$2$};
		% \node at (1.1,0.25) {Let .};
		
		\node (D) at (5,1) {$\mathcal{D}\coloneqq\{\bm{x}\in \mathcal{C}_2|\epsilon_{\bm{x}}=1\}$};
		\node (C2) at (5,2) {$\mathcal{C}_2$};
		\path[draw] (D) -- (C2);
% 		\path[draw] (D) -- (C2) node[midway,left] {$2$};
% 		\node (C2p) at (7.5,1) {$\mathcal{C}_2^\perp$};
% 		\node (Dp) at (7.5,2) {$\mathcal{D}^{\perp}$};
% 		\path[draw] (C2p) -- (Dp);
% 		\path[draw] (C2p) -- (Dp) node[midway,left] {$2$};
		
		\end{tikzpicture}
	\end{center}
	We capture sign information through character vectors 
% 	$\bm{y},~\bm{r}$
	$\bm{y}\in \F_2^n/\mathcal{C}_1,\bm{r} \in \F_2^n/\mathcal{C}_2^\perp$ (note that the choice of coset representatives is not unique) 
	defined for $Z$-stabilizers and $X$-stabilizers respectively by
	\begin{equation}
	\mathcal{B}=\mathcal{C}_1^\perp\cap \bm{y}^\perp,\text{equivalently, } \mathcal{B}^\perp = \langle \mathcal{C}_1,\bm{y}\rangle,
	\end{equation}
	and 
	\begin{equation}
	\mathcal{D} =\mathcal{C}_2 \cap \bm{r}^\perp,\text{equivalently, }D^\perp = \langle \mathcal{C}_2^\perp,\bm{r}\rangle.
	\end{equation}
    Then, for $ \epsilon_{(\bm{a},\bm{0})} \epsilon_{(\bm{0},\bm{b})} E\left(\bm{a},\bm{0}\right)E\left(\bm{0},\bm{b}\right) \in S$, we have $ \epsilon_{(\bm{a},\bm{0})} = (-1)^{\bm{a}\bm{r}^T}$ and $ \epsilon_{(\bm{0},\bm{b})} = (-1)^{\bm{b}\bm{y}^T}$. 
	In Example \ref{examp1}, we may choose the character vectors $\bm{r}=\bm{0}$ (character vector of $X$-stabilizers) and $\bm{y}=[1,0,0]$ (character vector of $Z$-stabilizers). 
	
	The generalized encoding map $g_e:\ket{\bm{\alpha}}_L\in \mathbb{F}_2^k \to \ket{\overline{\bm{\alpha}}}\in \mathcal{V}(\mathcal{S})$ is defined by
	\begin{equation} \label{eqn:gen_encode_map}
		\ket{\overline{\bm{\alpha}}}
		\coloneqq \frac{1}{\sqrt{|\mathcal{C}_2|}} \sum_{\bm{x}\in \mathcal{C}_2} (-1)^{\bm{x}\bm{r}^T}\ket{\bm{\alpha}  G_{\mathcal{C}_1/\mathcal{C}_2} \oplus \bm{x} \oplus \bm{y}}.
	\end{equation}
	To verify that the image of the general encoding map $g_e$ is in $\mathcal{V}(\mathcal{S})$, we show that for $\epsilon_{(\bm{a},\bm{0})} \epsilon_{(\bm{0},\bm{b})} E\left(\bm{a},\bm{0}\right)E\left(\bm{0},\bm{b}\right)  \in \mathcal{S}$ (that is $\bm{a}\in \mathcal{C}_2$, $ \epsilon_{(\bm{a},\bm{0})} = (-1)^{\bm{a}\bm{r}^T}$, $\bm{b}\in \mathcal{C}_1^\perp$, and $ \epsilon_{(\bm{0},\bm{b})} = (-1)^{\bm{b}\bm{y}^T}$), 
	\begin{widetext}
	\begin{align*}
	& \epsilon_{(\bm{a},\bm{0})} \epsilon_{(\bm{0},\bm{b})} E\left(\bm{a},\bm{0}\right)E\left(\bm{0},\bm{b}\right)  \ket{\overline{\bm{\alpha}}}\\
% 	&=  \epsilon_{(\bm{a},\bm{0})} \epsilon_{(\bm{0},\bm{b})} E\left(\bm{a},\bm{0}\right)E\left(\bm{0},\bm{b}\right) \frac{1}{\sqrt{|\mathcal{C}_2|}} \sum_{\bm{x}\in \mathcal{C}_2} (-1)^{\bm{x}\bm{r}^T}\ket{\bm{v}  G_{\mathcal{C}_1/\mathcal{C}_2} \oplus \bm{x} \oplus \bm{y}}\\
	&=\frac{1}{\sqrt{|\mathcal{C}_2|}}\sum_{\bm{x}\in \mathcal{C}_2} \Big(\epsilon_{(\bm{a},\bm{0})} (-1)^{\bm{x}\bm{r}^T} \epsilon_{(\bm{0},\bm{b})}(-1)^{\bm{b}(\bm{\alpha}G_{\mathcal{C}_1/\mathcal{C}_2}\oplus \bm{x} \oplus \bm{y})^T}\ket{\bm{\alpha}  G_{\mathcal{C}_1/\mathcal{C}_2} \oplus \bm{a} \oplus \bm{x}\oplus \bm{y}}\Big)\\
% 	& = \frac{1}{\sqrt{|\mathcal{C}_2|}}  \sum_{\bm{x}\in \mathcal{C}_2} (-1)^{\bm{a}\bm{r}^T} (-1)^{\bm{x}\bm{r}^T} (-1)^{\bm{b}\bm{y}^T} (-1)^{\bm{b}\bm{y}^T} \ket{\bm{v}  G_{\mathcal{C}_1/\mathcal{C}_2} \oplus \bm{a} \oplus \bm{x}\oplus \bm{y}}\\
	&=\frac{1}{\sqrt{|\mathcal{C}_2|}} \sum_{\bm{x}\in \mathcal{C}_2} (-1)^{(\bm{a}\oplus \bm{x})\bm{r}^T}  \ket{\bm{\alpha}  G_{\mathcal{C}_1/\mathcal{C}_2} \oplus \bm{a} \oplus \bm{x}\oplus \bm{y}}\\
	&=  \ket{\overline{\bm{\alpha}}}.\numberthis
	\end{align*}
	\end{widetext}
	
    \subsection{General Logical Pauli Operators for CSS codes}
    \label{subsec:general_logical_op_CSS}
    Given the choice of $G_{\mathcal{C}_1/\mathcal{C}_2}$, there exists a unique set of vectors $\{\bm{\gamma_1},\cdots,\bm{\gamma_k} \in \mathcal{C}_2^\perp: G_{\mathcal{C}_1/\mathcal{C}_2}\bm{\gamma_i} = \bm{e_i} \text{ for all } i = 1,\dots,k \}$, where $\{\bm{e_i}\}_{i =1,\dots,k}$ is the standard basis of $\F_2^k$.
	If $\bm{\gamma_i}$ is the $i$-the row of generator matrix $G_{\mathcal{C}_2^\perp / \mathcal{C}_1^\perp}$, then 
	\begin{equation}\label{eqn:log_req}
	G_{\mathcal{C}_1/\mathcal{C}_2} G_{\mathcal{C}_2^\perp /\mathcal{C}_1^\perp }^T = I_{k}.
	\end{equation}
	Assume we have 
	\begin{align}
	G_{\mathcal{C}_1/\mathcal{C}_2} =
	\left[
	\begin{array}{c}
	 \bm{w_1} \\
	 \bm{w_2}  \\
	 \vdots   \\
	 \bm{w_k} 
	\end{array}
	\right],~
	G_{\mathcal{C}_2^\perp / \mathcal{C}_1^\perp} =
	\left[
	\begin{array}{c}
	\bm{\gamma_1} \\
    \bm{\gamma_2} \\
	 \vdots    \\
	\bm{\gamma_k} 
	\end{array}
	\right].
	\end{align}
	Thus, we have for $i=1,\dots, k$
	\begin{align*}
	&E(\bm{w_i},\bm{0}) \ket{\overline{\bm{\alpha}}} \\
% 	&= E(\bm{w_i},\bm{0}) \frac{1}{\sqrt{|\mathcal{C}_2|}} \sum_{\bm{x}\in \mathcal{C}_2} (-1)^{\bm{x}\bm{r}^T}\ket{\bm{v}  G_{\mathcal{C}_1/\mathcal{C}_2} \oplus \bm{x} \oplus \bm{y}} \\
	&= \frac{1}{\sqrt{|\mathcal{C}_2|}} \sum_{\bm{x}\in \mathcal{C}_2} (-1)^{\bm{x}\bm{r}^T}\ket{\bm{\alpha}  G_{\mathcal{C}_1/\mathcal{C}_2} \oplus \bm{w_i} \oplus \bm{x} \oplus \bm{y}} \\
	&= \frac{1}{\sqrt{|\mathcal{C}_2|}} \sum_{\bm{x}\in \mathcal{C}_2} (-1)^{\bm{x}\bm{r}^T}\ket{ (X_i^L \bm{\alpha}) G_{\mathcal{C}_1/\mathcal{C}_2} \oplus \bm{x} \oplus \bm{y}}\\
	&= \bar{X}_i \ket{\overline{\bm{\alpha}}},\numberthis
	\end{align*}
	and
	\begin{align*}
	&(-1)^{\bm{\gamma_i}\bm{y}^T} E(\bm{0},\bm{\gamma_i}) \ket{\overline{\bm{\alpha}}} \\
% 	&= (-1)^{\bm{\gamma_i}\bm{y}^T}E(\bm{0},\bm{\gamma_i}) \frac{1}{\sqrt{|\mathcal{C}_2|}} \sum_{\bm{x}\in \mathcal{C}_2} (-1)^{\bm{x}\bm{r}^T}\ket{\bm{v}  G_{\mathcal{C}_1/\mathcal{C}_2} \oplus \bm{x} \oplus \bm{y}} \\
	&=  \frac{1}{\sqrt{|\mathcal{C}_2|}} \Big(\sum_{\bm{x}\in \mathcal{C}_2} (-1)^{{\bm{x}\bm{r}^T} \oplus {\bm{\gamma_i}\bm{y}^T} \oplus {\bm{\gamma_i} (\bm{\alpha}G_{\mathcal{C}_1/\mathcal{C}_2}\oplus \bm{x} \oplus \bm{y})^T} } \\
	&~~~~~~~~~~~~~~~~~~~ \ket{\bm{\alpha}  G_{\mathcal{C}_1/\mathcal{C}_2} \oplus \bm{x} \oplus \bm{y}}\Big)\\
% 	&=  \frac{1}{\sqrt{|\mathcal{C}_2|}} \sum_{\bm{x}\in \mathcal{C}_2} (-1)^{\bm{x}\bm{r}^T} (-1)^{\bm{\gamma_i} (\bm{v}G_{\mathcal{C}_1/\mathcal{C}_2})^T}  \ket{\bm{v}  G_{\mathcal{C}_1/\mathcal{C}_2} \oplus \bm{x} \oplus \bm{y}}  \\%~(\text{as } c\in C_2, w_i \in C_2^\perp)\\
	&= \frac{1}{\sqrt{|\mathcal{C}_2|}} \sum_{\bm{x}\in \mathcal{C}_2} (-1)^{\bm{x}\bm{r}^T} (-1)^{\bm{\alpha}\bm{e_i}^T}  \ket{\bm{v}  G_{\mathcal{C}_1/\mathcal{C}_2} \oplus \bm{x} \oplus \bm{y}} \\
	&= \bar{Z}_i \ket{\overline{\bm{\alpha}}},\numberthis
	\end{align*}
	where the second to last step follows from \eqref{eqn:log_req}. 
	% Note that $e_i$ has the length $k$ and the only support on the $i$-th entry.
	Thus we can choose 
	\begin{equation}\label{eqn:css_log_paulis}
		\bar{X}_{i} = E(\bm{w_i},\bm{0}) 
		\text{ and }
		\bar{Z}_{i} %= (-1)^{\bm{\gamma_i}\bm{y}^T}E(\bm{0},\bm{\gamma_i}) 
		= \epsilon_{(\bm{0},\bm{\gamma_i})}E(\bm{0},\bm{\gamma_i}),
	\end{equation}
	%for $i = 1,\dots, k$
	where $\bm{w_i}, \bm{\gamma_i}$ are the $i$-th rows of the above coset generator matrices $G_{\mathcal{C}_1/\mathcal{C}_2}$, $G_{\mathcal{C}_2^\perp / \mathcal{C}_1^\perp}$ respectively. 
	\begin{remark}
	\normalfont
	\label{rem:signs_matter_logical_op}
    Applying appropriate Pauli operators takes care of different signs in the stabilizer group and changes the sign of logical Pauli operators. Although the sign for a single logical Pauli operator is not observable, a general logical operator is a linear combination of logical Pauli operators, which may bring the global sign into some local phase. 
    \end{remark}
    \addtocounter{example}{-2}
	\begin{example}[The basis state and logical Pauli operators of the $\llbr 4,2,2 \rrbr $ code] 
	\label{examp2}
	\normalfont
	    Consider the CSS($X,\mathcal{C}_2 ;Z,\mathcal{C}_1^\perp$) code with $\mathcal{C}_2 = \mathcal{C}_1^\perp = \{\bm{0},\bm{1}\}$.
	   % \langle1111\rangle$. 
	    We may choose the generator matrices of $\mathcal{C}_1/\mathcal{C}_2$ and $\mathcal{C}_2^\perp/\mathcal{C}_1^\perp$ as
		\begin{align}
		\setlength\aboverulesep{0pt}\setlength\belowrulesep{0pt}
		\setlength\cmidrulewidth{0.5pt}
		G_{\mathcal{C}_1/\mathcal{C}_2 } = 
		\begin{bmatrix}
		0 & 1 & 1 & 0    \\
		0 & 0 & 1 & 1   
		\end{bmatrix},~
		G_{\mathcal{C}_2^\perp/\mathcal{C}_1^\perp } = 
		\begin{bmatrix}
		0 & 0 & 1 & 1   \\
		0 & 1 & 1 & 0    
		\end{bmatrix}.
		\end{align}
	    The encoded basis states and logical Pauli operators for two choices of the signs are given below.
	   % \begin{enumerate}
	    If $\mathcal{S}=\langle X^{\otimes 4}, Z^{\otimes 4} \rangle$ ($\bm{r} =\bm{y}=\bm{0}$), we have 
	        \begin{align*}
	            \ket{\overline{00}} &= \frac{1}{\sqrt{2}} \left( \ket{0000} +\ket{1111} \right),\\
	            \ket{\overline{01}} &=  \frac{1}{\sqrt{2}} \left( \ket{0011} +\ket{1100} \right),\\
	            \ket{\overline{10}} &= \frac{1}{\sqrt{2}} \left( \ket{0110} +\ket{1001} \right),\\
	            \ket{\overline{11}} &= \frac{1}{\sqrt{2}} \left( \ket{0101} +\ket{1010} \right),\\
	            \bar{X}_1 &= X_2X_3,~ \bar{X}_2 = X_3X_4,~  \\
	            \bar{Z}_1 &= Z_3Z_4,~~ \bar{Z}_2 = Z_2Z_3.
	        \end{align*}
	    When $\mathcal{S'}=\langle X^{\otimes 4}, -Z^{\otimes 4} \rangle$ ($\bm{r}'=\bm{0}$, $\bm{y}'= [0,0,0,1]$), we have
	        \begin{align*}
	            \ket{\overline{00}} &= \frac{1}{\sqrt{2}} \left( \ket{0001} +\ket{1110} \right),\\
	            \ket{\overline{01}} &=  \frac{1}{\sqrt{2}} \left( \ket{0010} +\ket{1101} \right),\\
	            \ket{\overline{10}} &= \frac{1}{\sqrt{2}} \left( \ket{0111} +\ket{1000} \right),\\
	            \ket{\overline{11}} &= \frac{1}{\sqrt{2}} \left( \ket{0100} +\ket{1011} \right),\\
	            \bar{X}_1 &= X_2X_3,~ \bar{X}_2 = X_3X_4, \\
	            \bar{Z}_1 &= -Z_3Z_4, ~~\bar{Z}_2 = Z_2Z_3.
	        \end{align*}
	   % \end{enumerate} 
	\end{example}

    \subsection{Quantum Channels} %and Density Operators
    The quantum states defined in Section \ref{subsec:prem_paulis} are called \emph{pure states}.
    %since they describe a certain state with probability 1. 
    When a system contains multiple pure states $\ket{\psi_x}$ with probabilities $p_x$, the ensemble $\{p_x,\ket{\psi_x}\}$, is described by a \emph{density operator} $\rho$ given by
    \begin{equation}\label{eqn:def_ds}
        \rho \coloneqq \sum_{x}p_x\ket{\psi_x}\bra{\psi_x} \in \C^{N\times N}.
    \end{equation}
    Every density operator is Hermitian, positive semi-definite, with unit trace. Conversely, any operator with these three properties can be written in the form \eqref{eqn:def_ds}. Every ensemble determines a unique density operator but a density operator can describe different ensembles.
    % could introduce the purity Tr(\rho^2) here but our paper does not use this definition
    
    Suppose we measure the density operator $\rho$ with a finite set of projectors $\{\Pi_j\}_j$ forming a resolution of the identity.
     If the initial state in the ensemble is $\ket{\psi_x}$, then we observe the outcome $j$ with probability $p(j|x) = \bra{\psi_x}\Pi_j\ket{\psi_x} = \mathrm{Tr}(\Pi_j\ket{\psi_x}\bra{\psi_x})$ and obtain the reduced state
    $\frac{\Pi_j\ket{\psi_x}}{\sqrt{p(j|x)}}$. From the perspective of density operators, we observe the outcome $j$ with probability $p_j = \sum_{x}p_xp(j|x) = \mathrm{Tr}(\Pi_j\rho)$ and the density operator evolves to be $\frac{\Pi_j\rho\Pi_j}{p_j}$. Thus, after measurement, we have a ensemble of ensembles described by a new density operator $\rho'$ given by \cite{wilde2013quantum}
    \begin{equation}\label{eqn:do_measure_evo}
        \rho' = \sum_{j}p_j\frac{\Pi_j\rho\Pi_j}{p_j} = \sum_j \Pi_j\rho\Pi_j.
    \end{equation}
    
     A quantum channel is linear, completely-positive, and trace-preserving, and can be characterized by a Kraus representation~\cite{nielsen2011quantum, wilde2013quantum}. A map $\Phi: \mathcal{H} \to \mathcal{G}$ is linear, completely-positive, and trace-preserving if and only if there exists a finite set of operators $\{B_k\}_{k}$ (from $\mathcal{H}$ to $\mathcal{G}$) such that for any $\rho\in \mathcal{H}$
     \begin{equation}
         \Phi(\rho) = \sum_{k} B_k\rho B_k^\dagger.
     \end{equation}
     The operators $\{B_k\}_{k}$ are called \emph{Kraus operators} and satisfy
     \begin{equation}
         \sum_k B_k^\dagger B_k = I_{2^{\dim(\mathcal{H})}} 
    \end{equation}
    and
    \begin{equation}
     |\{B_k\}_{k}| \le \dim(\mathcal{H})\dim(\mathcal{G}).
     \end{equation}
     Note that the Kraus representation of a quantum channel is not unique.

	\section{Generator Coefficients}
	\label{sec:intro_gcs}
	Starting from the general encoding map and logical Pauli operators of CSS codes introduced in Section \ref{subsec:CSS_general_Enc}, we study gates interacting with these codes. We consider quantum gates for which the Pauli expansion consists only of tensor products of Pauli $Z$'s (or Pauli $X$'s). We partition $\F_2^n$ into cosets of the $Z$-stabilizers (or $X$-stabilizers), and define generator coefficients that take advantage of the structure of stabilizer group. The framework of generator coefficients provides insight into the average logical channel, the necessary and sufficient conditions for a CSS code to be invariant under a particular gate, and the induced logical operator. We extend the framework of generator coefficients to general stabilizer codes in Appendix \ref{sec:stab_gcf}.
	
	Consider a $2^n \times 2^n$ unitary matrix (quantum gate) $U_Z=\sum_{\bm{v}\in \F_2^n}f(\bm{v})E(\bm{0},\bm{v})$, where $f(\bm{v})\in \C$. Since
	\begin{align*}
	I&=U_ZU_Z^\dagger\\
	&=\left( \sum_{\bm{v}\in \F_2^n}f(\bm{v})E(\bm{0},\bm{v})\right)\left( \sum_{\bm{v'}\in \F_2^n}\overline{f(\bm{v'})}E(\bm{0},\bm{v'})\right)\\
	&=\sum_{\bm{w}\in \F_2^n}\left(  \sum_{\bm{v}\in \F_2^n}f(\bm{v})\overline{f(\bm{v}\oplus \bm{w})}\right)E(0,\bm{w}),\numberthis
	\end{align*}
	we have
	\begin{equation}\label{eqn:unitarycoefficients}
	\sum_{\bm{v}\in \F_2^n}f(\bm{v})\overline{f(\bm{v}\oplus \bm{w})}
	=\left\{\begin{array}{lc}
	1, &  \text{ if } \bm{w}=\bm{0},\\
	0, &  \text{ if } \bm{w}\neq \bm{0}.
	\end{array} \right.
	\end{equation}
	We define the generator coefficients for $U_Z$ acting on a given CSS code as follows. 
	\begin{definition}[Generator Coefficients for $U_Z$]
	\label{def:gcs}
	Let CSS($X,\mathcal{C}_2;Z,\mathcal{C}_1^\perp$) be an $\left[\left[n, k_1-k_2, d\right]\right]$ stabilizer code defined by the stabilizer group $\mathcal{S}=\{\epsilon_{(\bm{a},\bm{0})} \epsilon_{(\bm{0},\bm{b})} E\left(\bm{a},\bm{0}\right)E\left(\bm{0},\bm{b}\right): \bm{a}\in \mathcal{C}_2, \bm{b}\in \mathcal{C}_1^\perp\}$ and the character vector $\bm{y}\in \F_2^n/\mathcal{C}_1$ for $Z$-stabilizers. 
		Let $\bm{\mu} \in \F_2^n / \mathcal{C}_2^\perp$ be any  $X$-syndrome and $\bm{\gamma} \in \mathcal{C}_2^\perp / \mathcal{C}_1^\perp$ be any $Z$-logical. Then, for any pair $\bm{\mu}$, $\bm{\gamma}$, we define the generator coefficient $A_{\bm{\mu},\bm{\gamma}}$ corresponding to the diagonal unitary gate $U_Z=\sum_{\bm{v}\in \F_2^n}f(\bm{v})E(\bm{0},\bm{v})$ by
		\begin{equation}\label{eqn:def_GC_UZ}
	    A_{\bm{\mu},\bm{\gamma}}\coloneqq \sum_{\bm{z}\in \mathcal{C}_1^\perp+\bm{\mu}+\bm{\gamma}}\epsilon_{(\bm{0},\bm{z})}f(\bm{z}),
	    \end{equation}
	    where $\epsilon_{(\bm{0},\bm{z})} = (-1)^{\bm{z}\bm{y}^T}$.
	\end{definition}
	Note that given a CSS code with not all positive signs, the character vector $\bm{y}$ is unique up to an element of $\mathcal{C}_1$. A different choice of the coset representatives of $\mathcal{C}_1$ in $\F_2^n$ only changes the signs of $A_{\bm{\mu},\bm{\gamma}}$, and leads to a global phase in the logical quantum channel induced by $U_Z$, which is given in Section \ref{sec:avg_log_chl}. % For ease of calculation, we may choose the unique $\bm{y}\in \F_2^n$ such that $\bm{y}(\bm{\mu}\oplus\bm{\gamma})^T = 0\bmod 2$ for all .
	
	By partitioning $\F_2^n$ into cosets of $\mathcal{C}_1^\perp$, we gain insight into the interaction of syndromes and logicals. The code projector is $\Pi_{\mathcal{S}} = \Pi_{\mathcal{S}_X}\Pi_{\mathcal{S}_Z}$, and we have 
	%To illustrate the reason for defining the generator coefficient in the above form, we first write the code projector as a product . Since $\Pi_{S_Z}$ commutes with the $Z$-rotation $R_Z(\theta)$, we have
	\begin{widetext}
	
	\begin{align*}
	\Pi_{\mathcal{S}_Z}U_Z 
	&= \frac{1}{2^{n-k_1}} \sum_{\bm{b}\in \mathcal{C}_1^\perp} \epsilon_{(\bm{0},\bm{b})} E(\bm{0},\bm{b}) \sum_{\bm{v}\in \F_2^n}f(\bm{v})E(\bm{0},\bm{v})
	= \frac{1}{2^{n-k_1}} \sum_{\bm{v}\in \F_2^n}f(\bm{v})\sum_{\bm{b}\in \mathcal{C}_1^\perp} \epsilon_{(\bm{0},\bm{b})} E(\bm{0},\bm{b}\oplus\bm{v})\\
	&= \frac{1}{2^{n-k_1}} \sum_{\bm{v}\in \F_2^n}\epsilon_{(\bm{0},\bm{v})}f(\bm{v})\sum_{\bm{u}\in \mathcal{C}_1^\perp+\bm{v}} \epsilon_{(\bm{0},\bm{u})} E(\bm{0},\bm{u})
	= \frac{1}{2^{n-k_1}} 
	\sum_{\bm{\mu}}
	\sum_{\bm{\gamma}}A_{\bm{\mu},\bm{\gamma}}
% 	\underbrace{\left(\sum_{\bm{v}\in \mathcal{C}_1^\perp+\bm{\mu}+\bm{\gamma}}\epsilon_{(\bm{0},\bm{v})}f(\bm{v})\right)}_{A_{\bm{\mu},\bm{\gamma}}}
	\sum_{\bm{u}\in \mathcal{C}_1^\perp+\bm{\mu}+\bm{\gamma}} \epsilon_{(\bm{0},\bm{u})} E(\bm{0},\bm{u}).\numberthis	\label{eqn:Pi_SZ_U_Z}
%	&= \frac{1}{2^{n-k_1}} \sum_{\bm{\mu}\in \F_2^n/\mathcal{C}_2^\perp}\sum_{\bm{\gamma}\in \mathcal{C}_2^\perp/\mathcal{C}_1^\perp}A_{\bm{\mu},\bm{\gamma}} \sum_{\bm{u}\in \mathcal{C}_1^\perp+\bm{\mu}+\bm{\gamma}} \epsilon_{(\bm{0},\bm{u})} E(\bm{0},\bm{u}),
	\end{align*}
	\end{widetext}
% 	which provides an insight into the role of syndromes and logicals by partition $\F_2^n$ into cosets of $\mathcal{C}_1^\perp$.
    In the above summations, $\bm{\mu}\in \F_2^n/\mathcal{C}_2^\perp$ and $\bm{\gamma}\in \mathcal{C}_2^\perp/\mathcal{C}_1^\perp$, and $A_{\bm{\mu},\bm{\gamma}}$ is given by \eqref{eqn:def_GC_UZ}.
	We now study the generator coefficients associated with two different types of quantum gate $U_Z$.
	
	\subsection{Transversal $Z$-Rotations $R_Z(\theta)$}
	\label{subsec:GCs_R_Z}
	There are two reasons to study how $R_Z(\theta)\coloneqq \left(\exp\left(-\imath\frac{\theta}{2}Z \right)\right)^{\otimes n}$ 
	$=\left(  \cos\frac{\theta}{2} I-\imath \sin\frac{\theta}{2} Z \right)^{\otimes n}$ acts on the states within a quantum error-correcting code. The first is that when $\theta$ is not a multiple of $\frac{\pi}{2}$, $R_Z(\theta)$ may realize a non-Clifford logical gate, and the second is that coherent noise can be modeled as $\{R_Z(\theta)\}_{\theta\in (0,2\pi) }$. %Thus, it is worthwhile to study how $R_Z(\theta)$ acts on the quantum states inside a quantum-error-correction code. %We now introduce a useful tool called \emph{generator coefficients} to analyze the logical channel of $R_Z(\theta)$ acting on a CSS code. The generator coefficient formalism can be also applied to other gates if we know their Pauli expansion.
	The Pauli expansion of $R_Z(\theta)$ is
	\begin{equation}%R_Z(\theta)=	
	\sum_{\bm{v}\in \mathbb{F}_2^n} \left( \cos\frac{\theta}{2}\right)^{n-w_H(\bm{v})}\left(-\imath \sin\frac{\theta}{2} \right)^{w_H(\bm{v})}E(\bm{0},\bm{v}).
	\end{equation}
	As $f(\bm{v})=\left( \cos\frac{\theta}{2}\right)^{n-w_H(\bm{v})}\left(-\imath \sin\frac{\theta}{2} \right)^{w_H(\bm{v})}$, we substitute it in \eqref{eqn:def_GC_UZ}, and obtain the generator coefficients of $R_Z(\theta)$, 
	\begin{align}\label{eqn:def_gc}
		&A_{\bm{\mu},\bm{\gamma}}(\theta) \coloneqq \nonumber \\
		&\sum_{\bm{z}\in \mathcal{C}_1^\perp + \bm{\mu} + \bm{\gamma}} \epsilon_{(\bm{0},\bm{z})} \left(\cos\frac{\theta}{2}\right)^{n-w_H(\bm{z})}\left(-\imath \sin\frac{\theta}{2}\right)^{w_H(\bm{z})}.
	\end{align}
	We now compute the generator coefficients for the $\llbr 7,1,3 \rrbr $ Steane code.
	\addtocounter{example}{-2}
	\begin{example}[Generator Coefficients for $R_Z(\theta)$ applied to the $\llbr 7,1,3 \rrbr $ Steane code]
	\label{examp3}
	\normalfont
	
	    The Steane code is a perfect CSS($X,\mathcal{C}_2;Z,\mathcal{C}_1^\perp$)  code with all positive signs and generator matrix
	    	\begin{align}
		\setlength\aboverulesep{0pt}\setlength\belowrulesep{0pt}
		\setlength\cmidrulewidth{0.5pt}
		G_{\mathcal{S}} = 
		\left[
		\begin{array}{c|c}
		H &  \\
		\hline
		 & H\\
		\end{array}
		\right],
		\end{align}
		where $H$ is the parity-check matrix of the Hamming code:
		\begin{align}
		\setlength\aboverulesep{0pt}\setlength\belowrulesep{0pt}
		\setlength\cmidrulewidth{0.5pt}
		H = 
		\left[
		\begin{array}{ccccccc}
		1 & 1 & 1 & 1 & 0 & 0 & 0    \\
		1 & 1 & 0 & 0 & 1 & 1 & 0   \\
		1 & 0 & 1 & 0 & 1 & 0 & 1   \\
		\end{array}
		\right].
		\end{align}
		Then, we have $C_1/C_2 = C_2^\perp / C_1^\perp = \{\bm{0}, \bm{1} \}$, where $\bm{0}, \bm{1}$ are the vectors of all ones and all zeros respectively. If we compute the generator coefficients directly from \eqref{eqn:def_gc}, then we need the weight enumerators of all cosets of $\mathcal{C}_1^\perp$. We may simplify these calculations using the MacWilliams Identities. Consider for example the case $\bm{\mu} = \bm{0}$ and $\bm{\gamma} = \bm{1}$, where we may write
		%If we directly follow from \eqref{eqn:def_gc} to compute the generator coefficients, we need the weight enumerators for all the cosets of $\mathcal{C}_1^\perp$. To avoid this, we may use the Macwilliams Identities to make the request to compute a linear character on $C_1^\perp$. Let's take the simplification of the generator coefficients correspond to the trivial syndrome $\bm{\mu} = \bm{0}$ and nontrivial logicals $\bm{\gamma} = \bm{1}$ as an example. We first write it as a difference of $P_{\theta}[\mathcal{C}]$ for some linear codes $\mathcal{C}\subset \F_2^7$, where $P_{\theta}[\mathcal{C}]$ is defined in \eqref{eqn:key_sub_Mac}.
		\begin{align} \label{eqn:7_1_3_ex}
		&A_{\bm{0},\bm{1}}(\theta)\nonumber \\
		&= \sum_{\bm{z}\in \mathcal{C}_1^\perp +\bm{1}}  \left(\cos\frac{\theta}{2}\right)^{7-w_H(\bm{z})}\left(-\imath \sin\frac{\theta}{2}\right)^{w_H(\bm{z})} \nonumber \\
		&= P_{\theta}[\langle \mathcal{C}_1^\perp,\bm{1}\rangle] - P_{\theta}[\mathcal{C}_1^\perp],
        % 		& = \sum_{\bm{z}\in \langle \mathcal{C}_1^\perp,\bm{1}\rangle }  \left(\cos\frac{\pi}{8}\right)^{7-w_H(\bm{z})}\left(-\imath \sin\frac{\pi}{8}\right)^{w_H(\bm{z})} - \sum_{\bm{z}\in \mathcal{C}_1^\perp }  \left(\cos\frac{\pi}{8}\right)^{7-w_H(\bm{z})}\left(-\imath \sin\frac{\pi}{8}\right)^{w_H(\bm{z})} \\
		\end{align}
		where $P_{\theta}[\mathcal{C}]$ is defined in \eqref{eqn:key_sub_Mac}.
		We apply the MacWilliams Identities to $P_{\theta}[ \mathcal{C}_1^\perp]$ to obtain
		\begin{align*}
		    P_{\theta}[ \mathcal{C}_1^\perp] 
		    & = \frac{1}{|\mathcal{C}_1|}P_{\mathcal{C}_1}\left(\cos\frac{\theta}{2} - \imath \sin \frac{\theta}{2}, \cos\frac{\theta}{2} + \imath \sin \frac{\theta}{2}\right)\\
		  %  & = \frac{1}{|\mathcal{C}_1|}\sum_{\bm{z}\in \mathcal{C}_1 }\left(\cos\frac{\theta}{2} - \imath \sin \frac{\theta}{2}\right)^{n-w_H(\bm{z})}\left( \cos\frac{\theta}{2} + \imath \sin \frac{\theta}{2}\right)^{w_H(\bm{z})}\\
		  %  & = \frac{1}{|\mathcal{C}_1|}\sum_{\bm{z}\in \mathcal{C}_1 }\left(\cos\frac{\theta}{2} - \imath \sin \frac{\theta}{2}\right)^{n-2w_H(\bm{z})}\\
		    &= \frac{1}{|\mathcal{C}_1|}\sum_{\bm{z}\in \mathcal{C}_1 }\left( e^{-\imath\frac{\theta}{2}}\right)^{n-2w_H(\bm{z})}\numberthis \label{eqn:7_1_3_C_1}.
		\end{align*}
		We simplify the term $P[\langle \mathcal{C}_1^\perp,\bm{1}\rangle]$ in the same way,
		\begin{align}
		  P_{\theta}[\langle \mathcal{C}_1^\perp,\bm{1}\rangle] &= \frac{1}{|\langle \mathcal{C}_1^\perp,\bm{1}\rangle|}\sum_{\bm{z}\in \langle \mathcal{C}_1^\perp,\bm{1}\rangle^\perp } \left(e^{-\imath\frac{\theta}{2}}\right)^{n-2w_H(\bm{z})} \nonumber\\
		  &=\frac{2}{|\mathcal{C}_1|}\sum_{\bm{z}\in  \mathcal{C}_1 \cap \bm{1}^\perp } \left(e^{-\imath\frac{\theta}{2}}\right)^{n-2w_H(\bm{z})}.\label{eqn:7_1_3_C_1_plus_1}
		\end{align}
		It follows from \eqref{eqn:7_1_3_ex}, \eqref{eqn:7_1_3_C_1}, and \eqref{eqn:7_1_3_C_1_plus_1} that
		%$P[\langle \mathcal{C}_1^\perp,\bm{1}\rangle]$ 
		\begin{align}
		A_{\bm{0},\bm{1}}(\theta) 
		& = \frac{1}{|\mathcal{C}_1|}\sum_{\bm{z}\in \mathcal{C}_1} (-1)^{\bm{1} \cdot \bm{z}^T} (e^{-\imath\frac{\theta}{2}})^{7-2w_H(\bm{z})} \label{eqn:7_1_3_fs} \\
% 		& = \frac{1}{16} \left( e^{-\imath \frac{7\theta}{2}} - e^{\imath \frac{7\theta}{2}} - 7e^{-\imath \frac{\theta}{2}} + 7e^{\imath \frac{\theta}{2}}\right) \label{eqn:ex1_simp_snd_last}\\
		& = \frac{1}{8}\left(-\imath \sin \frac{7\theta}{2} +7\imath \sin \frac{\theta}{2}\right),\label{eqn:ex1_simp}
		\end{align}
		where \eqref{eqn:ex1_simp} is obtained from \eqref{eqn:7_1_3_fs} by substituting in the weight enumerator of $\mathcal{C}_1$
		\begin{align*}
		P_{C_1}(x,y) = x^7 + 7x^4y^3 + 7x^3y^4 + y^{7}.
		\end{align*}
		We compute all the generator coefficients for the Steane code in Table \ref{tab:7_1_3_gcs}. We return to this data in Section \ref{subsec:Kraus} to provide more insight into the logical channel determined by $R_Z(\theta)$, and in Section \ref{subsec:prob} to calculate the probabilities of observing different syndromes. 
		\begin{table}[h!]
		    \centering
		   \caption{Generator coefficients $A_{\bm{\mu},\bm{\gamma}}(\theta)$ for $R_Z(\theta)$ applied to the Steane code. Each column corresponds to a $Z$-logical. The first row corresponds to the trivial $X$-syndromes, and second row  represents the seven non-trivial syndromes (they have equivalent behaviour due to symmetry).}
		    \renewcommand{\arraystretch}{1.2} 
		    \begin{tabular}{|c|c|c|}
		       \hline
		        $\bm{\mu}$ & $\bm{\gamma}=\bm{0}$  & $\bm{\gamma}=\bm{1}$ \\
		        \hline
		          $=\bm{0}$ & $\frac{1}{8}\left(\cos \frac{7\theta}{2} + 7 \cos \frac{\theta}{2}\right)$ & 
		          $\frac{\imath}{8}\left(7 \sin \frac{\theta}{2}- \sin \frac{7\theta}{2}\right)$\\
		          \hline
		        $\neq \bm{0}$ & $-\frac{\imath}{8}\left(\sin \frac{7\theta}{2}+ \sin \frac{\theta}{2}\right)$   &
		        $ \frac{1}{8}\left(\cos \frac{7\theta}{2} - \cos \frac{\theta}{2} \right) $\\
		        \hline
		    \end{tabular}
		    \label{tab:7_1_3_gcs}
		\end{table}
	
    \end{example}
	 Before introducing the Kraus decomposition of $R_Z(\theta)$ acting on a CSS code, we provide an alternative definition of generator coefficients which simplifies calculations. We first write $A_{\bm{\mu,\bm{\gamma}}(\theta)}$ as a linear combination of weight enumerators, then apply the MacWilliams Identities.
	 
	 %we provide a general derivation of \eqref{eqn:7_1_3_fs} so that we have a equivalent definition of generator coefficients, which is more friendly calculable. They key idea is first write $A_{\bm{\mu},\bm{\gamma}}(\theta)$ into a linear combination of evaluations of some weight enumerators, and then apply the MacWilliams identities to each of them to make simplifications.
	 
	 \begin{lemma}[Simplified Definition of Generator Coefficients]
	 \label{lemma:equ_def_gcs}
	     Consider a CSS($X,\mathcal{C}_2;Z,\mathcal{C}_1^\perp$) code, where $\bm{y}$ is the character vector for the $Z$-stabilizers $\left(\epsilon_{(\bm{0},\bm{z})} = (-1)^{\bm{z}\bm{y}^T}\right)$. Then, the generator coefficients $A_{\bm{\mu},\bm{\gamma}}(\theta)$ defined in \eqref{eqn:def_gc} can be written as
	     \begin{align}\label{eqn:A_mu,gamma}
	     &A_{\bm{\mu},\bm{\gamma}}(\theta) \nonumber \\
	     &=\frac{1}{|\mathcal{C}_1|} \sum_{\bm{z}\in \mathcal{C}_1 + \bm{y}}(-1)^{(\bm{\mu} \oplus\bm{\gamma})(\bm{z}\oplus \bm{y})^T} \left(e^{-\imath\frac{\theta}{2}}\right)^{n-2w_H(\bm{z})}. 
	     \end{align}
	 \end{lemma}
	 \begin{remark}\normalfont
	 \label{rem:eqv_gcs_trans_theta}
	     The original definition \eqref{eqn:def_gc} requires a sum over the weights of every coset $\mathcal{C}_1^\perp$. The alternative definition \eqref{eqn:A_mu,gamma} requires a sum over a single coset $\mathcal{C}_1+\bm{y}$, where the syndrome $\bm{\mu}$ and logical $\bm{\gamma}$ determine the hyperplane that specifies the signs in the sum.
	 \end{remark}
	 \begin{proof}
	     See Appendix \ref{subsec:proof_eqv_GCs_Rz}.
	 \end{proof}
	 
	 %The equivalent definition of generator coefficients does not require the weight enumerator of all the coset of $\mathcal{C}_1^\perp$. Instead, it needs only one weight enumerator of a fixed coset of $\mathcal{C}_1$ with representatives being the character vector $\bm{y}$ and the intersections of it and different hyperplanes determined by the syndromes $\bm{\mu}$ and logicals $\bm{\gamma}$. 
	 
	\subsection{Quadratic Form Diagonal Gates}
	\label{subsec:GCs_QFD}
	Rengaswamy et al. \cite{rengaswamy2019unifying} considered diagonal unitaries of the form 
	\begin{equation}\label{eqn:qfd_def}
	\tau_R^{(l)} = \sum_{\bm{v}\in \F_2^n} \xi_l^{\bm{v}R\bm{v}^T \bmod {2^l}} \ket{\bm{v}}\bra{\bm{v}},
	\end{equation}
	where $l\ge 1$ is an integer, $\xi_l=e^{\imath \frac{\pi}{2^{l-1}}}$, and $R$ is an $n\times n$ symmetric matrix with entries in $\Z_{2^l}$, the ring of integer modulo $2^l$. Note that the exponent $vRv^T \in \Z_{2^l}$. When $l=2$ and $R$ is binary, we obtain the diagonal Clifford unitaries. QFD gates defined by \eqref{eqn:qfd_def} include all $1$-local and $2$-local diagonal unitaries in the Clifford hierarchy, and they contain $R_Z(\theta)$ for $\theta = \frac{2\pi}{2^l}$, where $l\ge 1$ is an integer. 
	
	Recall that $N\times N$ Pauli matrices form an orthonormal basis for unitaries of size $N$ with respect to the normalized Hilbert-Schmidt inner product $\langle A,B\rangle \coloneqq \mathrm{Tr}(A^\dagger B)/N$. Hence,
	\begin{align*}
	\ket{\bm{v}}\bra{\bm{v}} 
	&= \sum_{\bm{a},\bm{b}\in \F_2^n} \frac{\mathrm{Tr}(\ket{\bm{v}}\bra{\bm{v}}E(\bm{a},\bm{b}))}{N}E(\bm{a},\bm{b})\\
% 	&= \sum_{\bm{a},\bm{b}\in \F_2^n}\left( \sum_{\bm{u}\in\F_2^n} \frac{\bra{\bm{u}}\bm{v}\rangle \bra{\bm{v}} E(\bm{a},\bm{b})\ket{\bm{u}}}{N}\right)E(\bm{a},\bm{b})\\
% 	&= \sum_{\bm{a},\bm{b}\in \F_2^n} \frac{\bra{\bm{v}}E(\bm{a},\bm{b})\ket{\bm{v}}}{N} E(\bm{a},\bm{b})\\
% 	&=\sum_{\bm{a},\bm{b}\in \F_2^n} \frac{\bra{\bm{v}}(-1)^{\bm{b}\bm{v}^T}\ket{\bm{v}\oplus \bm{a}}}{N} E(\bm{a},\bm{b})\\
	&=\frac{1}{2^n}\sum_{\bm{b}\in \F_2^n} (-1)^{\bm{b}\bm{v}^T} E(\bm{0},\bm{b}),\numberthis
	\end{align*}
	and the Pauli expansion of a QFD gate becomes
	\begin{align}
	\tau_R^{(l)} 
	&= \frac{1}{2^n} \sum_{\bm{u}\in \F_2^n} f(\bm{u})E(\bm{0},\bm{u}),
	\end{align}
	where 
	\begin{equation} \label{eqn:qfd_fu}
	    f(\bm{u}) = \sum_{\bm{v}\in\F_2^n} \xi_l^{\bm{v}R\bm{v}^T \bmod {2^l}} (-1)^{\bm{uv}^T} .
	\end{equation}
	\addtocounter{example}{+2}
	\begin{example}\label{examp4}
	\normalfont
		If $n=1,~l=3,~\xi_3=e^{\imath \frac{\pi}{4}},~ R=[1]$, then we have $f(0) =  1+e^{\imath\frac{\pi}{4}}$, $f(1)=  1-e^{\imath\frac{\pi}{4}}$, and
		$
		\tau_{R}^{(2)}  
% 		&= \frac{1}{2} \sum_{u\in \{0,1\}} \left(\sum_{v\in\{0,1\}} \left(e^{\imath \frac{\pi}{4} }\right)^{^{v^2 \bmod 8} } (-1)^{uv^T}\right) E(0,u) \\
		= \frac{1}{2}\left( 1+e^{\imath\frac{\pi}{4}}\right)E(0,0) + \frac{1}{2}\left( 1-e^{\imath\frac{\pi}{4}}\right)E(0,1)=T.
% 		\\
% 		&= 
% 		\begin{bmatrix}
% 		1 & 0 \\
% 		0 & e^{\imath \frac{\pi}{4}}
% 		\end{bmatrix}=T.\numberthis
		$
	\end{example}
	\begin{example}\label{examp5}
	\normalfont
	   Consider $n=2$, and $R=\begin{bmatrix}
		0 & 1 \\
		1 & 0
		\end{bmatrix}$. If $l=2$, then $~\xi_2=e^{\imath \frac{\pi}{2}}=\imath$ and 
		$
		\tau_{R}^{(2)}=\mathrm{C}Z
		\coloneqq\frac{1}{2}\left( E(\bm{0},\bm{0}) + E(\bm{0},01) + E(\bm{0},10) - E(\bm{0},\bm{1})\right).$ If $l=3$, then $~\xi_3=e^{\imath \frac{\pi}{4}}$ and 
		\begin{align*}
		\tau_{R}^{(3)} = \mathrm{C}P 
		&\coloneqq \frac{1}{4}( (3-\imath)E(\bm{0},\bm{0}) + (1+\imath)E(\bm{0},01) \\ 
		&~~~+ (1+\imath)E(\bm{0},10) - (1+\imath)E(\bm{0},\bm{1})). \numberthis
% 		\\
% 		&= \begin{bmatrix}
% 		1 & 0 & 0 & 0 \\
% 		0 & 1 & 0 & 0 \\
% 		0 & 0 & 1 & 0 \\
% 		0 & 0 & 0 & \imath
% 		\end{bmatrix} 
		\end{align*}
	\end{example}
	We substitute 
% 	$f(\bm{v})=\sum_{\bm{v}\in\F_2^n} \xi_l^{\bm{v}R\bm{v}^T \bmod {2^l}} (-1)^{\bm{uv}^T} $ 
    \eqref{eqn:qfd_fu} in \eqref{eqn:def_GC_UZ}, and obtain the generator coefficients for QFD gates 
	\begin{align}\label{eqn:def_gcs_qfd}
	&A_{\bm{\mu},\bm{\gamma}}(R,l) \coloneqq \nonumber \\
	&\frac{1}{2^n} \sum_{\bm{z}\in \mathcal{C}_1^\perp+\bm{\mu}+\bm{\gamma}}\epsilon_{(\bm{0},\bm{z})}\sum_{\bm{v}\in\F_2^n}\xi_l^{\bm{v}R\bm{v}^T\bmod{2^l}}(-1)^{\bm{zv}^T}.
	\end{align}
	Let $\bm{y}\in\F_2^n/\mathcal{C}_1$ be the character vector $\left( \epsilon_{(\bm{0},\bm{z})}=(-1)^{\bm{zy}^T} \right)$. Changing the order of summation, we have 
	\begin{equation}\label{eqn:gcs_qfd_2}
	A_{\bm{\mu},\bm{\gamma}}(R,l)
	=\frac{1}{2^n} \sum_{\bm{v}\in\F_2^n} p_{\bm{y}}(\bm{v},\bm{\mu},\bm{\gamma})\xi_l^{\bm{v}R\bm{v}^T\bmod{2^l}},
	\end{equation}
	where
% 	Let $\bm{u}\in \mathcal{C}_1^\perp$. It follows from  
% 	\begin{equation}
% 	(-1)^{\bm{zy}^T}(-1)^{\bm{zv}^T}=(-1)^{(\bm{\mu} \oplus \bm{\gamma} \oplus \bm{u})(\bm{y}  \oplus \bm{v})^T}=(-1)^{(\bm{\mu} \oplus \bm{\gamma})(\bm{y}  \oplus \bm{v})^T}(-1)^{\bm{u}(\bm{y}  \oplus \bm{v})^T}
% 	\end{equation}
% 	that 
    % The inner summation becomes
	\begin{align*} 
	&p_{\bm{y}}(\bm{v},\bm{\mu},\bm{\gamma})\\
	&=\sum_{\bm{z}\in \mathcal{C}_1^\perp+\bm{\mu}+\bm{\gamma}}
	(-1)^{\bm{zy}^T}(-1)^{\bm{zv}^T}\\
	&=(-1)^{(\bm{\mu} \oplus \bm{\gamma})(\bm{y} \oplus \bm{v})^T}\sum_{\bm{u}\in \mathcal{C}_1^\perp}(-1)^{\bm{u}(\bm{y} \oplus \bm{v})^T} \\
	&=\left\{
	\begin{array}{lc}
	|\mathcal{C}_1^\perp|  (-1)^{(\bm{\mu} \oplus \bm{\gamma})(\bm{y} \oplus \bm{v})^T},   \text{ if } \bm{y} \oplus \bm{v}\in \mathcal{C}_1,\\
	0, ~~~~~~~~~~~~~~~~~~~~~~~~~ \text{ otherwise.}
	\end{array}
	\right. \numberthis\label{eqn:qfd_gcs_simp_pros}
% 	\left\{\begin{array}{lc}
% 	|\mathcal{C}_1^\perp|  (-1)^{(\bm{\mu} \oplus \bm{\gamma})(\bm{y} \oplus \bm{v})^T}   &  \text{ if } \bm{y} \oplus \bm{v}\in \mathcal{C}_1\\
% 	0     &  \text{ otherwise }
% 	\end{array} \right..\numberthis\label{eqn:qfd_gcs_simp_pros}
	\end{align*}
	Substituting \eqref{eqn:qfd_gcs_simp_pros} in \eqref{eqn:gcs_qfd_2}, we obtain 
	\begin{align*}
	A_{\bm{\mu},\bm{\gamma}}(R,l)
% 	&=\frac{1}{2^n} \sum_{\bm{v}\in \mathcal{C}_1+\bm{y}}|\mathcal{C}_1^\perp| (-1)^{(\bm{\mu} \oplus \bm{\gamma})(\bm{y} \oplus \bm{v})^T} \xi_l^{\bm{v}R\bm{v}^T\bmod{2^l}}\\
	%A_{\mu,\gamma}(R,l) = \frac{1}{|C_1|} \sum_{\eta\in C_1} (-1)^{(\mu\oplus \gamma)\eta^T}\xi_l^{(\eta+y)R(\eta+y)^T\bmod{2^l}}
	&=\frac{1}{|\mathcal{C}_1|} \sum_{\bm{v}\in \mathcal{C}_1+\bm{y}}(-1)^{(\bm{\mu} \oplus \bm{\gamma})(\bm{y} \oplus \bm{v})^T} \xi_l^{\bm{v}R\bm{v}^T} \numberthis \label{eqn:qfd_gcs_simp}.
	\end{align*}
	When $R=I_n$, we obtain the transversal $Z$-rotation $R_Z(\frac{\pi}{2^{l-1}})$ up to a global phase. We now use \eqref{eqn:qfd_gcs_simp} to calculate generator coefficients of the $\llbr 4,2,2 \rrbr $ code.
	%which is an analog version of the generator coefficients after MacWilliams Identities in Lemma \ref{lemma:equ_def_gcs}. Note that $R_Z(\frac{\pi}{2^{l-1}})$ are QFD gates with $R=I_{n}$. Thus, \eqref{eqn:qfd_gcs_simp} matches \eqref{eqn:A_mu,gamma} when $R=I_{n}$ and $\theta = \frac{\pi}{2^{l-1}}$ up to a global phase $e^{\imath\frac{n\theta}{2}}$. We apply the simplified version to compute the generator coefficients for $\llbr 4,2,2 \rrbr $ code.
		    \begin{table*}[t]	    
	    \centering
		    \caption{Generator coefficients $A_{\bm{\mu},\bm{\gamma}}(R,l=2)$ for C$Z\otimes$C$Z$ applied to the $\llbr 4,2,2 \rrbr $ code with all positive signs.}
		    \renewcommand{\arraystretch}{1.3} 
		    \begin{tabular}{|c|c|c|c|c|}
		       \hline
		        \diagbox{$X$-syndromes}{$Z$-logicals} 
                &
		        $\bm{\gamma}=\bm{0}$ & $\bm{\gamma}= [0,0,1,1]$ & $\bm{\gamma}= [0,1,1,0]$ & $\bm{\gamma}= [0,1,0,1]$
		        \\
		        \hline
		          $\bm{\mu}=\bm{0}$ & 
		          $\frac{1}{2}$ & 
		          $-\frac{1}{2}$ &
		          $\frac{1}{2}$ &
		          $\frac{1}{2}$ \\
		          \hline
		        $\bm{\mu}=[1,0,0,0]$ & \multicolumn{4}{c|}{$0$} \\
		        \hline
		    \end{tabular}
		    \label{tab:4_2_2_gcs_CZ}
		\end{table*}
	    %We compute the generator coefficients of C$P^{\otimes 2}$ for comparison in Table \ref{tab:4_2_2_gcs_CP}.
	    
	   \begin{table*}[t]	    
	    \centering
		    \caption{Generator coefficients $A_{\bm{\mu},\bm{\gamma}}(R,l=3)$ of C$P\otimes$C$P$ for $\llbr 4,2,2 \rrbr $ code with all positive signs.}
		    \renewcommand{\arraystretch}{1.3} 
		    \begin{tabular}{|c|c|c|c|c|}
		       \hline
		        \diagbox{$X$-syndromes}{$Z$-logicals} 
                &
		        $\bm{\gamma}=\bm{0}$ & $\bm{\gamma}= [0,0,1,1]$ & $\bm{\gamma}= [0,1,1,0]$ & $\bm{\gamma}= [0,1,0,1]$
		        \\
		        \hline
		          $\bm{\mu}=\bm{0}$ & 
		          $\frac{1}{4}(2+\imath)$ & 
		          $\frac{1}{4}(-2+\imath)$ &
		          $-\frac{\imath}{4}$ &
		          $-\frac{\imath}{4}$ \\
		          \hline
		           $\bm{\mu}=[1,0,0,0]$ & \multicolumn{4}{c|}{$\frac{1}{4}$} \\
		      %  $\bm{\mu}=[1,0,0,0]$ & 
		      %    $\frac{1}{4}$ & 
		      %    $\frac{1}{4}$ &
		      %    $\frac{1}{4}$ &
		      %    $\frac{1}{4}$ \\
		        \hline
		    \end{tabular}
		    \label{tab:4_2_2_gcs_CP}
		\end{table*}
	\addtocounter{example}{-4}
	\begin{example}[Generator Coefficients of C$Z$ and C$P$ for the $\llbr 4,2,2 \rrbr $ code]
	\label{examp6}
	    \normalfont
	    The $\llbr 4,2,2 \rrbr $ code is a CSS code with $\mathcal{C}_1^\perp = \mathcal{C}_2=\{\bm{0},\bm{1}\}$. The $Z$-logical $\bm{\gamma} \in \langle [0,0,1,1], [0,1,1,0] \rangle $ and the $X$-syndrome $\bm{\mu} \in \langle [1,0,0,0] \rangle $. Assume all the stabilizers have positive signs (the character vector $\bm{y} = \bm{0}$). Set 
	    \begin{equation}
	    R=\begin{bmatrix}
		0 & 1 & 0 & 0 \\
		1 & 0 & 0 & 0 \\
		0 & 0 & 0 & 1 \\
		0 & 0 & 1 & 0
		\end{bmatrix}.
	    \end{equation} 
% 		$ R= \begin{bmatrix}
% 		0 & 1  \\
% 		1 & 0  \\
% 		\end{bmatrix} \otimes I_2$ 
		Setting $l=2$, we list the generator coefficients for C$Z^{\otimes 2}$ in Table \ref{tab:4_2_2_gcs_CZ}. 
	    Note that C$Z$ and C$P$ shared the same symmetric matrix $R$ but the level $l$ is different. Table \ref{tab:4_2_2_gcs_CP} lists the generator coefficients for C$P^{\otimes 2}$.

	\end{example}
	                                      
% 	\subsection{Controlled-Controlled-Z (CC$Z$)}
% 	\label{subsec:GCs_CCZ}
% 	Some of $3$-local diagonal gates are not include in the QFD gates, and the CC$Z$ gate is one of the example that outside the QFD representation and in the third-level of Clifford hierarchy. The Pauli expansion of CC$Z$ is
% 	\begin{align}
% 	    \mathrm{CC}Z 
% 	    &= \sum_{\bm{v}\in\F_2^3} \frac{1}{8}\mathrm{Tr}\left(E(\bm{0},\bm{v})^\dagger \mathrm{CC}Z\right)E(\bm{0},\bm{v})\\
% 	    &= \sum_{\bm{v}\in\F_2^3} \frac{1}{8}\mathrm{Tr}\left(E(\bm{0},\bm{v}) \left(\ket{00}\bra{00} + \ket{01}\bra{01} + \ket{10}\bra{10}\right)\otimes I + \ket{11}\bra{11}\otimes Z \right)E(\bm{0},\bm{v})\\
% 	    &= \sum_{\bm{v}\in\F_2^3} \frac{1}{8}\mathrm{Tr}\left( \left( \ket{00}\bra{00} + (-1)^{v_1}\ket{01}\bra{01} + (-1)^{v_2}\ket{10}\bra{10}\right)\otimes Z^{v_3} + (-1)^{v_1+v_2}\ket{11}\bra{11}\otimes Z^{v_3+1} \right)E(\bm{0},\bm{v}),
% 	\end{align}
% 	where $\bm{v}=(v_1,v_2,v_3)$.
% 	Thus, the generator coefficient of CC$Z$ is
% 	\begin{equation}
% 	    A_{\bm{\mu},\bm{\gamma}}(\mathrm{CC}Z) = \frac{1}{8}\sum_{\bm{z}\in\mathcal{C}_1^\perp +\bm{\mu} +\bm{\gamma}} 
%         \epsilon_{(\bm{0},\bm{z})} (1+(-1)^{z_1}+(-1)^{z_2})\mathrm{Tr}(Z^{z_3}) + (-1)^{z_1+z_2}\mathrm{Tr}(Z^{z_3+1}).
% 	\end{equation}
% 	\begin{example}\label{examp7}
% 	\normalfont
% 	(Generator Coefficients of CC$Z$ for the $3$-qubit code with negative sign).
	    
% 	\end{example}
	
	\section{Average Logical Channel}
	\label{sec:avg_log_chl}
	%One reason we defined generator coefficients as the way in Section \ref{sec:intro_gcs} is that they help to describe the average logical channel for CSS codes (we extend this argument for stabilizer codes in Appendix \ref{sec:stab_gcf}).  
	 We investigate the effect of $U_Z$ acting on a CSS codespace $\mathcal{V}(\mathcal{S})$ by considering the following steps:
    % 	The generator coefficients also provides the analog Kraus representation of $R_Z(\theta)$ on a CSS codespace $V(S)$. Given any density operator $\rho_i \in V(S)$, we go through the following steps:
	\begin{enumerate}
	    \item Choose any initial density operator $\rho_1$ in the CSS codespace $\mathcal{V}(\mathcal{S})$. Then, we have $\rho_1 = \Pi_{\mathcal{S}} \rho_1\Pi_{\mathcal{S}}$.
	    
		\item Apply $U_Z$ physically. Then the system evolves to 
		\begin{equation}\label{eqn:KR_step2}
		    \rho_2 = U_Z\rho_1 U_Z^\dagger = U_Z\Pi_{\mathcal{S}} \rho_1\Pi_{\mathcal{S}}U_Z^\dagger.
		\end{equation}
		
		\item Measure with $X$-stabilizers to obtain the syndrome $\bm{\mu} \in \F_2^n/\mathcal{C}_2^\perp$. It follows from \eqref{eqn:do_measure_evo} that the system evolves to
		\begin{align*}
	    &\rho_3 
	    = \sum_{\bm{\mu}\in \F_2/\mathcal{C}_2^\perp} \Pi_{\mathcal{S}_{X(\bm{\mu})}}\rho_2 \Pi_{\mathcal{S}_{X(\bm{\mu})}} \\
	    &= \sum_{\bm{\mu}\in \F_2/\mathcal{C}_2^\perp} \left(\Pi_{\mathcal{S}_{X(\bm{\mu})}}U_Z \Pi_{\mathcal{S}}\right)
	    \rho_1 
	    \left(\Pi_{\mathcal{S}}U_Z^\dagger  \Pi_{\mathcal{S}_{X(\bm{\mu})}}\right) \numberthis
	    \end{align*}
	    
		\item Based on the syndrome $\bm{\mu}$, we apply a Pauli correction to map the state back to $\mathcal{V}(\mathcal{S})$. This correction may introduce some logical operator $\epsilon_{(\bm{0},\bm{\gamma_{\mu}})}E(\bm{0},\bm{\gamma_{\mu}})$. The final state $\rho_4$ is in the CSS codespace.
	\end{enumerate} 
    Generator coefficients help describe the average logical channel resulting from $U_Z$ acting on a CSS codespace (steps 1-4). We extend our approach to arbitrary stabilizer codes in Appendix \ref{sec:stab_gcf}.
    
    \subsection{The Kraus Representation} %  of $R_Z(\theta)$ acting on a CSS code	
    \label{subsec:Kraus}
    %The generator coefficients take advantage of the structure of CSS codes and reveal the logical channel of $R_Z(\theta)$. More specifically, generator coefficients act as the coefficient in the Pauli expansion of Kraus operators describing the averaging logical channel correspond to $R_Z(\theta)$. 
   
    % 	We verify it step by step, at the beginning, since $\rho_i\in V(S)$, we have $\rho_i = \Pi_{S} \rho_i\Pi_{S}.$
    % 	After the first step, we conjugate the initial state by $R_Z(\theta)$ and obtain  
    % 	\begin{equation}
    % 	\rho_1 = R_Z(\theta)\rho_i \left(R_Z(\theta)\right)^\dagger = R_Z(\theta)\Pi_{S} \rho_i\Pi_{S}\left(R_Z(\theta)\right)^\dagger.
    % 	\end{equation}
    % 	Then, we do the $X$-measurement and have 
    % 	\begin{equation}
    % 	\rho_2 = \sum_{\mu\in \F_2/C_2^\perp} \Pi_{S_{X(\mu)}}R_Z(\theta) \Pi_{S}\rho_i \Pi_{S}\left(R_Z(\theta)\right)^\dagger  \Pi_{S_{X(\mu)}} .
    % 	\end{equation}
    Kraus operators describe the logical channels obtained by averaging the action of $U_Z$ over density operators in $\mathcal{V}(\mathcal{S})$. Generator coefficient appear as the coefficients in the Pauli expansion of Kraus operators. We use generator coefficients to simplify the term $U_Z \Pi_{\mathcal{S}}$ in \eqref{eqn:KR_step2}. It follows from \eqref{eqn:Pi_SZ_U_Z} and the derivation in Appendix \ref{seubsec:deriv_eqn} that
	\begin{align*}
	U_Z \Pi_{\mathcal{S}} 
% 	&= U_Z\Pi_{\mathcal{S}_Z}\Pi_{\mathcal{S}_X} \\
% 	&=\frac{1}{2^{n-k_1+k_2}}\sum_{\bm{\mu}\in \F_2^n/\mathcal{C}_2^\perp}\sum_{\bm{\gamma}\in \mathcal{C}_2^\perp/\mathcal{C}_1^\perp} A_{\bm{\mu},\bm{\gamma}}\left( \sum_{\bm{u}\in \mathcal{C}_1^\perp+\bm{\mu}+\bm{\gamma}}\epsilon_{(\bm{0},\bm{u})}E(\bm{0},\bm{u})\right)\left(\sum_{\bm{a}\in \mathcal{C}_2}\epsilon_{(\bm{a},\bm{0})}E(\bm{a},\bm{0})\right)\\
% 	&=\frac{1}{2^{n-k_1+k_2}}\sum_{\bm{\mu}\in \F_2^n/\mathcal{C}_2^\perp}\sum_{\bm{\gamma}\in \mathcal{C}_2^\perp/\mathcal{C}_1^\perp} A_{\bm{\mu},\bm{\gamma}}\left(\sum_{\bm{a}\in \mathcal{C}_2}(-1)^{\bm{a}\bm{\mu}^T}\epsilon_{(\bm{a},\bm{0})}E(\bm{a},\bm{0})\right)\left( \sum_{\bm{u}\in \mathcal{C}_1^\perp+\bm{\mu}+\bm{\gamma}}\epsilon_{(\bm{0},\bm{u})}E(\bm{0},\bm{u})\right)\\
	&=\sum_{\bm{\mu}\in \F_2^n/\mathcal{C}_2^\perp}\Pi_{\mathcal{S}_X(\bm{\mu})} 
	\sum_{\bm{\gamma}\in \mathcal{C}_2^\perp/\mathcal{C}_1^\perp} A_{\bm{\mu},\bm{\gamma}} ~q(\bm{\mu},\bm{\gamma}),
	%\left( \sum_{\bm{u}\in \mathcal{C}_1^\perp+\bm{\mu}+\bm{\gamma}}\epsilon_{(\bm{0},\bm{u})}E(\bm{0},\bm{u})\right)\right),
	\numberthis \label{eqn:kraus_simplify}
	\end{align*}
	where $\Pi_{\mathcal{S}_X(\bm{\mu})} = \frac{1}{|\mathcal{C}_2|}\sum_{\bm{a}\in \mathcal{C}_2} (-1)^{\bm{a}\bm{\mu}^T} \epsilon_{(\bm{a},\bm{0})} E(\bm{a},\bm{0})$ as described in \eqref{eqn:Pi_SX_mu}, and 
	\begin{equation}
	    q(\bm{\mu},\bm{\gamma}) \coloneqq \frac{1}{2^{n-k_1}}\sum_{\bm{u}\in \mathcal{C}_1^\perp+\bm{\mu}+\bm{\gamma}}\epsilon_{(\bm{0},\bm{u})}E(\bm{0},\bm{u}).
	\end{equation}
	Since the projectors $\{\Pi_{\mathcal{S}_X(\bm{\mu})}\}_{\bm{\mu}\in\F_2^n/\mathcal{C}_2^\perp}$ are pairwise orthogonal, it follows from that for any fixed $\bm{\mu_0}\in\F_2^n/\mathcal{C}_2^\perp $, we have
	\begin{align*}
	\Pi_{\mathcal{S}_{X(\bm{\mu_0})}}U_Z\Pi_{\mathcal{S}} =\Pi_{\mathcal{S}_X(\bm{\mu_0})} \sum_{\bm{\gamma}\in \mathcal{C}_2^\perp/\mathcal{C}_1^\perp} A_{\bm{\mu_0},
	\bm{\gamma}} ~q(\bm{\mu_0},\bm{\gamma}).\numberthis
% 	\left( \sum_{\bm{u}\in \mathcal{C}_1^\perp+\bm{\mu_0}+\bm{\gamma}}\epsilon_{(\bm{0},\bm{u})}E(\bm{0},\bm{u})\right).
	\end{align*}
	%Let $K\coloneqq \frac{1}{2^{n-k_1}}\sum_{\bm{u}\in \mathcal{C}_1^\perp+\bm{\mu_0}+\bm{\gamma_0}}\epsilon_{(\bm{0},\bm{u})}E(\bm{0},\bm{u}) $. 
	Since $\rho_1$ describes an ensemble of states in the codespace $\mathcal{V}(\mathcal{S})$, it follows from that for fixed $\bm{\gamma_0}\in \mathcal{C}_2^\perp/\mathcal{C}_1^\perp$, we have 
	\begin{align}
	    q(\bm{\mu_0},\bm{\gamma_0}) \rho_1 q(\bm{\mu_0},\bm{\gamma_0}) = K
	    \rho_1 
	    K,
	\end{align}
	where $K\coloneqq \epsilon_{(\bm{0},\bm{\mu_0} \oplus \bm{\gamma_0})}E(\bm{0},\bm{\mu_0} \oplus \bm{\gamma_0}) $.
	Thus, we may write $\rho_3$ as
	\begin{equation}
	    \rho_3 = \sum_{\bm{\mu}\in \F_2^n/\mathcal{C}_2^\perp} \Pi_{\mathcal{S}_{X(\bm{\mu})}} 
	    K_1
	    \rho_1
	    K_1
	   % \left(\sum_{\bm{\gamma}\in \mathcal{C}_2^\perp/\mathcal{C}_1^\perp} A_{\bm{\mu},\bm{\gamma}} \epsilon_{(\bm{0},\bm{\mu} \oplus \bm{\gamma})}E(\bm{0},\bm{\mu} \oplus \bm{\gamma})\right),
	\end{equation}
	where $K_1 \coloneqq \sum_{\bm{\gamma}\in \mathcal{C}_2^\perp/\mathcal{C}_1^\perp} A_{\bm{\mu} ,
	    \bm{\gamma}} ~\epsilon_{(\bm{0},\bm{\mu} \oplus \bm{\gamma})}E(\bm{0},\bm{\mu} \oplus \bm{\gamma})$.
	Although the sign $\epsilon$ does not matter here, we carry it along for consistency with the logical Pauli $Z$ operators derived in \eqref{eqn:css_log_paulis}. 
	Based on the syndrome $\bm{\mu}$, the decoder applies a correction and maps the quantum state back to the codespace $\mathcal{V}(\mathcal{S})$. This correction might induce some undetectable $Z$-logical $\epsilon_{(\bm{0},\bm{\gamma_{\mu}})}E(\bm{0},\bm{\gamma_{\mu}})$ with $\bm{\gamma}_{\bm{0}}=\bm{0}$. 
	Hence, the final state after step 4 becomes
	\begin{equation}\label{eqn:kraus_rep}
	\rho_4= \sum_{\bm{\mu}\in \F_2^n/\mathcal{C}_2^\perp} B_{\bm{\mu}} \rho_1 B_{\bm{\mu}}^\dagger,
	\end{equation}
	where 
	\begin{align} 
	 B_{\bm{\mu}} 
	 &\coloneqq \epsilon_{(\bm{0},\bm{\gamma_{\mu}})}E(\bm{0},\bm{\gamma_{\mu}}) \sum_{\bm{\gamma}\in \mathcal{C}_2^\perp / \mathcal{C}_1^\perp} A_{\bm{\mu},\bm{\gamma}}~ \epsilon_{(\bm{0},\bm{\gamma})}E(\bm{0},\bm{\gamma}) \nonumber \\
	 &=\sum_{\bm{\gamma}\in \mathcal{C}_2^\perp / \mathcal{C}_1^\perp} A_{\bm{\mu},\bm{\gamma}}~\epsilon_{(\bm{0},\bm{\gamma} \oplus \bm{\gamma_{\mu}})}E(\bm{0},\bm{\gamma} \oplus \bm{\gamma_{\mu}}),\label{eqn:kraus_ops}
	\end{align}
	is the effective physical operator corresponding to syndrome $\bm{\mu}$.
    It follows from \eqref{eqn:css_log_paulis} that for $\bm{\gamma} \in \mathcal{C}_2^\perp / \mathcal{C}_1^\perp$,  $\epsilon_{(\bm{0},\bm{\gamma} \oplus \bm{\gamma_{\mu}})}E(\bm{0},\bm{\gamma} \oplus \bm{\gamma_{\mu}})$ is a logical Pauli $Z$, and \eqref{eqn:kraus_rep}, \eqref{eqn:kraus_ops} can be considered just in the logical space. 
    
    Note that the evolution described in \eqref{eqn:kraus_rep} works for any initial code state $\rho_1$ in step 1. The interaction between the diagonal gate $U_Z$ and the structure of CSS code in step 2 is captured in the generator coefficients $A_{\bm{\mu},\bm{\gamma}}$. The syndrome of the measurement in step 3 is reflected by the sum in \eqref{eqn:kraus_rep}, and the decoder chosen in step 4 is expressed by some logical Pauli $Z$ determined by $\bm{\gamma}_{\bm{\mu}}$ for each syndrome.
    
    To show $\{B_{\bm{\mu}}\}_{\bm{\mu}\in \F_2/\mathcal{C}_2^\perp}$ is the set of Kraus operators, we need to verify that
	\begin{equation}\label{eqn:kraus_prop}
	\sum_{\bm{\mu}\in\F_2^n/\mathcal{C}_2^\perp } B_{\bm{\mu}}^\dagger B_{\bm{\mu}} = I.
	\end{equation}
	We may simplify the summation as 
% 	Now, we use Theorem \ref{thm:GC_prop} to simplify
	\begin{align*}
	&\sum_{\bm{\mu}} 
	B_{\bm{\mu}}^\dagger B_{\bm{\mu}} \\
% 	&= \sum_{\bm{\mu}\in\F_2^n/\mathcal{C}_2^\perp } 
%     \sum_{\bm{\gamma},\bm{\gamma'}\in \mathcal{C}_2^\perp / \mathcal{C}_1^\perp} 
% 	\overline{A_{\bm{\mu},\bm{\gamma}}}
% 	A_{\bm{\mu},\bm{\gamma'}} 
% 	\epsilon_{(\bm{0},\bm{\gamma} \oplus \bm{\gamma_{\mu}})}
% 	\epsilon_{(\bm{0},\bm{\gamma'} \oplus \bm{\gamma_{\mu})}}
% 	E(\bm{0},\bm{\gamma} \oplus \bm{\gamma_{\mu}})  E(\bm{0},\bm{\gamma'} \oplus \bm{\gamma_{\mu}})\\
	&= \sum_{\bm{\mu}}\sum_{\bm{\gamma}}|A_{\bm{\mu},\bm{\gamma}}|^2 I \\
	&~+ \sum_{\bm{\mu}}\sum_{\bm{\gamma} \neq \bm{\gamma'}} \overline{A_{\bm{\mu},\bm{\gamma}}}A_{\bm{\mu},\bm{\gamma'}} ~\epsilon_{(\bm{0},\bm{\gamma} \oplus \bm{\gamma'})} E(\bm{0},\bm{\gamma} \oplus \bm{\gamma'})\\
	%&= \sum_{\bm{\mu}}\sum_{\bm{\eta}\in \mathcal{C}_2^\perp /\mathcal{C}_1^\perp}\sum_{\bm{\gamma}} \overline{A_{\bm{\mu},\bm{\gamma}}}A_{\bm{\mu},\bm{\eta} \oplus \bm{\gamma}}\epsilon_{(\bm{0},\bm{\eta})} E(\bm{0},\bm{\eta}), 
	&= \sum_{\bm{\eta}} \epsilon_{(\bm{0},\bm{\eta})}
	 \left(\sum_{\bm{\mu}} \sum_{\bm{\gamma}} \overline{A_{\bm{\mu},\bm{\gamma}}}A_{\bm{\mu},\bm{\eta} \oplus \bm{\gamma}}\right) E(\bm{0},\bm{\eta}), 
	\numberthis \label{eqn:kr_sum_to_1}
% 	&= I.
	\end{align*}
	where the new variable $ \bm{\eta} = \bm{\gamma} \oplus \bm{\gamma'}\in \mathcal{C}_2^\perp /\mathcal{C}_1^\perp$. 
	In Theorem \ref{thm:GCs_prop}, we verify \eqref{eqn:kraus_prop} by showing that the coefficient of $E(\bm{0},\bm{0})=I$ is 1 
	%(Lemma \ref{lemma:sumisone}) 
	and that the coefficients of $E(\bm{0},\bm{\eta}),~ \bm{\eta}\neq \bm{0}$ are all zero. Theorem \ref{thm:GCs_prop} describes the general property of generator coefficients, which mainly because quantum gates are unitaries.
    \begin{theorem}\label{thm:GCs_prop}
    Suppose that a $Z$-unitary gate $U_Z=\sum_{v\in \F_2^n}f(\bm{v})E(\bm{0},\bm{v})$ induces generator coefficients $A_{{\bm{\mu}},{\bm{\gamma}}}$ on a CSS($X,\mathcal{C}_2;Z,\mathcal{C}_1^\perp$) code. If $\bm{\eta} \in \mathcal{C}_2^\perp /\mathcal{C}_1^\perp$, then 
        % 	Let $A_{{\bm{\mu}},{\bm{\gamma}}}$ be defined on certain CSS($X,\mathcal{C}_2;Z,\mathcal{C}_1^\perp$) code for some $Z$-unitary gate $U_Z=\sum_{v\in \F_2^n}f(\bm{v})E(\bm{0},\bm{v})$. Let $\bm{\eta} \in \mathcal{C}_2^\perp /\mathcal{C}_1^\perp $. Then,
	\begin{align}\label{eqn:sum_squre_one}
	\sum_{{\bm{\mu}} \in \F_2^n/\mathcal{C}_2^\perp} \sum_{\bm{\gamma} \in \mathcal{C}_2^\perp/\mathcal{C}_1^\perp} \overline{A_{\bm{\mu},\bm{\gamma}}} A_{\bm{\mu},\bm{\eta} \oplus {\bm\gamma}} = 
	\left\{ \begin{array}{lc}
	1, &  \text{ if } \bm{\eta} = \bm{0}, \\
	0, &  \text{ if } \bm{\eta} \neq \bm{0}.
	\end{array}\right.
	\end{align}
\end{theorem}
\begin{proof}
% 	If $\bm{\eta}=\bm{0}$, then 
% 	\begin{align}
% 	\overline{A_{{\mu},{\gamma}}} A_{{\mu},{\eta} \oplus {\gamma}} =|A_{\mu,\gamma}|^2&=\left(\sum_{z\in C_1^\perp+\mu+\gamma}\epsilon_{(0,z)}f(n,z) \right)\left(\sum_{z'\in C_1^\perp+\mu+\gamma}\epsilon_{(0,z')} f(n,z')\right)\\
% 	&=\sum_{w\in C_1^\perp}\epsilon_{(0,w)}\left(\sum_{z\in C_1^\perp + \mu + \gamma}f(n,z)\overline{f(n,z\oplus w)} \right).
% 	\end{align}
% 	Therefore, we have
% 	\begin{align}
% 	\sum_{{\mu}\in \F_2^n / {C}_2^\perp} 
% 	\sum_{{\gamma} \in {C}_2^\perp / {C}_1^\perp } \left|A_{{\mu},{\gamma}}\right|^2
% 	&=\sum_{{\mu}\in \F_2^n / {C}_2^\perp} 
% 	\sum_{{\gamma} \in {C}_2^\perp / {C}_1^\perp }
% 	\sum_{w\in C_1^\perp}\epsilon_{(0,w)}
% 	\left(\sum_{z\in C_1^\perp + \mu + \gamma}f(n,z)\overline{f(n,z\oplus w)} \right)\\
% 	&=\sum_{w\in C_1^\perp}\epsilon_{(0,w)}
% 	\left( \sum_{{\mu}\in \F_2^n / {C}_2^\perp} \sum_{{\gamma} \in {C}_2^\perp / {C}_1^\perp }
% 	\sum_{z\in C_1^\perp + \mu + \gamma}f(n,z)\overline{f(n,z\oplus w)} \right)\\
% 	&=\sum_{w\in C_1^\perp}\epsilon_{(0,w)}
% 	\left( \sum_{z\in \F_2^n}f(n,z)\overline{f(n,z\oplus w)} \right)\\
% 	&=\epsilon_{(0,0)}\\
% 	&=1,
% 	\end{align}
% 	where the last step follows from \eqref{eqn:unitarycoefficients}.
	
% 		If $\eta\neq\mathbf{0}$, then 
  See Appendix \ref{subsec:proof_Kraus_ver}.
\end{proof}
    We conclude that the Kraus operators describing the action of $U_Z$ on a CSS code are given by \eqref{eqn:kraus_ops}.
    
    When $U_Z = R_Z(\theta)$, the generator coefficients $A_{\bm{\mu},\bm{\gamma}}$ take the form \eqref{eqn:def_gc}. Consider now a one-logical-qubit system, where one of the pair $(A_{\bm{\mu}=\bm{0},\bm{\gamma}=\bm{0}}(\theta), A_{\bm{\mu}=\bm{0},\bm{\gamma}\neq \bm{0}}(\theta))$ is real and the other is pure imaginary. Then the logical qubit is rotated with angle $\theta_L$ and we can express $\theta_L$ in terms of the physical rotation angle $\theta$ \cite{debroy2021optimizing} as 
    \begin{equation}
        \theta_L(\theta) = 2\tan^{-1}\left(\imath\frac{A_{\bm{\mu} = \bm{0},\bm{\gamma}\neq \bm{0}}(\theta)}{A_{\bm{\mu} =\bm{0},\bm{\gamma}=\bm{0}}(\theta)}\right).
    \end{equation}
    See Appendix \ref{sec:log_rot} for details. We again take the Steane code as an example, substitute the values from Table \ref{tab:7_1_3_gcs} and obtain the logical rotation angle 
    \begin{align*}\label{eqn:steane_trivial_angles}
        \theta_L(\theta) 
        &= 2\tan^{-1}\left(\frac{\sin\frac{7\theta}{2} - 7\sin\frac{\theta}{2}}{\cos\frac{7\theta}{2}+7\cos\frac{\theta}{2}}\right) \\
        &= -\frac{28}{15}\theta^3 + O(\theta^5). \numberthis
    \end{align*}
    Figure \ref{fig:Strane_Rot} plots $\theta_L(\theta)$ displaying third-order convergence about $\theta=0$. Note that $\theta_L(\frac{\pi}{4}) = -\frac{\pi}{4}$. In Appendix \ref{sec:MSD_Steane}, we explain how $R_Z(\frac{\pi}{4})$ supports magic state distillation with the aid of a logical Phase gate. When $\theta < \frac{\pi}{4}$, $\theta_L < \theta$, and the Steane code might be applied to convert 7 noisy copies of the state $(\ket{0}+e^{\imath \theta}\ket{1})/\sqrt{2}$ into 1 copy of the state $(\ket{0}+e^{\imath \theta_L}\ket{1})/\sqrt{2}$ with higher fidelity.
    
    \begin{figure}[h!]
        \centering
        \includegraphics[scale=0.2]{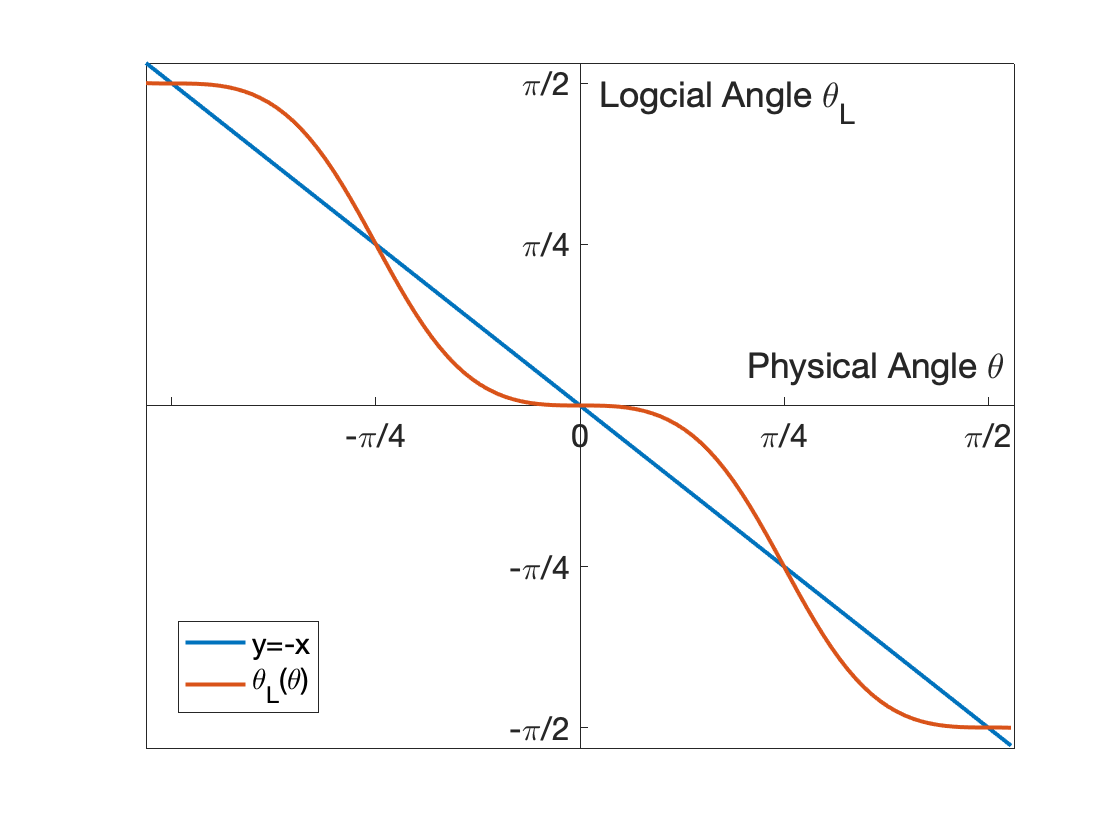}
        \caption{The Steane Code: the logical angle $\theta_L$ in terms of physical angle $\theta$, assuming we observe the trivial syndrome.}
      \label{fig:Strane_Rot}
    \end{figure}
    We now compute all Kraus operators induced by $R_Z(\theta)$ acting on the Steane code.
    \addtocounter{example}{-2}
    \begin{example}[continued]
    %[Kraus Representation of $R_Z(\frac{\pi}{4})$ for Steane code]
    \label{examp8}
    \normalfont
    We take the data in Table \ref{tab:7_1_3_gcs} and substitute $\theta=\frac{\pi}{4}$ to obtain
    \begin{align}
        A_{\bm{0}, \bm{0}} \left(\frac{\pi}{4}\right)&= \frac{3}{4}\cos\frac{\pi}{8},~ 
        A_{\bm{0}, \bm{1}} \left(\frac{\pi}{4}\right)= \frac{3}{4} \imath \sin\frac{\pi}{8},\nonumber\\
        &A_{\bm\mu\neq\bm{0}, \bm{0}} \left(\frac{\pi}{4}\right)= -\frac{1}{4} \imath \sin\frac{\pi}{8},~\nonumber\\
        &A_{\bm\mu\neq\bm{0}, \bm{1}} \left(\frac{\pi}{4}\right)= -\frac{1}{4}\cos\frac{\pi}{8}. \label{eqn:Steane_Kraus}
    \end{align}
    We assume $\bm{\gamma}_{\bm{\mu}}=\bm{0}$ for all $\bm{\mu}$, and use these generator coefficients to
    compute the Kraus operators
    \begin{align}
        B_{\bm{\mu}=\bm{0}}\left(\frac{\pi}{4}\right) 
        &= \frac{3}{4}\cos\frac{\pi}{8}\bar{I} + \frac{3}{4}\imath\sin\frac{\pi}{8}\bar{Z} 
        \equiv\frac{3}{4}\bar{T^\dagger},\\
        B_{\bm{\mu}\neq\bm{0}}\left(\frac{\pi}{4}\right) 
        &= -\frac{1}{4} \imath \sin\frac{\pi}{8} \bar{I} - \frac{1}{4}\cos\frac{\pi}{8}\bar{Z} 
        \equiv \frac{1}{4}\bar{Z}\bar{T^\dagger},
    \end{align}
    which describe the average logical channel corresponds to the transversal $T$ gate. Reichardt \cite{reichardt2005quantum} discussed the $\llbr 7,1,3 \rrbr $ Steane code in magic state distillation. The computed average logical channel makes it clear that we can choose proper corrections based on syndromes ($\bm{\gamma}_{\bm{\mu}}=\bar{Z}$ for $\bm{\mu}\neq \bm{0}$) to obtain the logical operator $T^\dagger$ from all the syndromes. 
    \end{example}
    
    Note that the Steane code is not a triorthogonal code \cite{bravyi2012magic}, but it can be used in state distillation \cite{reichardt2005quantum}. The generator coefficients framework may help to characterize codes that are not preserved by transversal $T$ but realize a logical $T$ gate when the trivial syndrome is observed. Recently, Vasmer and Kubica \cite{vasmer2021morphing} introduced a new $\llbr 10,1,2\rrbr$ code by \emph{morphing} the $\llbr 15,1,3\rrbr$ quantum Reed-Muller code \cite{knill1996accuracy,bravyi2005universal} and the $\llbr 8,3,2\rrbr$ color code \cite{campbell2017unified}. It provides the first protocol in state distillation that supports a fault-tolerant logical $T$ gate from a diagonal physical gate that is not transversal $T$. The generator coefficient framework applies to arbitrary diagonal gates, and may facilitate finding more examples of distillation.
    
    When $U_Z$ is a QFD gate, the Kraus operators can be derived in the same way. Table \ref{tab:4_2_2_gcs_CZ} in Example \ref{examp6} implies that the $\llbr 4,2,2 \rrbr $ code is preserved by C$Z^{\otimes 2}$ and that the induced logical operator is $Z_{1}^L \circ$ C$Z^L$.

	\subsection{Probability of Observing Different $X$-Syndromes}	
	\label{subsec:prob}
	The Kraus operators derived in Section \ref{subsec:Kraus} describe logical evolution conditioned on different outcomes from stabilizer measurement, and it is natural to calculate the probability of observing different syndromes $\bm{\mu}$. 
	Generator coefficients provide a means of calculating these probabilities that illuminates dependence on the initial state, and we will provide examples where the initial state and the outcome of syndrome measurement are entangled. 
	%It turns out that the generator coefficients also help to calculate these probabilities, in which it separates the dependence on the initial state. In some cases, we can see the entanglement between the initial state and the syndrome outcomes, that is, after obtaining some syndrome $\bm{\mu}$ from the measurement, some information of the logical qubits leak. The Steane code and the $\llbr 4,2,2 \rrbr $ examples are revisited after introducing the formula for the probabilities.  
	
	Consider a CSS($X,\mathcal{C}_2;Z,\mathcal{C}_1^\perp$) code with codespace $\mathcal{V}(\mathcal{S})$. For any fixed $\ket{\phi}\in \mathcal{V}(\mathcal{S})$ 
    % 	$= \{\ket{\psi} \in \mathbb{C}^N: g\ket{\psi} = \ket{\psi} \text{ for all } g\in \mathcal{S}\}$
    , we first apply $U_Z$, and then measure with projectors $\{\Pi_{\mathcal{S}_{X({\bm{\mu}})}}\}_{\bm{\mu} \in \F_2^n/\mathcal{C}_2^\perp}$, where $\Pi_{\mathcal{S}_{X({\bm{\mu}})}} = \frac{1}{|\mathcal{C}_2|}\sum_{\bm{a}\in \mathcal{C}_2} (-1)^{\bm{a}\bm{\mu}^T} \epsilon_{(\bm{a},\bm{0)}} E(\bm{a},\bm{0})$. Then the probability of obtaining a syndrome $\bm{\mu} \in \F_2^n/\mathcal{C}_2^\perp$ is
	\begin{equation}
	p_{\bm{\mu}}\left(\ket{\phi}\right) = \bra{\phi} U_Z^\dagger \Pi_{\mathcal{S}_X(\bm{\mu})} U_Z \ket{\phi}.
	\end{equation}
	It follows from equation \eqref{eqn:Pi_SZ_U_Z} that
	\begin{align*}
	U_Z\ket{\phi} 
	&= U_Z \Pi_{\mathcal{S}_Z}\ket{\phi} \\
	&= \sum_{\bm{\mu}}\sum_{\bm{\gamma}} A_{\bm{\mu},\bm{\gamma}} ~\epsilon_{(\bm{0}, \bm{\mu} \oplus \bm{\gamma})}E(\bm{0},\bm{\mu} \oplus \bm{\gamma} ) \ket{\phi}, \numberthis
	\end{align*} % \bm{\mu} \in \F_2^n/\mathcal{C}_2^\perp \bm{\gamma}\in \mathcal{C}_2^\perp/\mathcal{C}_1^\perp
	and similarly
	\begin{align*}
	\bra{\phi}U_Z^\dagger 
	&= \bra{\phi} \Pi_{\mathcal{S}_Z}U_Z^\dagger \\
	&= \bra{\phi}\sum_{\bm{\mu}}\sum_{\bm{\gamma}} \overline{A_{\bm{\mu},\bm{\gamma}}}~ \epsilon_{(0,\bm{\mu} \oplus \bm{\gamma})}E(0,\bm{\mu} \oplus \bm{\gamma}). \numberthis
	\end{align*}
	%Note that $\Pi_{S_X(\mu)}\Pi_{S_X(\mu)} = \Pi_{S_X(\mu)}$ as it is a projector, we have 
	For any fixed $\bm{\mu_0} \in \F_2^n/\mathcal{C}_2^\perp$, since $\Pi_{\mathcal{S}_X(\bm{\mu_0})}\Pi_{\mathcal{S}_X(\bm{\mu_0})} = \Pi_{\mathcal{S}_X(\bm{\mu_0})}$, we have
	\begin{equation}\label{eqn:p_mu_0}
	p_{\bm{\mu_0}} 
    % 	= \bra{\phi} U_Z^\dagger \Pi_{S_X(\mu_0)} U_Z\ket{\phi} 
    =  \bra{\phi} \Pi_{\mathcal{S}_Z} U_Z^\dagger \Pi_{\mathcal{S}_X(\bm{\mu_0})}  \Pi_{\mathcal{S}_X(\bm{\mu_0})} U_Z\Pi_{\mathcal{S}_Z} \ket{\phi}.
	\end{equation}
% 	To compute the probability, we first simplify
    It follows from the simplification in Appendix \ref{subsec:derv_simplify} of the later half in \eqref{eqn:p_mu_0} that  
	\begin{align*}
	&\Pi_{\mathcal{S}_X(\bm{\mu_0})} U_Z\Pi_{\mathcal{S}_Z} \ket{\phi}\\
% 	& = \frac{1}{|\mathcal{C}_2|}\sum_{\bm{a}\in \mathcal{C}_2} (-1)^{\bm{a}\bm{\mu_0}^T} \epsilon_{(\bm{a},\bm{0)}} E(\bm{a},\bm{0}) \sum_{\bm{\mu}\in \F_2^n/\mathcal{C}_2^\perp}\sum_{\bm{\gamma}\in \mathcal{C}_2^\perp/\mathcal{C}_1^\perp} A_{\bm{\mu},\bm{\gamma}} \epsilon_{(\bm{0},\bm{\mu} \oplus \bm{\gamma})}E(\bm{0},\bm{\mu} \oplus \bm{\gamma}) \ket{\phi}\\
% 	& = \frac{1}{|\mathcal{C}_2|}\sum_{\bm{\mu}}\sum_{\bm{\gamma}} 
% 	A_{\bm{\mu},\bm{\gamma}} \epsilon_{(\bm{0},\bm{\mu} \oplus \bm{\gamma})}E(\bm{0},\bm{\mu} \oplus \bm{\gamma})  \sum_{\bm{a}\in \mathcal{C}_2} (-1)^{\bm{a}(\bm{\mu}+\bm{\mu_0})^T} \epsilon_{(\bm{a},\bm{0})} E(\bm{a},\bm{0}) \ket{\phi}\\
	& = \frac{1}{|\mathcal{C}_2|}\sum_{\bm{\mu}}\sum_{\bm{\gamma}} A_{\bm{\mu},\bm{\gamma}} ~\epsilon_{(\bm{0},\bm{\bm{\mu} \oplus \bm{\gamma}})}E(\bm{0},\bm{\mu} \oplus \bm{\gamma}) s(\bm{a})  \ket{\phi},\numberthis \label{eqn:later_half}
	\end{align*}
	where $s(\bm{a})\coloneqq \sum_{\bm{a}\in \mathcal{C}_2} (-1)^{\bm{a} (\bm{\mu}\oplus\bm{\mu_0})^T} $.
% 	\eqref{eqn:prob_simp} follows from the fact $\epsilon_{(\bm{a},\bm{0})} E(\bm{a},\bm{0})  \in \mathcal{S}$. 
    Note that since $\bm{a} \in \mathcal{C}_2$ and $\bm{\mu} \oplus \bm{\mu_0} \in \F_2^n /\mathcal{C}_2^\perp$, the inner summation is nonzero only when $\bm{\mu}=\bm{\mu_0}$ so that
	\begin{align*}
	&\Pi_{\mathcal{S}_X(\mu_0)} U_Z \Pi_{\mathcal{S}_Z} \ket{\phi} =\\ &~~~~~~~~~~~\sum_{\bm{\gamma} \in \mathcal{C}_2^\perp/\mathcal{C}_1^\perp} A_{\bm{\mu_0},\bm{\gamma}}~ \epsilon_{(\bm{0},\bm{\mu_0}\oplus\bm{\gamma})}E(\bm{0},\bm{\mu_0} \oplus \bm{\gamma}) \ket{\phi}.\numberthis
	\end{align*}
	Similarly, we have 
	\begin{align*}
	&\bra{\phi} \Pi_{\mathcal{S}_Z} U_Z^\dagger \Pi_{\mathcal{S}_X(\mu_0)}  =\\ &~~~~~~~~~~~\bra{\phi}\sum_{\bm{\gamma} \in \mathcal{C}_2^\perp/\mathcal{C}_1^\perp} \overline{A_{\bm{\mu_0},\bm{\gamma}}}~ \epsilon_{(\bm{0},\bm{\mu_0} \oplus \bm{\gamma})}E(\bm{0},\bm{\mu_0}\oplus\bm{\gamma}) .\numberthis
	\end{align*}
	Thus, the probability of observing the syndrome $\bm{\mu}$ can be written as 
	\begin{align*}
	&p_{\bm{\mu}}\left(\ket{\phi}\right) 
% 	&= \bra{\phi}\sum_{\bm{\gamma},\bm{\gamma'}
% 	\in \mathcal{C}_2^\perp/\mathcal{C}_1^\perp} \overline{A_{\bm{\mu},\bm{\gamma}}}A_{\bm{\mu},\bm{\gamma'}} \epsilon_{(\bm{0},\bm{\mu} \oplus \bm{\gamma})}E(\bm{0},\bm{\mu} \oplus \bm{\gamma}) \epsilon_{(\bm{0},\bm{\mu} \oplus \bm{\gamma'})}E(\bm{0},\bm{\mu} \oplus \bm{\gamma'}) \ket{\phi}\\
	= \sum_{\bm{\gamma}}|A_{\bm{\mu},\bm{\gamma}}|^2 +\\
	&\sum_{\bm{\gamma} \neq \bm{\gamma'}}\overline{A_{\bm{\mu},\bm{\gamma}}} A_{\bm{\mu},\bm{\gamma'}}\bra{\phi} \epsilon_{(\bm{0},\bm{\gamma}\oplus\bm{\gamma'})} E(\bm{0},\bm{\gamma}\oplus\bm{\gamma'})\ket{\phi}. \numberthis \label{eqn:prob_mu}
	\end{align*}
	Note that only the second term depends on the initial state.    If some $\ket{\phi_i} \in \{\ket{+},\ket{-}\}$ in the initial state $\ket{\phi}=\ket{\overline{\phi_1\otimes\cdots\otimes\phi_k}}$, then the second term (the cross terms) in \eqref{eqn:prob_mu} vanishes since every $\epsilon_{(\bm{0},\bm{\gamma}\oplus\bm{\gamma'})} E(\bm{0},\bm{\gamma}\oplus\bm{\gamma'})$ with $\bm{\gamma} \neq \bm{\gamma'}$ is some nontrivial Pauli $Z$ logical.
	%and it makes $\bra{\phi}\epsilon_{(\bm{0},\bm{\gamma}\oplus\bm{\gamma'})} E(\bm{0},\bm{\gamma}\oplus\bm{\gamma'})\ket{\phi}= 0$. 
	Note that it follows from Theorem \ref{thm:GCs_prop} that $\sum_{\bm{\mu}} \sum_{\bm{\gamma}}|A_{\bm{\mu},\bm{\gamma}}|^2=1$. Since $\sum_{\bm{\mu}}p_{\mu}(\ket{\phi})=1$ for any initial state $\ket{\phi}\in \mathcal{V}(\mathcal{S})$, it follows that the sum of the second term over all the $X$-syndromes is 0, that is, \begin{equation}
	    \sum_{\bm{\mu}}\sum_{\bm{\gamma} \neq \bm{\gamma'}}\overline{A_{\bm{\mu},\bm{\gamma}}} A_{\bm{\mu},\bm{\gamma'}}\bra{\phi} \epsilon_{(\bm{0},\bm{\gamma}\oplus\bm{\gamma'})} E(\bm{0},\bm{\gamma}\oplus\bm{\gamma'})\ket{\phi}=0.
	\end{equation}
	
	Note that Pauli $Z$ logicals only change signs in the $\ket{0}\&\ket{1}$ basis. If the second term is the same for all $\ket{0}\&\ket{1}$ computational basis states in the codespace, then the probability of observing different syndromes is the same for different initial states $\ket{\phi}$. If not, the probabilities depend on the initial state, and encode the mutual information between initial state and syndrome measurement.
	%which means that there is a correlation between them and after measurement, their mutual information is exploited. 
	In these circumstances, 
	%the quantum error-correction conditions are violated, that is,
	we cannot find a recovery operator for $U_Z$ that is good for the entire codespace. An important special case is when a decoherence-free subspace is embedded in the codespace (useful for passive control of coherent errors $U_Z=R_Z(\theta)$).
	
	We now introduce two examples to illustrate how \eqref{eqn:prob_mu} provides insight into invariance of the codespace, the probability of success in magic state distillation, and existence of an embedded decoherence-free subspace.
	Continuing Example \ref{examp1_in_intro} below, we compute the probabilities of observing different syndromes for the $\llbr 7,1,3 \rrbr $ Steane code and discuss implications. Continuing Example \ref{examp10} below, we demonstrate that by changing signs of $Z$-stabilizers in the $\llbr 4,2,2 \rrbr $ code, we can switch from the case where the second term is the same for every initial state to the case of an embedded decoherence-free subspace.
% 	\begin{align}
% 	p_{\bm{\mu_0}} 
% 	&= \bra{\phi}\sum_{\bm{\gamma},\bm{\gamma'}
% 	\in \mathcal{C}_2^\perp/\mathcal{C}_1^\perp} \overline{A_{\bm{\mu_0},\bm{\gamma}}(\theta)}A_{\bm{\mu_0},\bm{\gamma'}}(\theta) \epsilon_{(\bm{0},\bm{\mu_0} \oplus \bm{\gamma})}E(\bm{0},\bm{\mu_0} \oplus \bm{\gamma}) \epsilon_{(\bm{0},\bm{\mu_0} \oplus \bm{\gamma'})}E(\bm{0},\bm{\mu_0} \oplus \bm{\gamma'}) \ket{\phi}\\
% 	&= \sum_{\gamma}|A_{\bm{\mu_0},\bm{\gamma}}(\theta)|^2 + \sum_{\bm{\gamma} \neq \bm{\gamma'}}\overline{A_{\bm{\mu_0},\bm{\gamma}}(\theta)} A_{\bm{\mu_0},\bm{\gamma'}}(\theta)\bra{\phi} \epsilon_{(\bm{0},\bm{\gamma}\oplus\bm{\gamma'})} E(\bm{0},\bm{\gamma}\oplus\bm{\gamma'})\ket{\phi}. \label{eqn:prob_mu}
% 	\end{align}
% 	For one-qubit logical state, if $\ket{\phi}=\ket{\bar{+}}$ or $\ket{\bar{-}}$, then $\bra{\phi}\epsilon_{(0,\gamma+\gamma')}E(0,\gamma+\gamma')\ket{\phi}= 0$ for $\gamma\neq \gamma'$, which means the second term vanishes and we have
% 	\begin{equation}
% 	1= \sum_{\mu}p_{\mu} = \sum_{\mu} \sum_{\gamma} |A_{\mu,\gamma}(\theta)|^2,
% 	\end{equation}
% 	which verifies a special case of our Lemma \ref{lemma:sumisone}.
% 	However, for different initial state $\ket{\phi}\neq\ket{\bar{+}}$ and $\ket{\phi}\neq\ket{\bar{-}}$, the probability could depend on $\ket{\phi}$. 
% 	For multi-qubits system, the second term vanishes individually when
% 	$\ket{\phi}=\ket{\overline{\phi_1\otimes\cdots\otimes\phi_k}}$ with $\ket{\phi_i} \in \{\ket{+},\ket{-}\}$.
% 	Let's compute several complete examples to illustrate \eqref{eqn:prob_mu}.
	\addtocounter{example}{-1}
	\begin{example}[continued]
% 	\label{examp9}
		\normalfont
	   % (Probabilities of observing syndromes for Steane Code under $R_Z(\theta)$). 
		The Steane $\llbr 7,1,3 \rrbr $ code has only one logical qubit, and we let $\ket{\overline{0}}$, $\ket{\overline{1}}$ denote the the two computational basis states. Given a syndrome $\bm{\mu}$, we observe that one of the generator coefficients $A_{\bm{\mu},\bm{\gamma}=\bm{0}}(\theta)$, $A_{\bm{\mu},\bm{\gamma}\neq \bm{0}}(\theta)$, is real and the other is purely imaginary, so that the crossterms vanish in \eqref{eqn:prob_mu}. 
		Hence, the probabilities of observing different syndromes are constant for different initial states and are given by
		\begin{align*}
		p_{\bm{\mu}=\bm{0}}(\ket{\overline{0}},\theta)  = p_{\bm{\mu}=\bm{0}}(\ket{\overline{1}},\theta) 
% 		&= |A_{\bm{0},\bm{0}}(\theta)|^2 + |A_{\bm{0}, \bm{1}}(\theta)|^2 \\
		&= \frac{1}{32} \left(7\cos 4\theta + 25\right), \\
% 		\end{align*}
% 		and
% 		\begin{align*}
		p_{\bm{\mu} \neq \bm{0}}(\ket{\overline{0}},\theta)  = p_{\bm{\mu} \neq \bm{0}}(\ket{\overline{1}},\theta) 
% 		&=\sum_{\bm{\gamma}\in\{\bm{0},\bm{1}\}}|A_{\bm{\mu} \neq \bm{0},\bm{0}}(\theta)|^2 \\ 
		%+ |A_{\bm{\mu} \neq \bm{0}, \bm{1}}(\theta)|^2
		&= \frac{1}{32} \left(1-\cos 4\theta \right). \numberthis \label{eqn:steane_probs}
		\end{align*}
		It is not hard to verify that $\sum_{\bm{\mu}} p_{\bm{\mu}}(\ket{\phi},\theta) = \frac{1}{32} \left(7\cos 4\theta + 25\right) + \frac{7}{32} \left(1-\cos 4\theta \right) = 1$ for all $\ket{\phi} \in \mathcal{V}(\mathcal{S})$ and for all $\theta$.
		Figure \ref{fig:7_1_3} plots the probability of observing the trivial syndrome as a function of the rotation angle. 
		\begin{figure}[h!]
			\centering
			\includegraphics[scale=0.2]{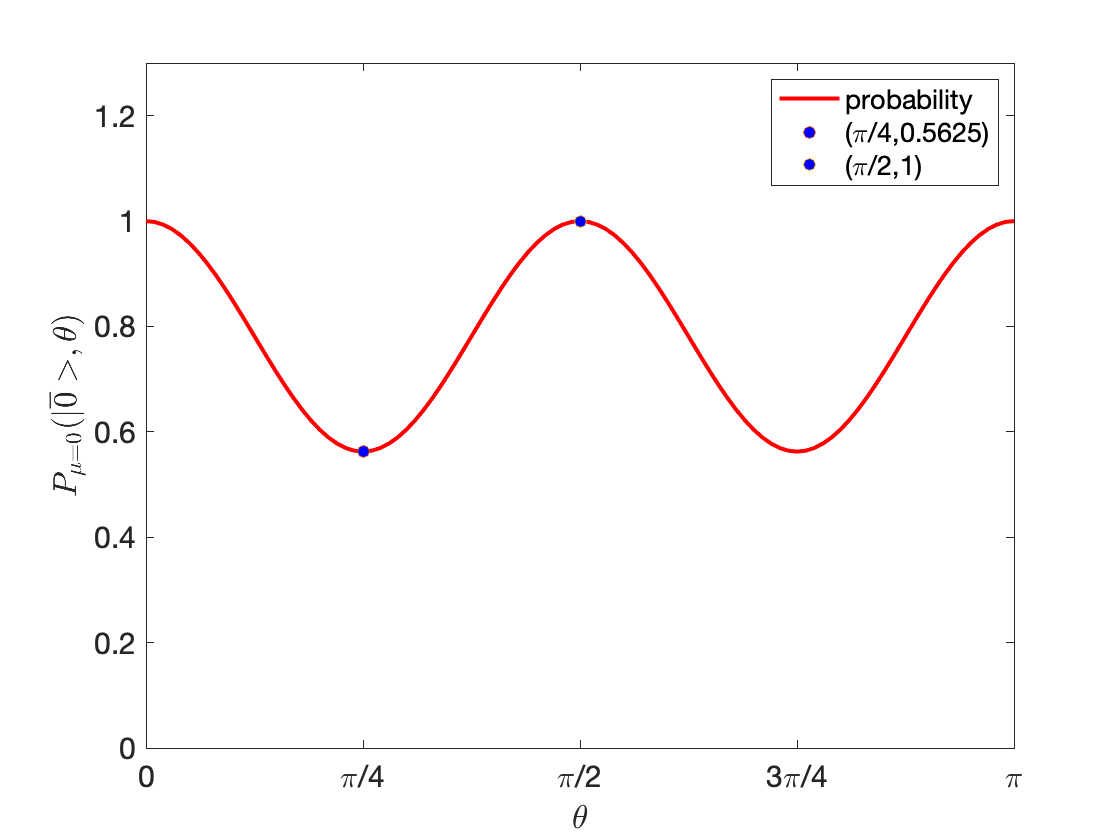}
			%[width=4in,height=2in]
			\caption{The probability of observing the trivial syndrome for the Steane Code under $R_Z(\theta)$ for varying physical angles $\theta$.}
			\label{fig:7_1_3}
% 			\caption{$p_{\mu=0}(\ket{\overline{0}},\theta)= \frac{1}{32} \left(7\cos 4\theta + 25\right) $ for the $\llbr 7,1,3 \rrbr $ Steane Code}
		\end{figure}
		
		We observe from Figure \ref{fig:7_1_3} that when $\theta$ is a multiple of $\frac{\pi}{2}$, $p_{\bm{\mu}=\bm{0}}(\ket{\phi})=1$ for all the states $\ket{\phi}$ in the Steane codespace $\mathcal{V}(\mathcal{S})$, which implies that $R_Z(\frac{k\pi}{2})$ preserves $\mathcal{V}(\mathcal{S})$. The angle $\theta = \frac{\pi}{4} + \frac{k\pi}{2}$ minimizes the probability of obtaining the zero syndrome and this minimum value relates to the probability of success in magic state distillation. Substituting $\theta = \frac{\pi}{4}$, we obtain  
		$
		p_{\bm{\mu}=\bm{0}}\left(\ket{\phi},\frac{\pi}{4}\right)  = \frac{9}{16} \text{, and } 
		p_{\bm{\mu} \neq\bm{0}}\left(\ket{\phi},\frac{\pi}{4}\right)  = \frac{1}{16},
		$
		for all $\ket{\phi} \in \mathcal{V}(\mathcal{S})$.
	
% 		We see that the symmetry happens at $\theta=\frac{\pi}{4}$, which can be used in the magic state distillation even if the codespace is in general not invariant under $R_Z(\frac{\pi}{4})$. Also, we see that when $\theta <0.3$, the logical rotation angle keeps small $\theta_L<0.05$ while after that there is a range that $\theta_L$ changes rapidly with $\theta$. For more details about the change of rates, we could analyze its first and second derivatives.
	\end{example}
% 	\begin{remark}\normalfont
	  %For general one-logical-qubit codes, if all the pairs of generator coefficients associated with same syndrome consists of a pure real number and a pure imagrinary number, then the second term in \eqref{eqn:prob_mu} vanish and the probabilities are consistent. We consider some multiple-logical-qubit codes in the following example.
% 	\end{remark}
	\begin{example}[continued]
	\label{examp10}
		\normalfont
		%(The influence of signs on the $\llbr 4,2,2 \rrbr $ code)
		%(Revisit of the $\llbr 4,2,2 \rrbr $ codes and the Embedded Decoherence-Free Subspace). 
		Recall the $\llbr 4,2,2 \rrbr $  CSS($X,\mathcal{C}_2 = \{\bm{0},\bm{1}\};Z,\mathcal{C}_1^\perp = \mathcal{C}_2$) code with two different choices of signs defined by the character vectors $\bm{y} = \bm{0}$ (all positive signs), and $\bm{y'} = [0,0,0,1]$ (negative $Z^{\otimes 4}$ in the stabilizer group). 
		
		\begin{table}[h!]
		    \centering
		    \caption{Generator coefficients $A_{\bm{\mu},\bm{\gamma}}(\theta)$ for $R_Z(\theta)$ of the $\llbr 4,2,2 \rrbr $ code with all positive signs ($\bm{y} = \bm{0}$).}
		    \renewcommand{\arraystretch}{1.3} 
		    \begin{tabular}{|c|c|c|}
		       \hline
		        %\diagbox{$X$-syndromes}{$Z$-logicals} 
                &
		        $\bm{\gamma}=\bm{0}$ & $\bm{\gamma}\neq \bm{0}$
		        \\
		        \hline
		          $\bm{\mu}=\bm{0}$ & 
		          $\frac{1}{4}\left(\cos 2\theta + 3\right)$ & 
		          $\frac{1}{4}\left(\cos 2\theta -1\right)$\\
		          \hline
		        $\bm{\mu}=[1,0,0,0]$ & \multicolumn{2}{c|}{$-\frac{1}{4}\imath \sin 2\theta$} \\
		        \hline 
		    \end{tabular}
		    \label{tab:4_2_2_gcs+}
		\end{table}
		Table \ref{tab:4_2_2_gcs+} lists the generator coefficients for all positive signs ($\bm{y}=\bm{0}$). We now use the data to calculate the probabilities of observing different syndromes as described in \eqref{eqn:prob_mu}. For the encoded $\ket{\overline{00}}$ state, we have 
		\begin{align*}
		p_{\bm{\mu} =\bm{0}}(\ket{\overline{00}},\theta) 
% 		&= \frac{1}{8}(\cos 4\theta + 7)  + \frac{3}{8}(\cos 4\theta - 1) \\
		&= \frac{1}{2}\cos 4\theta + \frac{1}{2},\label{eqn:prob_4_2_2+}  \\ 
% 		\end{align*}
% 		and 
% 		\begin{align*}
		p_{\bm{\mu}=[0,0,0,1]}(\ket{\overline{00}},\theta) 
% 		&=  \frac{1}{8}(1-\cos 4\theta)  - \frac{3}{8}(\cos 4\theta - 1)\\
		&= -\frac{1}{2}\cos 4\theta + \frac{1}{2}.\numberthis
		\end{align*}
		%%%%%%%%%%%%%%%%%%%%%%%%%%%%%%%%%%%%%%%%%%%%%%%%%%%%%%%%%%
		\begin{table*}[ht!]
		    \centering
		    \caption{Generator coefficients $A_{\bm{\mu},\bm{\gamma}}(\theta)$ for $R_Z(\theta)$ of the $\llbr 4,2,2 \rrbr $ code with negative $Z^{\otimes 4}$ stabilizer ($\bm{y} = [0,0,0,1]$).}
		    \renewcommand{\arraystretch}{1.3} 
			\begin{tabular}{|c|c|c|c|c|}
		    \hline
		         \diagbox{$X$-syndromes}{$Z$-logicals} &
		        $\bm{\gamma}=\bm{0}$ 
		        & $\bm{\gamma_1} = [0,0,1,1]$
		        & $\bm{\gamma_2} = [0,1,1,0]$
		        & $\bm{\gamma_3} = \bm{\gamma_1} \oplus \bm{\gamma_2} $\\
		        \hline
		          $\bm{\mu}=\bm{0}$ & 
		          $\cos \theta $ & 0 & 0 & 0 \\
		          \hline
		        $\bm{\mu}=[1,0,0,0]$ & 
		        $-\frac{1}{2}\imath \sin \theta$ &
		        $\frac{1}{2}\imath \sin \theta$ &
		        $-\frac{1}{2}\imath \sin \theta$ &
		        $-\frac{1}{2}\imath \sin \theta$
		        \\
		        \hline
		    \end{tabular}
		   \label{tab:4_2_2_gcs-}
		\end{table*}
		%%%%%%%%%%%%%%%%%%%%%%%%%%%%%%%%%%%%%%%%%%%%%%%%%%%%%%%%%%
		The remaining three states have the same probabilities of observing $X$-syndromes:
		\begin{align*}
		&p_{\bm{\mu}=\bm{0}}(\ket{\phi}\in \{\ket{\overline{01}},\ket{\overline{10}},\ket{\overline{11}} \},\theta) = \\
% 		&=
% 		p_{\bm{\mu}=\bm{0}}(,\theta) =	
% 		p_{\bm{\mu}=\bm{0}}(,\theta) \\
		& ~~~~~~~~~~~~~~
		\frac{1}{8}(\cos 4\theta + 7) +\frac{1}{8}\left(1-\cos 4\theta\right) = 1,\numberthis \\
% 		\end{align*}
% 		and
% 		\begin{align*}
		&p_{\bm{\mu} = [1,0,0,0]}(\ket{\phi}\in \{\ket{\overline{01}},\ket{\overline{10}},\ket{\overline{11}} \},\theta) =\\	
% 		p_{[1,0,0,0]}(\ket{\overline{10}},\theta) =
% 		p_{[1,0,0,0]}(\ket{\overline{11}},\theta) \\
		& ~~~~~~~~~~~~~~
		\frac{1}{8}(1-\cos 4\theta)  -\frac{1}{8}\left(1-\cos 4\theta\right) = 0. \numberthis
		\end{align*}
		If the initial state is among $\ket{\overline{01}}, \ket{\overline{10}}, \ket{\overline{11}}$, then it evolves within the codespace for all angles $\theta$, which implies that $\mathcal{F} \coloneqq \mathrm{span}(\ket{\overline{01}},\ket{\overline{10}},\ket{\overline{11}})$ forms a embedded decoherence-free subspace (DFS) inside the codespace \cite{coherent_noise}. 
		
		 \begin{figure}[h!]
			\centering
			\includegraphics[scale=0.2]{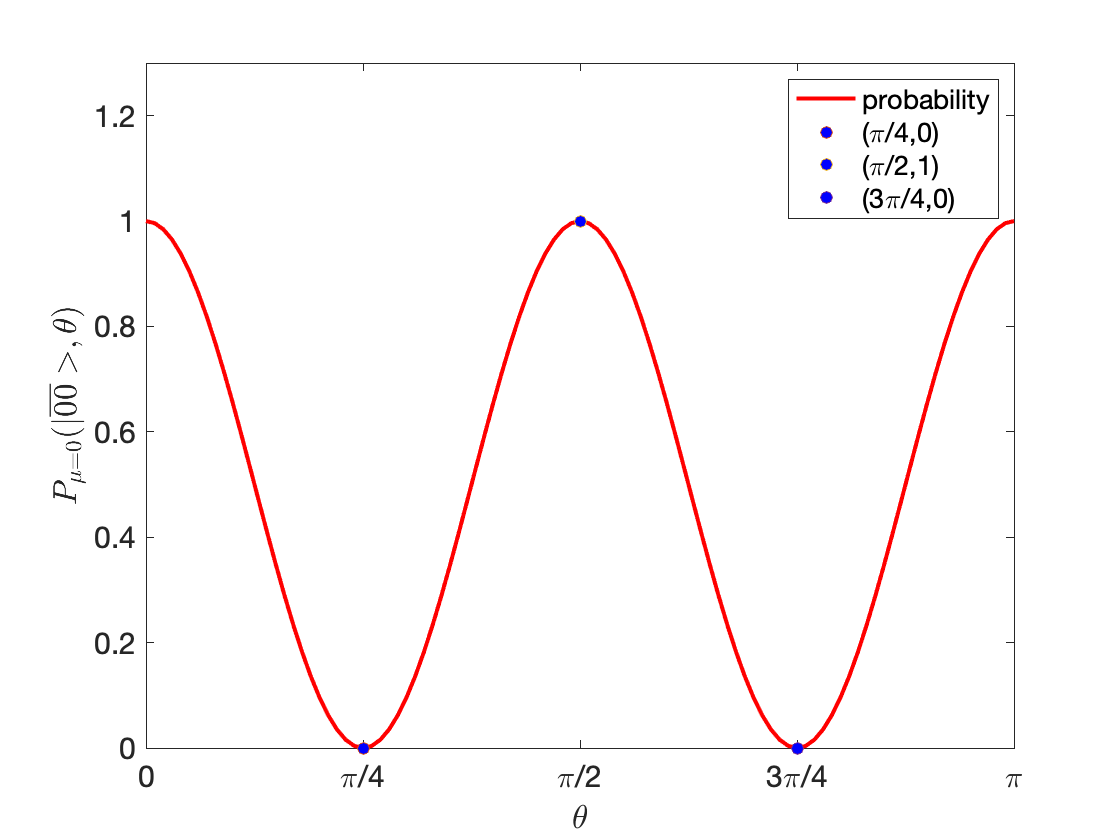}
			\caption{The $\llbr 4,2,2 \rrbr $ code with all positive stabilizers. The probability of observing the trivial syndrome for the initial encoded state $\ket{\overline{00}}$ under $R_Z(\theta)$ for varying physical angles $\theta$.}
			\label{fig:4_2_2+}
% 			{$p_{\mu=0}(\ket{\overline{00}},\theta)=\frac{1}{2}\cos 4\theta+\frac{1}{2}$ for the $\llbr 4,2,2 \rrbr $ CSS code}
		\end{figure}
		Figure \ref{fig:4_2_2+} plots \eqref{eqn:prob_4_2_2+} for different physical angles $\theta$. 
% 		 the probability of observing the trivial syndrome under the $\ket{\overline{00}}$ encoded state respect to varying rotation angles $\theta$. 
        When $\theta=\frac{\pi}{4}+\frac{k\pi}{2}$ for some integer $k$, syndrome measurement acts as projection from $\mathcal{V}(\mathcal{S})$ to the embedded DFS, and we are able to learn whether the initial state was $\ket{\overline{00}}$; When $\theta=\frac{k\pi}{2}$ for some integer $k$, the measurement outcome is always the zero syndrome, which implies that $R_Z(\frac{\pi}{2})$ perserve the codespace and some logical operator is induced. The Kraus operators derived in \eqref{eqn:kraus_ops} imply that the induced logical operator is
        \begin{align*}
            B_{\bm{\mu} = \bm{0}} \left(\frac{\pi}{2}\right) 
            &= \sum_{\bm{\gamma}}A_{\bm{0},\bm{\gamma}}\left(\frac{\pi}{2}\right)E(\bm{0},\bm{\gamma})\\ %\\
            % &= \frac{1}{2}E(\bm{0},\bm{0}) - \frac{1}{2}E(\bm{0},\bm{\gamma_1}) - \frac{1}{2}E(\bm{0},\bm{\gamma_2}) - \frac{1}{2}E(\bm{0},\bm{\gamma_1 \oplus \gamma_2})\\
            % &= \frac{1}{2}\bar{I}\otimes\bar{I} - 
            % \frac{1}{2} \bar{I}\otimes\bar{Z} - 
            % \frac{1}{2} \bar{Z}\otimes\bar{I} -
            % \frac{1}{2} \bar{Z}\otimes\bar{Z} \\
            &\equiv \left(\bar{Z} \otimes \bar{Z}\right) \circ \overline{CZ}.\numberthis
        \end{align*}
% 	We can think about the $\llbr 2^k,k,2 \rrbr $ code family and other generalization for the embedded DFS and its value.
        Next, we compute the generator coefficients for the same $\llbr 4,2,2 \rrbr $ code but with nontrivial signs (character vector $\bm{y}=[0,0,0,1]$).
	
     It follows from \eqref{eqn:prob_mu} and Table \ref{tab:4_2_2_gcs-} that 
        \begin{align*}
    	&p_{\bm{\mu}=\bm{0}}(\ket{\phi} \in \{\ket{\overline{00}}, \ket{\overline{01}}, \ket{\overline{10}}, \ket{\overline{11}}\},\theta) =(\cos\theta)^2,\\
% 		\end{align*}
%       and
%     	\begin{align*}
    	&p_{\bm{\mu} = [1,0,0,0]}(\ket{\phi} \in \{\ket{\overline{00}}, \ket{\overline{01}}, \ket{\overline{10}}, \ket{\overline{11}}\},\theta) \\
    	&~~~~~~~~~~~~~~~~~~~~~~~~~~~~~~~~~~~~~~~~~~~~~~~= (\sin\theta)^2.
		\end{align*}
		In this case, the probabilities are independent of the different initial states and there is no embedded decoherence-free subspace in the codespace. This example shows that for the same code, state evolution depends very strongly on signs of $Z$-stabilizers. 
	\end{example}
    In prior work \cite{coherent_noise}, we have derived criteria that ensure a stabilizer code is a DFS, and \eqref{eqn:prob_mu} opens the door to developing criteria for embedded DFS, in which the second term acts as an amendment to the first term and implies the probability is either $0$ or $1$ for a subset of initial $\ket{0}\&\ket{1}$-basis state in the codespace.
    %The criteria that ensuring stabilizer codespace inside some DFS has been studied in \cite{coherent_noise} and it is possible to take advantage of \eqref{eqn:prob_mu} to find a general conditions for existence of embedded DFS. We leave it to furture work. 
	
	%	\begin{example}
	%		\normalfont
	%		Consider the $\llbr 9,1,3 \rrbr $ Shor code with all positive signs.
	%	\end{example}
	
	\section{CSS codes Preserved by $U_Z$}
	\label{sec:nece_suff_cond_invariant}
	When a CSS code is preserved by a unitary $U_Z$, the probability of observing the zero syndrome is 1, and the Kraus operators capture evolution of logical states. Theorem \ref{thm:preserved_by_Uz} provides necessary and sufficient conditions for a unitary $U_Z$ to preserve a CSS code.
	
	We prove Theorem \ref{thm:preserved_by_Uz} by writing $\Pi_{\mathcal{S}}$ as a product $\Pi_{\mathcal{S}}=\Pi_{\mathcal{S}_X}\Pi_{\mathcal{S}_Z}$, where $U_Z$ commutes with the $Z$-projector $\Pi_{\mathcal{S}_Z}$, and we then translate commutativity to conditions on generator coefficients. We generalize these conditions to arbitrary stabilizer codes in Appendix \ref{sec:stab_gcf}. 
	
		\begin{theorem}\label{thm:preserved_by_Uz}
		Let CSS($X,\mathcal{C}_2=\langle  c_i : 1\le i\le k_2 \rangle;Z,\mathcal{C}_1^\perp=\langle  d_j : 1\le j\le n-k_1 \rangle$) be an $\left[\left[n, k_1-k_2, d\right]\right]$ CSS code $\mathcal{V}(\mathcal{S})$ defined by the stabilizer group $\mathcal{S}$ with code projector $\Pi_{\mathcal{S}}$. Then the unitary $U_Z=\sum_{\bm{v}\in \F_2^n}f(\bm{v})E(\bm{0},\bm{v})$ preserves $\mathcal{V}(\mathcal{S})$ (i.e. $U_Z\Pi_\mathcal{S} U_Z^\dagger=\Pi_\mathcal{S}$) if and only if 
		\begin{equation}\label{eqn:preserved_by_Uz}
		\sum_{\bm{\gamma}\in \mathcal{C}_2^\perp/\mathcal{C}_1^\perp} |A_{\bm{0},\bm{\gamma}}|^2=1.
		\end{equation}
	\end{theorem}
	\begin{proof}
	    See Appendix \ref{subsec:proof_thm_preserved_by_Uz}.
	\end{proof}

	\begin{remark}[Logical Operator induced by $U_Z$]
		\normalfont
		\label{rem:log_op_U_Z}
		We assume that $U_Z\Pi_{\mathcal{S}} U_Z^\dagger = \Pi_{\mathcal{S}}$ for a CSS code defined by $\mathcal{S}$. By Theorem \ref{thm:preserved_by_Uz}, \eqref{eqn:preserved_by_Uz} holds, so that by Theorem \ref{thm:GCs_prop} we only have one Kraus operator left in \eqref{eqn:kraus_ops} that is given by
		\begin{equation}
		    B_{\bm{\mu} =\bm{0}} = \sum_{\bm{\gamma}\in \mathcal{C}_2^\perp / \mathcal{C}_1^\perp} A_{\bm{0},\bm{\gamma}} ~\epsilon_{(\bm{0},\bm{\gamma})}
		    E(\bm{0},\bm{\gamma}).
		\end{equation}
		Note that $\F_2^k \simeq \mathcal{C}_2^\perp / \mathcal{C}_1^\perp$ and we have a bijective map $g: \F_2^k \to \mathcal{C}_2^\perp / \mathcal{C}_1^\perp $ defined by $g(\bm{\alpha}) =\bm{\alpha} G_{\mathcal{C}_2^\perp /\mathcal{C}_1^\perp}$, where $G_{\mathcal{C}_2^\perp /\mathcal{C}_1^\perp}$ is the generator matrix selected. Let $U_Z^L$ be the logical operator induced by $U_Z$, and let $\alpha_j$ be the $j$th entry of the vector $\bm{\alpha}$. Then, using \eqref{eqn:css_log_paulis}, we translate the Kraus operator into the logical space as
		\begin{align*}
		U_Z^L 
		&= \sum_{\bm{\alpha}\in \F_2^k} A_{\bm{0}, g(\bm{\alpha})} \left(\prod_{j=1}^k  \left(Z_j^L\right)^{\alpha_j}\right)\\
		&= \sum_{\bm{\alpha}\in \F_2^k} A_{\bm{0}, g(\bm{\alpha})} E(\bm{0},\bm{\alpha}),\numberthis \label{eqn:logical_final}
		\end{align*} 
		Thus, if a CSS code is preserved by $U_Z=\sum_{\bm{v}\in \F_2^n}f(\bm{v})E(\bm{0},\bm{v})$, then the generator coefficients corresponding to the zero syndrome %$\bm{\mu}=\bm{0}$ 
		are simply the coefficients in the Pauli expansion of the induced logical operator. We also observe that $U_Z^L$ given in \eqref{eqn:logical_final} is unitary if and only if \eqref{eqn:preserved_by_Uz} holds.
	\end{remark}
	
	In the following subsections, we simplify  \eqref{eqn:preserved_by_Uz} in special cases when $U_Z$ is a QFD gate, and when $U_Z = R_Z(\frac{\pi}{p})$ for some integer $p$. We then provide necessary and sufficient conditions for quantum Reed-Muller codes to be preserved by $R_Z(\frac{2\pi}{2^l})$, and connect to the conditions in \cite[Theorem 17]{rengaswamy2020optimality}. 
	
	    \subsection{QFD Gates}
    Theorem \ref{thm:div_cond_qfd} below specializes Theorem \ref{thm:preserved_by_Uz} to the broad class of diagonal level-$l$ QFD gates $\tau_R^{(l)}$ determined by symmetric matrices $R\in \Z_{2^{l}}^{n\times n}$. Note that Theorem \ref{thm:div_cond_qfd} applies to CSS codes with arbitrary signs and $R_Z\left(\frac{2\pi}{2^l}\right)$ form a subset of QFD gates. Theorem \ref{thm:div_cond_qfd} includes the divisibility conditions derived in \cite{zeng2011transversality,landahl2013complex,vuillot2019quantum} as a special case. 
    %We now simplify Theorem \ref{thm:preserved_by_Uz} to obtain Theorem \ref{thm:div_cond_qfd} in the special case where $U_Z$ is a level-$l$ QFD gate determined by a symmetric matrix $R\in \Z_{2^{l}}^{n\times n}$. Special cases of Theorem \ref{thm:div_cond_qfd} include the divisibility conditions introduced in \cite{landahl2013complex,haah2018towers} as we consider a board class of diagonal gates and include the freedom of signs.
    \begin{theorem}\label{thm:div_cond_qfd}
     	Consider a CSS($X,\mathcal{C}_2;Z,\mathcal{C}_1^\perp$)%($X,\mathcal{C}_2=\langle  c_i : 1\le i\le k_2 \rangle;Z,\mathcal{C}_1^\perp=\langle  d_j : 1\le j\le n-k_1 \rangle$)
     	 code, where $\bm{y}$ is the character vector of the $Z$-stabilizers. Then, a QFD gate $\tau_R^{(l)}= \sum_{\bm{v}\in \F_2^n} \xi_l^{\bm{v}R\bm{v}^T \bmod {2^l}} \ket{\bm{v}}\bra{\bm{v}}$ 
         preserves the codespace $\mathcal{V}(\mathcal{S})$
		if and only if
		\begin{equation} \label{eqn:div_cond_qfd}
		2^l \mid (\bm{v_1}R\bm{v_1}^T - \bm{v_2}R\bm{v_2}^T )
		\end{equation}
		for all $\bm{v_1}, \bm{v_2} \in \mathcal{C}_1+\bm{y}$ such that $\bm{v_1} \oplus \bm{v_2} \in \mathcal{C}_2$. 
    \end{theorem}
    \begin{proof}
    It follows from \eqref{eqn:qfd_gcs_simp} that
    \begin{align}\label{eqn:qfd_gcs_sz_sq}
   &\sum_{\bm{\gamma} \in \mathcal{C}_2^\perp/\mathcal{C}_1^\perp} \left|A_{\bm{0},\bm{\gamma}}(R,l)\right|^2 = \nonumber\\ &~~~~~~~~~~~~~~~
   \frac{1}{|\mathcal{C}_1|^2}\sum_{\bm{v}\in \mathcal{C}_1} s(\bm{v},\bm{y}) \sum_{{\gamma}\in \mathcal{C}_2^\perp/\mathcal{C}_1^\perp} (-1)^{\bm{\gamma}\bm{v}^T} ,
       % \left|A_{\bm{0},\bm{\gamma}}(R,l)\right|^2 = \frac{1}{|\mathcal{C}_1|^2}\sum_{\bm{v}\in \mathcal{C}_1} (-1)^{\bm{\gamma}\bm{v}^T} \sum_{\substack{\bm{v_1},\bm{v_2}\in \mathcal{C}_1+\bm{y},\\ \bm{v_1}\oplus\bm{v_2}=\bm{v}}} \xi_l^{\bm{v_1}R\bm{v_1}^T - \bm{v_2}R\bm{v_2}^T \bmod{2^l}}.
    \end{align}
    where 
    \begin{equation} 
    s(\bm{v},\bm{y}) \coloneqq \sum_{\bm{v_1}\in \mathcal{C}_1+\bm{y}} \xi_l^{\bm{v_1}R\bm{v_1}^T - (\bm{v}\oplus \bm{v_1})R(\bm{v}\oplus \bm{v_1})^T \bmod{2^l}}. 
    \end{equation}
    We simplify \eqref{eqn:preserved_by_Uz} using \eqref{eqn:qfd_gcs_sz_sq} to obtain
        \begin{align*}
		1&= \sum_{{\gamma}\in \mathcal{C}_2^\perp/\mathcal{C}_1^\perp} \left|A_{\bm{0},\bm{\gamma}}(R,l)\right|^2\\
		 &=  \frac{1}{|\mathcal{C}_1|^2}\sum_{\bm{v}\in \mathcal{C}_1}  s(\bm{v},\bm{y})
% 		 \sum_{\bm{v_1}\in \mathcal{C}_1+\bm{y}} \xi_l^{\bm{v_1}R\bm{v_1}^T - (\bm{v}\oplus \bm{v_1})R(\bm{v}\oplus \bm{v_1})^T \bmod{2^l}}
		 \sum_{\bm{\gamma}\in \mathcal{C}_2^\perp/\mathcal{C}_1^\perp} (-1)^{\bm{\gamma}\bm{v}^T}\\
		 &= \frac{\sum_{\bm{v}\in \mathcal{C}_2} 
% 		 s(\bm{v},\bm{y})
		 \sum_{\bm{v_1}\in \mathcal{C}_1+\bm{y}} \xi_l^{\bm{v_1}R\bm{v_1}^T - (\bm{v}\oplus \bm{v_1})R(\bm{v}\oplus \bm{v_1})^T }}{|\mathcal{C}_1||\mathcal{C}_2|},\numberthis
		\end{align*}
		which requires each term to contribute $1$ to the summation. We complete the proof by setting $\bm{v_2}=\bm{v}\oplus \bm{v_1}$.
    \end{proof}
    \begin{remark}
        \normalfont
        \label{rem:trans_theta_as_sp_cs_in_qfd}
        When $R=I$, then $\bm{v}R\bm{v}^T = w_H(\bm{v})$ and the divisibility condition simplifies to the condition previously obtained for $R_Z\left(\frac{2\pi}{2^l}\right)$. If a CSS code is preserved by $R_Z\left(\frac{2\pi}{2^l}\right)$ for all $l\ge 1$, then it follows \eqref{eqn:div_cond_qfd} that for any fixed $\bm{w}\in \mathcal{C}_1\ \mathcal{C}_2$, all elements in the coset $\mathcal{C}_2+\bm{w}+\bm{y}$ have the same Hamming weight. It then follows from the generalized encoding map given in \eqref{eqn:gen_encode_map} that any CSS code invariant under $R_Z\left(\frac{2\pi}{2^l}\right)$ for all $l\ge 1$ is a constant-excitation code \cite{zanardi1997noiseless}. 
       \end{remark}
        %If we think that the quadratic form $\bm{v}R\bm{v}^T$ as a generalized version of weight, then the divisibility conditions for QFD gates 
        We now explore the influence of signs by analyzing and separating the effect of the character vector $\bm{y}$.
        \begin{lemma}\label{lemma:two_div_conds_qfd}
        		Consider a CSS($X,\mathcal{C}_2;Z,\mathcal{C}_1^\perp$) code, where $\bm{y}$ is the character vector of the $Z$-stabilizers. Then, \eqref{eqn:div_cond_qfd} holds for all $\bm{v_1}, \bm{v_2} \in \mathcal{C}_1+\bm{y}$ such that $\bm{v_1} \oplus \bm{v_2} \in \mathcal{C}_2$ if and only if 
        		   %\begin{enumerate}[1)]
        			 \begin{equation}\label{eqn:first_div_cond_qfd}
        			2^l \mid (\bm{v_1}R\bm{v_1}^T - \bm{v_2}R\bm{v_2}^T), \text{ for all } \bm{v_1},\bm{v_2}\in \mathcal{C}_2 +\bm{y};
        			 \end{equation}
        			\begin{equation}\label{eqn:sec_div_cond_qfd}
        			2^{l-1}  \mid   (\bm{u_1}-\bm{u_2}) R \bm{w}^T, 
        			\end{equation}
        			for all $\bm{u_1},\bm{u_2}\in\mathcal{C}_2$ and $\bm{w}\in \mathcal{C}_1 / \mathcal{C}_2.$
        		%\end{enumerate}
        \end{lemma}
        \begin{proof}
            See Appendix \ref{subsec:proof_lemma_qfd_div}.
        \end{proof}
        Note that only \eqref{eqn:first_div_cond_qfd} depends on the character vector $\bm{y}$, and its contribution is moving the divisible requirement for a set to that for a coset. 
        
        Note that by varying the level $l$, the same symmetric matrix $R$ can determine different gates (for example, the gates C$Z$ and C$P$ in Example \ref{examp5}). The divisibility conditions corresponding to successive levels differ by a factor of $2$. This suggests using concatenation to lift a code preserved by a level $l$ QFD gate determined by $R$ to a code preserved by a level $l+1$ QFD gate determine by $I_2\otimes R$. We defer investigation to future work.
        %We make this explicit in the following theorem.
  
	\subsection{Transversal $\theta$ $Z$-Rotation $R_Z(\theta)$}
	\label{subsec:nece_suff_cond_invariant_R_Z}
	\subsubsection{$R_Z(\pi/p)$ and RM Constructions}
	   If the physical rotation angle $\theta$ is a fraction of $\pi$, then the constraint on generator coefficients in \eqref{eqn:preserved_by_Uz} is equivalent to conditions on the Hamming weights that appear in the classical codes $\mathcal{C}_1$ and $\mathcal{C}_2$ that determine the quantum CSS code. 
	\begin{theorem}\label{thm:div_cond_RZ}
	    Let $p\in\Z$. Then $R_Z\left(\frac{\pi}{p}\right)$ preserves the  CSS($X,\mathcal{C}_2;Z,\mathcal{C}_1^\perp$) codespace %$\mathcal{V}(\mathcal{S})$
		 %Given $p\in\Z$, $\sum_{\bm{\gamma}\in \mathcal{C}_2^\perp/\mathcal{C}_1^\perp} \left|A_{\bm{0},\bm{\gamma}}\left(\frac{\pi}{p}\right)\right|^2=1$ 
		 if and only if 
		 \begin{equation}\label{eqn:div_cond_invariant}
		  2p ~|~ \left(w_H(\bm{w})-2w_H(\bm{w}*\bm{z})\right),
		 \end{equation}
		for all $\bm{w}\in \mathcal{C}_2$ and all $\bm{z}\in \mathcal{C}_1+\bm{y}$, where $\bm{y}$ is the character vector that determines signs of $Z$-stabilizers and $\bm{w}*\bm{z}$ is the coordinate-wise product of $\bm{w}$ and $\bm{z}$. 
	\end{theorem}
	\begin{proof}
	    See Appendix \ref{subsec:proof_div_Rz}.
	\end{proof}

% 	\begin{corollary}\label{coro:div_cond}
% 		Let $p\in\Z$. $R_Z\left(\frac{\pi}{p}\right)$ preserves CSS codespace $\mathcal{V}(\mathcal{S})$ if and only if 
% 		$2p ~|~ (w_H(\bm{w})-2w_H(\bm{w}*\bm{z}))$
% 		for all $\bm{w}\in \mathcal{C}_2$ and $\bm{z} \in \mathcal{C}_1+\bm{y}$. 
% 	\end{corollary}
	\begin{remark}[Transversal $\pi/2^l$ $Z$-rotation]
	\normalfont
	\label{rem:RM_construction}
		 Assume positive signs (character vector $\bm{y}=\bm{0}$) and set $p=2^{l-1}$ for some integer $l\ge 1$. %Let $\mathcal{C}_1*\mathcal{C}_2$ be the subspace generated by vectors $\bm{c_1}*\bm{c_2}$, where $\bm{c_1}\in \mathcal{C}_1$ and $\bm{c_2}\in \mathcal{C}_2$. 
		 Since $\bm{0}\in \mathcal{C}_1$ and $\bm{0} \in \mathcal{C}_2$, it follows from Theorem \ref{thm:div_cond_RZ} that $R_Z\left(\frac{\pi}{2^{l-1}}\right)$ preserves a CSS codespace $\mathcal{V}(\mathcal{S})$ if and only if 
		%   $\theta=\frac{\pi}{4}$ $\&$ $y=0$
		\begin{equation} \label{eqn:div_pi_2^l_1}
		2^l ~|~ w_H(\bm{w}) \text{ for all } \bm{w}\in \mathcal{C}_2 \text{, and}
		\end{equation}
		 \begin{equation}\label{eqn:div_pi_2^l_2}
		 2^{l-1} ~|~ w_H(\bm{w}*\bm{z}) \text{ for all } \bm{w}\in \mathcal{C}_2 \text{ and for all } \bm{z}\in \mathcal{C}_1.
		 \end{equation}
		%where $\mathcal{C}_1*\mathcal{C}_2$ is the subspace generated by a linear independent set of $\{\bm{c_1}*\bm{c_2}: \bm{c_1}\in \mathcal{C}_1 \text{ and } \bm{c_2}\in \mathcal{C}_2 \}$.
	    This result coincides with the sufficient conditions in \cite[Proposition 4]{vuillot2019quantum}, which is a special case of the quasitransversality introduced earlier by Campbell and Howard \cite{campbell2017unified}. For example, if a CSS code with all positive stabilizers is invariant under $R_Z\left(\frac{\pi}{4}\right)$, then the weight of every $X$-stabilizers needs to be divisible by 8. We note that the $\llbr 8,3,2 \rrbr $ color code is the smallest error-detecting CSS code with all positive signs that is preserved by $R_Z\left(\frac{\pi}{4}\right)$. We defer the study of non-trivial character vectors $\bm{y}$ to future work.
	    %The divisibility conditions suggest constructions based on classical Reed-Muller codes. Consider $\mathcal{C}_1 = \mathrm{RM}(r_1,m) \supset \mathcal{C}_2 = \mathrm{RM}(r_2,m)$ with $r_1 > r_2$. Then, for $\bm{w} \in \mathcal{C}_2$ and $\bm{z}\in \mathcal{C}_1$, we have $2^{m-r_2} \mid w_H(\bm{w})$ and $2^{m-r_1-r_2+1} \mid 2w_H(\bm{w}*\bm{z})$. Note that $r_1 > r_2$, and so $m-r_2 \ge m-r_1-r_2+1$. It follows from \eqref{eqn:div_cond_invariant} that the CSS($X,\mathrm{RM}(r_2,m);Z,\mathrm{RM}(m-r_1-1,m)$) code with all positive stabilizers is preserved by $R_Z(\frac{\pi}{2^{l}})$ for $l\le m-r_1-r_2+1$. The CSS code has parameters $n=2^m$, $k=\sum_{j=r_2+1}^{r_1} \binom{m}{j}$, and $d=2^{\min\{r_2+1,m-r_1\}}$. If we consider third level transversal $Z$-rotations or higher level ($l\ge2$), then $2\le l \le m-r_1-r_2+1$. Hence, $r_2+1 \le m-r_1$, and we have the family of the $\llbr 2^m,\sum_{j=r_2+1}^{r_1} \binom{m}{j},2^{r_2+1} \rrbr $ CSS codes. 
	    %When $l=2$, the CSS($X,\mathrm{RM}(r_2,m);Z,\mathrm{RM}(m-r_1-1,m)$) code also include Felice's construction \cite{} . 
	\end{remark}
	%The main tool we used is the generator coefficients, and we consider the code projector of CSS code $\Pi_{\mathcal{S}}$ as a product of two projectors $\Pi_{\mathcal{S}_X}$ and $\Pi_{\mathcal{S}_Z}$ rather than the summation over a subspace in $\mathcal{HW}_N$. Since $R_Z(\theta)$ automatically commutes with $\Pi_\mathcal{S_{Z}}$, $\Pi_\mathcal{S_{Z}}$ provides $R_Z(\theta)$ a structure defined by generator coefficients before meeting with $\Pi_\mathcal{S_{X}}$. After commuting with $\Pi_\mathcal{S_{X}}$, we derive a constraint on the generator coefficients. We think it is possible to apply the generator coefficients formalism to some other gates in higher level in the Clifford hierarchy if their Pauli expansions are known.
    
    %The divisibility conditions for $R_Z\left(\frac{\pi}{2^{l-1}}\right)$ suggest constructions based on classical Reed-Muller codes. We construct CSS codes, invariant under $R_Z\left(\frac{\pi}{2^{l-1}}\right)$, out of RM components in the Lemma below. 
    The divisibility conditions \eqref{eqn:div_pi_2^l_1}, \eqref{eqn:div_pi_2^l_2} suggest constructing CSS codes from classical Reed-Muller codes.
	\begin{theorem}[Reed-Muller Constructions]
	\label{thm:RM_cosntruction}
	Consider Reed-Muller codes $\mathcal{C}_1 = \mathrm{RM}(r_1,m) \supset \mathcal{C}_2 = \mathrm{RM}(r_2,m)$ with $r_1 > r_2$. The $\llbr n=2^m,k=\sum_{j=r_2+1}^{r_1} \binom{m}{j},d=2^{\min\{r_2+1,m-r_1\}} \rrbr $ CSS($X,\mathcal{C}_2;Z,\mathcal{C}_1^\perp$) code with all positive stabilizers is preserved by $R_Z(\frac{\pi}{2^{l-1}})$ if and only if
	\begin{widetext}
	\begin{equation}\label{eqn:RM_const_relation_l}
	    l\le \left\{\begin{array}{lc}
	    \left\lfloor \frac{m-1}{r_1} \right\rfloor + 1, ~~~~~~~~~~~~~~~~~~~~~~~~~~~~\text{ if } r_2 = 0, \\
    	\min \left\{ \left\lfloor \frac{m-r_2-1}{r_1} \right\rfloor +1 , 
	    \left\lfloor \frac{m-r_1}{r_2} \right\rfloor+1 \right\}, \text{ if } r_2 \neq 0.
	    \end{array} \right.
	\end{equation}
	\end{widetext}
	\end{theorem}
	\begin{proof}
	    %Since all $Z$-stabilizers have positive signs, it corresponds to the case $\bm{y}=\bm{0}$ in Theorem \ref{thm:div_cond_RZ}. 
	    Note that all $Z$-stabilizers have positive signs corresponding to the case $\bm{y}=\bm{0}$ in Theorem \ref{thm:div_cond_RZ}.  Then, $R_Z\left(\frac{\pi}{2^{l-1}}\right)$ preserves a CSS codespace if and only if 
		\eqref{eqn:div_pi_2^l_1} and \eqref{eqn:div_pi_2^l_2} hold. 
		
		Let $\bm{w}\in \mathcal{C}_2$ and $\bm{z}\in \mathcal{C}_1$. If $r_2=0$, then $\mathcal{C}_2 = \{\bm{0},\bm{1}\}$ and $w_H(\bm{w})\in\{0,2^m\}$. It follows from McEliece \cite{mceliece1971periodic} (see also Ax \cite{ax1964zeroes}) that 
		\begin{equation}\label{eqn:RM_const_cond}
		    2^{\left\lfloor \frac{m-1}{r_1} \right\rfloor} \mid  w_H(\bm{w}*\bm{z}),
		\end{equation}
		and this bound is tight. The two conditions become $ l\le\min\{m,\left\lfloor\frac{m-1}{r_1}\right\rfloor+1\} = \left\lfloor\frac{m-1}{r_1}\right\rfloor+1$. 
	    
	    If $r_2 \neq 0$, then it follows from McEliece \cite{McEliece,borissov2013mceliece} that $\left\lfloor \frac{m-1}{r_2} \right\rfloor$ is the highest power of $2$ that divides $w_H(\bm{w})$ for all $\bm{w} \in \mathcal{C}_2=\mathrm{RM}(r_2,m) $. 
	    We first show \eqref{eqn:RM_const_relation_l} is necessary. It follows from \eqref{eqn:div_pi_2^l_1} that 
	    \begin{equation}\label{eqn:RM_const_cond1}
	        l\le \left\lfloor \frac{m-1}{r_2} \right\rfloor .%,
	    \end{equation}  
	    %which requires $l\le \left\lfloor \frac{m-1}{r_2}\left\lfloor$.
	    We need to understand divisibility of weights $w_H(\bm{w}*\bm{z})$ where $\bm{w} \in \mathcal{C}_2$ and $\bm{z} \in  \mathcal{C}_1$. The codeword $\bm{w}$ is the evaluation vector of a sum of monomials, and we start by considering the case of a single monomial. 
	    Consider a codeword $\bm{w_1}\in \mathcal{C}_2$ corresponding to the evaluation of a monomial of degree $s$. For all $\bm{z}\in\mathcal{C}_1$, we observe that $\bm{w_1}*\bm{z}$ is a codeword in RM$(\min\{r_1,m-s\},m-s)$ supported on  $\bm{w_1}$. Then, $\left\lfloor \frac{m-s-1}{\max\{r_1,m-s\}}  \right\rfloor$ is the highest power of $2$ that divides $w_H(\bm{w_1}*\bm{z})$ for all $\bm{z}\in\mathcal{C}_1$. Note that since $s$ takes values from $0$ to $r_2$, we have 
	    \begin{align}\label{eqn:RM_const_cond2}
	     l&\le \left\lfloor \frac{m-r_2-1}{\max\{r_1,m-r_2\}} \right\rfloor +1 \nonumber \\
	     &=\left\{\begin{array}{lc}
	        \left\lfloor \frac{m-r_2-1}{r_1} \right\rfloor +1,  & \text{ if } r_1+r_2\le m,\\
	         1=\left\lfloor \frac{m-r_1}{r_2}\right\rfloor+1, & \text{ if } m<r_1+r_2.
	     \end{array} \right. %,
	    \end{align}
	    %which implies $l\le $
	    We now consider $\bm{w}\in\mathcal{C}_2$ such that $\bm{w}=\bm{w_1}\oplus\bm{w_2}$, where $\bm{w_1}$, $\bm{w_2}$ are evaluation vectors correspond to monomials in $\mathcal{C}_2$. Then, for $\bm{z} \in \mathcal{C}_1$, we have
	    \begin{align*}\label{eqn:RM_const_simp1}
	       w_H(\bm{w}*\bm{z}) &= w_H(\bm{w_1}*\bm{z}) + w_H(\bm{w_2}*\bm{z}) \\
	       &~~~- 2w_H(\bm{w_1}*\bm{w_2}*\bm{z}).\numberthis
	    \end{align*}
	    Since $\bm{w},\bm{w_1},\bm{w_2}\in \mathcal{C}_2$, it follows from \eqref{eqn:div_pi_2^l_2} that $2^l$ divides $2w_H(\bm{w}*\bm{z}),~ 2w_H(\bm{w_1}*\bm{z})$, and so  $2w_H(\bm{w_2}*\bm{z})$. By \eqref{eqn:RM_const_simp1}, we have 
	    \begin{equation}
	        2^l | 4w_H(\bm{w_1}*\bm{w_2}*\bm{z}).
	    \end{equation}
	    Since $\bm{w_1}*\bm{w_2}$ is the evaluation vector of a monomial with degree $s'\le \min\{m,2r_2\}$, $\bm{w_1}*\bm{w_2}*\bm{z}$ is a codeword in RM$(\min\{r_1,m-s'\},m-s')$ supported on $\bm{w_1}*\bm{w_2}$. Then, $\left\lfloor \frac{m-2r_2-1}{\max\{r_1,m-2r_2\}} \right\rfloor$ is the highest power of $2$ that divides $w_H(\bm{w_1}*\bm{w_2}*\bm{z})$ for all $\bm{w_1}*\bm{w_2}\in \mathcal{C}_2$. The extremum is achieved when the monimials corresponding to $\bm{w_1}$ and $\bm{w_2}$ have degree $r_2$ and do not share a variable. Hence, 
	    
	\begin{widetext}
	    \begin{equation}\label{eqn:RM_const_cond3}
	     l \le \left\lfloor \frac{m-2r_2-1}{\max\{r_1,m-2r_2\}} \right\rfloor +2 
	     =\left\{\begin{array}{lc}
	        \left\lfloor \frac{m-2r_2-1}{r_1} \right\rfloor +2,  & \text{ if } r_1+2r_2\le m,\\
	         2=\left\lfloor \frac{m-r_1}{r_2}\right\rfloor+1, & \text{ if } r_1+r_2\le m<r_1+2r_2.
	     \end{array} \right. 
	    \end{equation}
	\end{widetext}
	    It remains to consider the case   $\bm{w}=\bm{w_1}\oplus\bm{w_2}\oplus\cdots\oplus\bm{w_t} \in \mathcal{C}_2$, where each $\bm{w_i}$ is the evaluation vector of a monomial. We use inclusion-exclusion to rewrite \eqref{eqn:div_pi_2^l_2} as
	    \begin{widetext}
	    \begin{align}\label{eqn:RM_const_wt_expan}
	       2^{l-1}\Big|\sum_{i=1}^{t}(-2)^{i-1}\sum_{1\le j_1\le \dots\le j_i \le t } w_H(\bm{w_{j_1}}*\cdots * \bm{w_{j_i}}*\bm{z}).
	    \end{align}
	    \end{widetext}
	   % and \eqref{eqn:div_pi_2^l_2} is equivalent to 
	   % \begin{equation}\label{eqn:RM_const_wt_expan}
	   %     2^{l-1}|\sum_{i=1}^{\min\{t,l-1\}}(-1)^{i-1}2^{i-1}\sum_{1\le j_1\le \dots\le j_i \le t } w_H(\bm{w_{j_1}}*\cdots * \bm{w_{j_i}}*\bm{z}).
	   % \end{equation}
	    %Consider $t=3,4,\dots \left\lfloor\frac{m-r_1}{r_2}\right\rfloor$.
	    We now use induction. Assume for $1\le i\le t-1$, we have
	    \begin{widetext}
	    \begin{align}\label{eqn:RM_const_inducton_hyp1}
	     l \le \left\lfloor \frac{m-ir_2-1}{\max\{r_1,m-ir_2\}} \right\rfloor +i =\left\{\begin{array}{lc}
	        \left\lfloor \frac{m-ir_2-1}{r_1} \right\rfloor +i,  & \text{ if } r_1+ir_2\le m,\\
	         i=\left\lfloor \frac{m-r_1}{r_2}\right\rfloor+1, & \text{ if } 
	         %m\in [r_1+(i-1)r_2,r_1+ir_2).
	          (i-1)r_2\le m-r_1<ir_2.
	     \end{array} \right. 
	    \end{align}
	    \end{widetext}
	    Note that for $ 1\le i\le t-1 $, $\bm{w_{j_1}}*\cdots * \bm{w_{j_i}}$ corresponds to a monomial with degree $s''\le \min\{m,ir\}$, hence $\bm{w_{j_1}}*\cdots * \bm{w_{j_i}}*\bm{z}$ is a codeword in RM($\min\{r_1,m-s''\},m-s''$) supported on  $\bm{w_{j_1}}*\cdots * \bm{w_{j_i}}$. Then, we have
	    \begin{equation}\label{RM_const_div_indu}
	         2^{\left\lfloor \frac{m-ir_2-1}{\max\{r_1,m-ir_2\}}  \right\rfloor + i} \mid 
	         2^{i}w_H(\bm{w_{j_1}}*\cdots * \bm{w_{j_i}}*\bm{z}),
	    \end{equation}
	    in which the bound on the exponent is tight since we can choose $\bm{w_1},\cdots,\bm{w_i}$ to be evulations vectors corresponding to $i$ disjoint monomials of degree $r_2$. Hence, $2^{l-1}$ divides all terms in \eqref{eqn:RM_const_wt_expan} for $i=1,2,\dots, t-1$. Hence, for the last term, we must have 
	    \begin{equation}
	        2^{l-1}|2^{t-1} w_H(\bm{w_{1}}*\cdots * \bm{w_{t}}*\bm{z}),
	    \end{equation}
	    which implies that 
	    \begin{widetext}
	    \begin{align}\label{eqn:RM_const_inducton_hyp2}
	     l \le \left\lfloor \frac{m-tr_2-1}{\max\{r_1,m-tr_2\}} \right\rfloor +t
	     =\left\{\begin{array}{lc}
	        \left\lfloor \frac{m-tr_2-1}{r_1} \right\rfloor +t,  & \text{ if } r_1+t r_2\le m,\\
	         t=\left\lfloor \frac{m-r_1}{r_2}\right\rfloor+1, & \text{ if } (t-1)r_2\le m-r_1<tr_2,
	     \end{array} \right.
	    \end{align}
	    \end{widetext}
	    and the induction is complete.
	    Note that since $r_1>r_2$, we have
	    \begin{equation}
	        \left\lfloor \frac{m-tr_2-1}{r_1} \right\rfloor +t\ge \left\lfloor \frac{m-r_2-1}{r_1} \right\rfloor +1 \text{ for } t\ge 1,
	    \end{equation}
	    and the necessary condition reduces to 
	    \begin{equation}
	        l\le \min \left\{ \left\lfloor \frac{m-r_2-1}{r_1} \right\rfloor +1 ,
	    \left\lfloor \frac{m-r_1}{r_2} \right\rfloor+1 \right\}.
	    \end{equation}
	    To prove the sufficiency of the case $r_2\neq 0$, we simply reverse the steps.
	\end{proof}
    \begin{remark}[Puncturing RM codes by removing the first coordinate]
        \label{rem:elemtary_op}
        \normalfont
        Consider the classical RM$(r,m)$ code, and two elementary operations on its generator matrix: 1. removing the first column which is $[1,0,\dots,0]^T$; 2. removing the first row of all 1s. After either of the two operations, we observe that $2^{\lfloor \frac{(m-1)}{2}\rfloor}$ is still the highest power of $2$ that divides all of its weights. Hence, the RM constructions described in Theorem \ref{thm:RM_cosntruction} can be extended to punctured RM codes. If operation 1 is applied on $\mathcal{C}_1 = \mathrm{RM}(r_1,m)$, and operations 1 and 2 are applied on $\mathcal{C}_2 = \mathrm{RM}(r_2,m)$, then we can relax the relation between $r_1$ and $r_2$ as $r_1 \ge r_2$. It follows from the same arguments that the resulting $\llbr 2^m-1,\sum_{j-r_2+1}^{r_1}\binom{m}{j}+1,2^{\min\{r_2+1,m-r_1\}}-1 \rrbr $ CSS code is preserved by $R_Z(\frac{\pi}{2^{l-1}})$ with the same constraint on $l$ as described in \eqref{eqn:RM_const_relation_l}. This family contains the $\llbr 2^m-1,1,3 \rrbr $ triorthogonal codes described in \cite{bravyi2012magic}. 
        % If both operators are applied on $\mathcal{C}_1 = \mathrm{RM}(r_1,m)$ and $\mathcal{C}_2 = \mathrm{RM}(r_2,m)$, then we obtain the $\llbr 2^m-1,\sum_{j-r_2+1}^{r_1}\binom{m}{j},2^{\min\{r_2+1,m-r_1\}}-1 \rrbr $ CSS code that is preserved by $R_Z(\frac{\pi}{2^{l-1}})$ with the same constraint on $l$ as described in \eqref{eqn:RM_const_relation_l}. 
    \end{remark}
    \begin{remark}[QRM$(r,m)$ Codes]
    \label{rem:QRM_codes}
    \normalfont
        When $r_1=r$ and $r_2=r-1$, this family of CSS codes coincides with the QRM($r,m$) $\llbr 2^m,\binom{m}{r},2^{\min\{r,m-r\}} \rrbr $ codes constructed in \cite{haah2018codes} and \cite[Theorem 19]{rengaswamy2020optimality}. The code QRM($r,m$) is preserved by $R_Z(\frac{2\pi}{2^{m/r}})$ if $1\le r\le m/2$ and $r \mid m$. When $r_2=0$, we obtain the $\llbr 2^m,m,2 \rrbr $ family that is preserved by $R_Z(\frac{2\pi}{2^m})$. If $r_2\neq 0$, since $r \mid m$, we have 
        \begin{align*}
            l&=\frac{m}{r} 
            = \min\left\{ \left\lfloor \frac{m-r}{r}\right\rfloor+1, \left\lfloor \frac{m-1}{r-1}\right\rfloor\right\} \\
            &= \min\left\{ \left\lfloor \frac{m-(r-1)-1}{r}\right\rfloor+1, \left\lfloor \frac{m-r}{r-1}\right\rfloor + 1 \right\}, \numberthis
        \end{align*}
        which satisfies the necessary and sufficient conditions in \eqref{eqn:RM_const_relation_l}.
    \end{remark}
	We now illustrate Theorem \ref{thm:preserved_by_Uz} and Theorem \ref{thm:div_cond_RZ} through two CSS codes preserved by $R_Z\left(\frac{\pi}{4}\right)$, one with a single logical qubit, the other with multiple logical quibts. 
	\addtocounter{example}{+3}
	\begin{example}[The $\llbr 15,1,3 \rrbr $ punctured quantum Reed-Muller code~\cite{knill1996accuracy,bravyi2005universal}]
		\normalfont
		\label{examp11}
		%  ~\cite{bravyi2005universal,Opt}
		Consider the CSS($X, \mathcal{C}_2;Z,\mathcal{C}_1^\perp$) code defined by $\mathcal{C}_2 = \langle x_1,x_2,x_3,x_4\rangle$ and $\mathcal{C}_1^\perp = \langle x_1,x_2,x_3,x_4,x_1x_2,x_1x_3,x_1x_4,x_2x_3,x_2x_4,x_3x_4 \rangle$, with the first coordinate removed in both $\mathcal{C}_2$ and $\mathcal{C}_1^{\perp}$. 
		It is well-known \cite{bravyi2005universal,rengaswamy2020optimality} that $R_Z(\frac{\pi}{4})$ preserves the CSS codespace when the signs of $Z$-stabilizers are trivial. Since $8\mid w_H(\bm{v})$, for $\bm{v}\in \mathrm{RM}(1,4)$ and $4 \mid w_H(\bm{u})$ for $\bm{u}\in \mathrm{RM}(2,4)$), the code satisfies the divisibility conditions in Theorem \ref{thm:div_cond_RZ}. We compute the induced logical operator by computing the generator coefficients for the zero syndrome. Note that $\mathcal{C}_2^\perp /\mathcal{C}_1^\perp = \{\bm{0},\bm{1}\}$. The weight enumerators of $\mathcal{C}_1$ and $\mathcal{C}_1 + \bm{1}$ are given by
% 		$\gamma = \underline{0}_{15}$ and $\underline{1}_{15}$ as 
% 		\begin{equation}
% 		A_{0,\gamma} (\frac{\pi}{4})
% 		= \sum_{z\in C_1^\perp +\gamma}  \left(\cos\frac{\pi}{8}\right)^{15-w_H(z)}\left(-\imath \sin\frac{\pi}{8}\right)^{w_H(z)} 
% 		= \frac{1}{|C_1|}\sum_{v\in C_1} (-1)^{\gamma v^T} (e^{-\imath\frac{\pi}{8}})^{15-2w_H(v)}, 
% 		\end{equation}
% 		where the last step follows from the MacWilliams Identities. Note that the weight enumerator of both $C_1$ and $C_1 + \underline{1}_{15}$ is
		\begin{align*}
		&P_{\mathcal{C}_1}(x,y) = P_{\mathcal{C}_1+ \bm{1}}(x,y) \\
		&~~~~~~~~~~~~~~= x^{15}+15x^8y^7+15x^7y^8+y^{15}.
		\end{align*}
		We have
		%and the equivalent definition of generator coefficients in \eqref{eqn:A_mu,gamma} that
		\begin{align}
		A_{\bm{0},\bm{0}}\left(\frac{\pi}{4}\right) 
		%&= \frac{1}{|\mathcal{C}_1|} \left(  \left(e^{-\imath \frac{\pi}{8}} \right)^{15}+ 15 \left(e^{-\imath \frac{\pi}{8}}\right)^{15-2\cdot7} + 15 \left(e^{-\imath \frac{\pi}{8}}\right)^{15-2\cdot8} + \left(e^{-\imath \frac{\pi}{8}}\right)^{15-2\cdot15}  \right)\\
		&= \frac{1}{32} \left(2\cos\frac{15\pi}{8} + 30\cos\frac{\pi}{8} \right) = \cos\frac{\pi}{8},\nonumber \\
% 		\end{align}
% 		and
% 		\begin{align}
		A_{\bm{0},\bm{1}}\left(\frac{\pi}{4}\right) 
		%&= \frac{1}{|\mathcal{C}_1|} \left(  \left(e^{-\imath \frac{\pi}{8}} \right)^{15} - 15 \left(e^{-\imath \frac{\pi}{8}}\right)^{15-2\cdot7} + 15 \left(e^{-\imath \frac{\pi}{8}}\right)^{15-2\cdot8} - \left(e^{-\imath \frac{\pi}{8}}\right)^{15-2\cdot15}  \right)\\
% 		&= \frac{1}{32} \left(-2\imath\sin\frac{15\pi}{8} + 30\imath\sin\frac{\pi}{8}  \right) 
        &= \imath\sin\frac{\pi}{8}.
		\end{align}
		The constraint on generator coefficients in \eqref{eqn:preserved_by_Uz} is satisfied: 
		\begin{align*}
		\sum_{\bm{\gamma}\in\{\bm{0},\bm{1}\}}\left|A_{\bm{0},\bm{\gamma}}\left(\frac{\pi}{4}\right) \right|^2 
% 		+ \left|A_{\bm{0},\bm{1}}\left(\frac{\pi}{4}\right) \right|^2 
        = \left(\cos\frac{\pi}{8}\right)^2 + \left(\sin\frac{\pi}{8}\right)^2 = 1.
		\end{align*}
		It follows from \eqref{eqn:logical_final} that the logical operator induced by $R_Z\left(\frac{\pi}{4}\right)$ is
		\begin{align*}
		R_Z^L\left(\frac{\pi}{4}\right) &= A_{\bm{0},\bm{0}}\left(\frac{\pi}{4}\right) I^L + A_{\bm{0},\bm{1}}\left(\frac{\pi}{4}\right) Z^L \\
		&= \cos\frac{\pi}{8} I^L + \imath\sin\frac{\pi}{8}Z^L = (T^\dagger)^L.
		% \left[\exp(-\imath \frac{\pi}{8}\sigma_Z)\right]
		\end{align*}
		% We describe this kinds of CSS codes in which a transversal physical $\theta$ rotation implements a transversal logical $\theta$ to be \emph{symmetric}.
	\end{example}

	\begin{example}[The $\llbr 8,3,2 \rrbr $ code]
	    \label{examp12}
		\normalfont
		%($\llbr 8,3,2 \rrbr $ Color Code). 
		The $\llbr 8,3,2 \rrbr $ color code \cite{campbell2017unified} is defined on $8$ qubits which we identify with vertices of the cube. All vertices participate in the X-stabilizer and generators of the Z-stabilizers can be identified with 4 independent faces of the cube. The signs of all the stabilizers are positive. The $\llbr 8,3,2 \rrbr $ color code can also be thought as a Reed-Muller CSS($X,~ \mathcal{C}_2 =\{ \bm{0,\bm{1}}\};~Z,~ \mathcal{C}_1^\perp = \text{RM}(1,3)$) code with generator matrix
		\begin{align}
		\setlength\aboverulesep{0pt}\setlength\belowrulesep{0pt}
		\setlength\cmidrulewidth{0.5pt}
		G_S = 
		\left[
		\begin{array}{c|cccccccc}
		\bm{1} &  &  &  &  &  &  &  &    \\
		\hline
		  & 1 & 1 & 1 & 1 & 1 & 1 & 1 & 1   \\
		  & 0 & 0 & 0 & 0 & 1 & 1 & 1 & 1   \\
		  & 0 & 0 & 1 & 1 & 0 & 0 & 1 & 1   \\
		  & 0 & 1 & 0 & 1 & 0 & 1 & 0 & 1   \\
		\end{array}
		\right].
		\end{align}
		The $\llbr 8,3,2 \rrbr $ code can be used in magic state distillation for the controlled-controlled-$Z$ (CCZ) gate in the third-level of Clifford hierarchy. To verify that the code is preserved by $R_Z\left(\frac{\pi}{4}\right)$ and the induced logical operator is CC$Z$ (up to some logical Pauli $Z^L$), we first compute the generator coefficients corresponding to the trivial syndrome. The weight enumerators of $\mathcal{C}_1^\perp$ and $\mathcal{C}_1^\perp + \bm{\gamma}$ for $\bm{\gamma} \in \mathcal{C}_2^\perp /\mathcal{C}_1^\perp \setminus \{\bm{0}\}$ are given by
		\begin{align*}
		P_{\mathcal{C}_1^\perp}(x,y) &= x^8+ 14x^4y^4 + y^8,\\ P_{\mathcal{C}_1^\perp + \bm{\gamma}}(x,y) &= 4x^6y^2+ 8x^4y^4 + 4x^2y^6,
		\end{align*}
		so that
		\begin{align}
		    A_{\bm{0},\bm{0}}\left(\frac{\pi}{4}\right) = \frac{3}{4}, \text{ and }  A_{\bm{0},\bm{\gamma}\neq \bm{0}}\left(\frac{\pi}{4}\right) = -\frac{1}{4}
		\end{align}
		for all the seven non-zero $\bm{\gamma} \in \mathcal{C}_2^\perp /\mathcal{C}_1^\perp$.
		Then, 
		\begin{align*}
		\sum_{\gamma\in \mathcal{C}_2^\perp/\mathcal{C}_1^\perp} \left|A_{\bm{0},\bm{\gamma}}\left(\frac{\pi}{4}\right)\right|^2 = \left(\frac{3}{4}\right)^2 + 7\cdot\left(-\frac{1}{4}\right)^2 = 1,
		\end{align*}
		so \eqref{eqn:preserved_by_Uz} holds, and the induced logical operator is
		\begin{align*}
		R_Z^L\left(\frac{\pi}{4}\right) 
		&= \sum_{\bm{\alpha}\in \F_2^3} A_{\bm{0}, g(\bm{\alpha)}} \left(\frac{\pi}{4}\right) E(\bm{0},\bm{\alpha}) \\
		%({Z}^L)^{v_1} ({Z}^L)^{v_2} ({Z}^L)^{v_3} \\
% 		&= \frac{3}{4}I^L - \frac{1}{4} \sum_{\bm{v}\in \F_2^3\setminus\{\bm{0}\}} A_{\bm{0}, g(v_1,v_2,v_3)} \left(\frac{\pi}{4}\right) ({Z}_1^L)^{v_1} ({Z}_2^L)^{v_2}  ({Z}_3^L)^{v_3} \\
% 		\left(I_L I_L Z_L + I_L Z_L I_L + I_L Z_L Z_L + Z_L I_L I_L + Z_L I_L Z_L + Z_LZ_LI_L + Z_LZ_LZ_L\right)\\
		& \equiv (Z^L\otimes I^L \otimes Z^L)\circ\mathrm{CC}Z^L.\numberthis
		\end{align*}
% 		up to some logical Pauli $Z^L$.
		%		\left[
		%		\begin{array}{cccccccc}
		%			-1 &  &  &  & & & & &  \\
		%			& 1 &  &  & & & &  &\\
		%			&  & 1 & & & & &   & \\
		%			&  &  & 1  & & & & & \\
		%			& & & & 1 &  & & & \\
		%			& & & & & 1 &  & & \\
		%			& & & & & & 1 &  & \\
		%			& & & & &  & & 1 & \\
		%		\end{array}
		%		\right]
	\end{example}

    \subsubsection{Generator Coefficients and Trigonometric Identities}
	When $\theta = \frac{\pi}{2^l}$ for some integer $l$, Rengaswamy et al. \cite{rengaswamy2020optimality} derived necessary and sufficient conditions for a stabilizer code to be invariant under $R_Z(\theta)$. This derivation depends on prior work characterizing conjugates of arbitrary Pauli matrices by $R_Z(\frac{\pi}{2^l})$ \cite{rengaswamy2019unifying}. The necessary and sufficient conditions provided in \cite[Theorem 17]{rengaswamy2020optimality} are expressed as two types of trigonometric identity. We now show that our constraint on generator coefficients is equivalent to the first trigonometric identity, and that the second trigonometric identity follows from the first. Our main tool is the MacWilliams Identities \cite{Mac}, and our analysis extends from CSS codes to general stabilizer codes. 
	
	We demonstrate equivalence through a sequence of three lemmas.
	 %We first start with our constraint on generator coefficients in \eqref{eqn:preserved_by_Uz} and derive a middle-step condition in Lemma \ref{lemma:simplify}, and then show the middle-step condition is equivalent to the first kind of trigonometric identities in \cite[Theorem17]{Opt} by Lemma \ref{lemma:Mac1} and Lemma \ref{lemma:space_equal}. As a result, we conclude in Theorem \ref{thm: connect_to17} to connect directly between \eqref{eqn:preserved_by_Uz} and the first kind of trigonometric identities based on the three lemmas.
	
% 	Given two binary vectors $x,y$, we write $x\preceq y$ to mean that the \emph{support} of $x$, i.e., the set of indices with non-zero entries in $x$, is contained in the support of $y$. We consider the CSS($X,C_2;Z,C_1^\perp$) code and given any nontrivial $w\in C_2\subset C_1$, we define 
% 	\begin{equation}
% 	Z_w \coloneqq \{\tilde{z}\big|_{\mathrm{supp}(w)}: \epsilon_{\tilde{z}} E\left(0,\tilde{z}\right) \in S \text{ and } \tilde{z} \preceq w\},
% 	\end{equation}
% 	so that $Z_w$ is a binary code of length $w_H(w)$. For all $z\in Z_w$, we define $$\tilde{z} \in \F_2^n \text{ such that } \tilde{z}|_{\mathrm{supp}(w)} = z \text{ and all positions outside the support of $w$ are zero.}$$
	
	\begin{lemma}\label{lemma:simplify}
		Given a CSS($X,\mathcal{C}_2;Z,\mathcal{C}_1^\perp$) code, let $\mathcal{B} = \{\bm{z}\in \mathcal{C}_1^\perp : \epsilon_{\bm{z}} = 1\}$ and $\mathcal{B}^\perp = \langle \mathcal{C}_1,\bm{y}\rangle $. For all nontrivial $\bm{w}\in \mathcal{C}_2$, define %$\mathcal{K}_{\bm{w}} \coloneqq \{\bm{z}\in \mathcal{C}_1: w_H(\bm{w}*\bm{z})=0\}$, and 
		$\mathcal{D}_{\bm{w}} \coloneqq \{\bm{w}*\bm{v} : \bm{v}\in \mathcal{C}_1 \}$. 
		% / \mathcal{K}_{\bm{w}}
		Let $\theta \in (0,2\pi)$. Then, \eqref{eqn:preserved_by_Uz} holds 
% 		\begin{equation} 
% 		\sum_{\bm{\gamma} \in \mathcal{C}_2^\perp /\mathcal{C}_1^\perp } \left|A_{\bm{0},\bm{\gamma}}(\theta)\right|^2
%         =1
% 		\end{equation} 
		if and only if for all non-zero $\bm{w}\in \mathcal{C}_2$
		\begin{equation}\label{eqn:D_w}
		\frac{1}{|\mathcal{D}_{\bm{w}}|}\sum_{\bm{x}\in \mathcal{D}_{\bm{w}}+\bm{w}*\bm{y}} \left(e^{\imath\theta}\right)^{w_H(\bm{w})-2w_H(\bm{x})} = 1.
		% s_w = \frac{1}{|\mathcal{C}_1|}\sum_{z\in \mathcal{C}_1+y} \left(e^{\imath\theta}\right)^{w_H(w)-2w_H(w*z)} = 1,
		\end{equation}
	\end{lemma}
	\begin{proof}
	    See Appendix \ref{subsec:proof_connect_lem_1}.
	\end{proof}

	The \emph{support} of a binary vector $\bm{x}$ is the set of coordinates for which the corresponding entry is non-zero. Given two binary vectors $\bm{x}$, $\bm{y}$, we write $\bm{x} \preceq \bm{y}$ to mean that the support of $\bm{x}$ is contained in the support of $\bm{y}$. Let $\mathrm{supp}(\bm{x})$ be the support of $\bm{x}$. We define $\bm{y}|_{\mathrm{supp}(\bm{x})} \in \F_2^{w_H(\bm{x})}$ to be the truncated binary vector that drops all the coordinates outside $\mathrm{supp}(\bm{x})$. Given a space $\mathcal{C}$, we denote $\mathrm{proj}_{\bm{x}}(\mathcal{C})\coloneqq \{\bm{v}\in\mathcal{C} : \bm{v} \preceq \bm{x}\}$. The next lemma finds equivalent representations of the cosets $\mathcal{D}_{\bm{w}} + \bm{w}*\bm{y}$ for non-zero $\bm{w}\in \mathcal{C}_2$. %Based on these aforementioned notations, we find equivalent representations for the cosets $\mathcal{D}_{\bm{w}} + \bm{w}*\bm{y}$ for non-zero $\bm{w}\in \mathcal{C}_2$ as in the following lemma.
	
	\begin{lemma}\label{lemma:space_equal}
	Given a CSS($X,\mathcal{C}_2;Z,\mathcal{C}_1^\perp$) code, define %$\mathcal{B}, 
	$\mathcal{D}_{\bm{w}}$ and $\bm{y}$ as above.
	%, \mathcal{K}_{\bm{w}}$ the same as in Lemma \ref{lemma:simplify}. 
	For any non-zero $\bm{w}\in \mathcal{C}_2$, define 
% 	$\mathcal{Z}_{\bm{w}} \coloneqq \{\tilde{\bm{z}}\big|_{\mathrm{supp}(\bm{w})}: 
% % 	\epsilon_{\tilde{\bm{z}}} E\left(\bm{0},\tilde{\bm{z}}\right) \in \mathcal{S} 
%     \tilde{\bm{z}} \in \mathcal{C}_2 
% 	\text{ and } \tilde{\bm{z}} \preceq \bm{w}\}$ 
	$\mathcal{Z}_{\bm{w}} \coloneqq \{\bm{z}\big|_{\mathrm{supp}(\bm{w})} \in \F_2^{w_H(\bm{w})}: 
    \bm{z} \in \mathcal{C}_1^\perp 
	\text{ and } \bm{z} \preceq \bm{w}\}$ 
	and $\mathcal{B}_{\bm{w}} = \{\bm{v}\in \mathcal{Z}_{\bm{w}}: \epsilon_{\bm{v}} = 1\}$. 
	Define $\tilde{\mathcal{Z}}_{\bm{w}} \subset \F_2^n $ (resp. $\tilde{\mathcal{B}}_{\bm{w}} \subset \F_2^n$) by adding all the zero coordinates outside $\mathrm{supp}(\bm{w})$ back into $\mathcal{Z}_{\bm{w}}$ (resp. $\mathcal{B}_{\bm{w}}$).  Note that $\dim(\mathrm{proj}_{\bm{w}}(\tilde{\mathcal{B}}_{\bm{w}}^\perp)) = \dim(\mathrm{proj}_{\bm{w}}(\tilde{\mathcal{Z}}_{\bm{w}}^\perp)) +1$. Define $\bm{y}' \in \F_2^{n}$ such that $\mathrm{proj}_{\bm{w}}(\tilde{\mathcal{B}}_{\bm{w}}^\perp) = \langle \mathrm{proj}_{\bm{w}}(\tilde{\mathcal{Z}}_{\bm{w}}^\perp), \bm{y}' \rangle$. Then for all nontirvial $\bm{w}\in \mathcal{C}_2$,
% 		Given a CSS($X,C_2;Z,C_1^\perp$) code. Let $B = \{z\in C_1^\perp : \epsilon_z = 1\}$ and $B^\perp = \langle C_1,y\rangle $. For any nontrivial $w\in C_2$, we define $Z_w = \{\tilde{z}\big|_{\mathrm{supp}(w)}: \epsilon_{\tilde{z}} E\left(0,\tilde{z}\right) \in S \text{ and } \tilde{z} \preceq w\}$,  and $B_w = \{v\in Z_w: \epsilon_v = 1\}$. Let $\tilde{Z_w} \in \F_2^n $ (resp. $\tilde{B_w} \in \F_2^n$) by completing all the zero entries back on $Z_w$ (resp. $\tilde{B_w}$). Let $\mathrm{proj}_w(\tilde{B_w}^\perp) = \langle \mathrm{proj}_w(\tilde{Z_w}^\perp), y' \rangle$, $K_w = \{z\in C_1: w_H(w*z)=0\bmod 2\}$, and $D_w = \{w*v : v\in C_1 / K_w\}$. Then for all nontirvial $w\in C_2$,
		\begin{equation}
		\mathcal{D}_{\bm{w}} + \bm{w}*\bm{y} = \mathrm{proj}_{\bm{w}}(\tilde{\mathcal{Z}}_{\bm{w}}^\perp) + \bm{y}'.
		\end{equation}
	\end{lemma}
	\begin{proof}
	    See Appendix \ref{subsec:proof_connect_lem_2}.
	\end{proof}

	\begin{lemma}\label{lemma:Mac1}
		Given a CSS($X,\mathcal{C}_2;Z,\mathcal{C}_1^\perp$) code, let $\mathcal{B} = \{\bm{z}\in \mathcal{C}_1^\perp : \epsilon_{\bm{z}} = 1\}$, and define $\mathcal{Z}_{\bm{w}}$, $\tilde{\mathcal{Z}}_{\bm{w}}$, $\mathcal{B}_{\bm{w}}$, $\tilde{\mathcal{B}}_{\bm{w}}$, $\bm{y'}$ as above.
		%For any nontrivial $w\in \mathcal{C}_2$, we define $Z_w = \{\tilde{z}\big|_{\mathrm{supp}(w)}: \epsilon_{\tilde{z}} E\left(0,\tilde{z}\right) \in S \text{ and } \tilde{z} \preceq w\}$ and $B_w = \{v\in Z_2: \epsilon_v = 1\}$. Let $\tilde{Z_w} \in \F_2^n $ (resp. $\tilde{B_w} \in \F_2^n$) by adding all the zero entries back on $Z_w$ (resp. $\tilde{B_w}$). 
		Recall that $\mathrm{proj}_{\bm{w}}(\tilde{\mathcal{B}}_{\bm{w}}^\perp) = \langle \mathrm{proj}_{\bm{w}}(\tilde{\mathcal{Z}}_{\bm{w}}^\perp), \bm{y}' \rangle$. For any $\theta$ and any nontrivial $\bm{w}\in \mathcal{C}_2$,
		\begin{align}\label{eqn:Z_w}
		&1= \nonumber \\
		&\frac{1}{\left|\mathrm{proj}_{\bm{w}}(\tilde{\mathcal{Z}}_{\bm{w}}^\perp)\right|}
		\sum_{\bm{v}\in \mathrm{proj}_{\bm{w}}(\tilde{\mathcal{Z}}_{\bm{w}}^\perp) + \bm{y'}} \left(e^{i\theta}\right)^{w_H(\bm{w})-2w_H(\bm{v})},
		\end{align}
		if and only if 
		\begin{equation}\label{eqn:sec}
		\sum_{\bm{v}\in \mathcal{Z}_{\bm{w}}} \epsilon_{\bm{v}}\left(\imath \tan\theta \right)^{w_H(\bm{v})} = \left(\sec\theta\right)^{w_H(\bm{w})}.
		\end{equation}
	\end{lemma}
	\begin{proof}
	See Appendix \ref{subsec:proof_connect_lem_3}.
	\end{proof}
    \begin{theorem} \label{thm:connect_to_tri}
    The unitary $R_Z(\theta)$ realizes a logical operation on the codespace $V(S)$ of an $\llbr n,k,d \rrbr $ CSS($X,C_2;Z,C_1^\perp$) code if and only if for all non-zero $\bm{w}\in \mathcal{C}_2$,
		\begin{equation}\label{eqn:ori_cond}
	    \sum_{\bm{v}\in \mathcal{Z}_{\bm{w}}} \epsilon_{\bm{v}}\left(\imath \tan\theta \right)^{w_H(\bm{v})} = \left(\sec\theta\right)^{w_H(\bm{w})}.
		\end{equation}
    \end{theorem}
    	\begin{proof}
    	By Lemma \ref{lemma:space_equal}, we know \eqref{eqn:D_w} equals \eqref{eqn:Z_w}. It now follows from Lemma \ref{lemma:simplify} and Lemma \ref{lemma:Mac1} that \eqref{eqn:preserved_by_Uz} equals \eqref{eqn:ori_cond}.
		It then follows directly from Theorem \ref{thm:preserved_by_Uz}.
	\end{proof}
	
	\begin{remark} \normalfont
	\label{rem:local_global_conds}
	Rengaswamy~\cite[Theorem 17]{rengaswamy2020optimality} derived a pair of necessary and sufficient conditions for a CSS code to be invariant under $R_Z(\frac{\pi}{2^l})$. Theorem \ref{thm:connect_to_tri} shows that the first of these conditions implies the second and also generalizes the first condition to arbitrary angle $\theta$. Note that the trigonometric conditions are local whereas the square sum constraint on generator coefficients is global. 
		%The first condition can implies the second condition in Thm17 (See 2/4 note).
	\end{remark}

	\section{Conclusion} % and Some Open Problems}
	\label{sec:concln}
	We have introduced a framework that describes the process of preparing a code state, applying a diagonal physical gate, measuring a code syndrome, and applying a Pauli correction. We have described the interaction of code states and physical gates in terms of generator coefficients determined by the induced logical operator, and have shown that this interaction depends strongly on the signs of $Z$-stabilizers in a CSS code. We have derived necessary and sufficient conditions for a diagonal gate to preserve the code space of a CSS code, and have provided an explicit expression of its induced logical operator. When the diagonal gate is a transversal $Z$-rotation through an angle $\theta$, we derived a simple global condition that can be expressed in terms of divisibility of weights in the two classical codes that determine the CSS code. When all signs in the CSS code are positive, we have proved the necessary and sufficient conditions for Reed-Muller component codes to construct families of CSS codes invariant under transversal Z-rotation through $\pi/2^l$. It remains open to investigate the constraints for a CSS code determined by two classical decreasing monomial codes to be invariant under transversal $\pi/2^l$ $Z$-rotation.   
	
	The generator coefficient framework provides a tool to analyze the evolution under any given diagonal gate of stabilizer codes with arbitrary signs, and we are working to characterize more valid CSS codes can be used in magic state distillation.
    % the value of this degree of freedom in constructing non-Clifford logical gates. 
    % We are also working to connect concatenation of code states with lifting of logical operators.

	\section*{Acknowledgement}
	We would like to thank Ken Brown, Dripto Debroy, and Felice Manganiello for helpful discussions. Ken Brown suggested we look at the method of simulating coherent noise for surface codes described in \cite{bravyi2018correcting}, and this led to our generator coefficient framework, and its use in analyzing the average logical channel. Dripto Debroy encouraged us to interpret the decoherence-free subspace appearing in Example \ref{examp10} in terms of entanglement of initial states and syndrome measurements. Felice Manganiello shared his construction of CSS codes with RM components that are preserved by a transversal $T$ gate, and this led to our construction of CSS codes with RM components that are preserved by transversal $\pi/2^l$ $Z$-rotations.
	
	The work of the authors was supported in part by NSF under grants CCF-1908730 and CCF-2106213. 

\bibliographystyle{plainurl}
\bibliography{bibliography}
% \begin{thebibliography}{9}
% \bibitem{examplecitation}
%   Name Surname,
%   \href{https://doi.org/10.22331/
%         idonotexist}{Quantum
%         \textbf{123}, 123456 (1916).}

% \bibitem{biblatexsubmittingtothearxiv}
%   StackExchange discussion on \href{http://tex.stackexchange.com/questions/26990/biblatex-submitting-to-the-arxiv}{``Biblatex: submitting to the arXiv'' (2017-01-10)}

% \bibitem{arxivpdfoutput}
%   Help article published by the arXiv on \href{https://arxiv.org/help/submit_tex}{``Considerations for TeX Submissions'' (2017-01-10)}

% \bibitem{howtogetdoilinksinbibliography}
%   StackExchange discussion on \href{http://tex.stackexchange.com/questions/3802/how-to-get-doi-links-in-bibliography}{``How to get DOI links in bibliography'' (2016-11-18)}
  
% \bibitem{automaticallyaddingdoifieldstoahandmadebibliography}
%   StackExchange discussion on \href{http://tex.stackexchange.com/questions/6810/automatically-adding-doi-fields-to-a-hand-made-bibliography}{``Automatically adding DOI fields to a hand-made bibliography'' (2016-11-18)}
% \end{thebibliography}

\onecolumn\newpage
\appendix

    \section{Magic State Distillation Using the Steane Code}
    \label{sec:MSD_Steane}
    %The main intuition here is we can distribute the strength of error-correction (the distance for correcting $Z$-errors) to either correct the stochastic errors or the desired logical gate. 
    We use the Steane code as an example to show the trade-off between fidelity and the probability of success in magic state distillation. Classical magic state distillation post-selects on the trivial syndrome without considering the error correction. If we follow this procedure, then the $\llbr 7,1,3\rrbr$ Steane code can be used to distill the state with linear convergence as described in Case 1. In Case 2, we try to increase the probability of success by introducing error-correction instead of post-selecting on the trivial syndrome. In Case 3, we consider only correcting one of non-trivial syndromes. 
    % that the output error rate does not below the line $y=x$ for the positive $x$-coordinate. Thus, we mentioned that if we introduce the error correction to increase the probability of success, the distillation process no longer converges. 

% \begin{itemize}
    % \item 
    Case 1: Reichardt \cite{reichardt2005quantum} calculated error rate by tracking evolution of code states. The generator coefficient framework makes it possible to calculate the output error rate by tracking operators.
    % reproduce the result of Reichardt 2005 for the Steane code using the framework. - actually the result is different!
	\begin{enumerate}[(i)]
		\item Encode to get the $\ket{\overline{+}}$ of the Steane codestate.
		\item Given seven copies of $\ket{A} := T|+\rangle = (|0\rangle + e^{\imath \pi/4}|1\rangle)/\sqrt{2}$ and ancillary qubits, we can realize the phsyical transversal $T^{\otimes7} = [\exp(-\imath\frac{\pi}{8}Z)]^{\otimes 7}$ with the help of Clifford gates and Pauli measurements. If the states $\ket{A}$ are exact, the probability of observing the trivial syndrome is $p^e_{\bm{\bm{\mu}=\bm{0}}}=\frac{9}{16}$ and the probability of observing each non-trivial syndrome is $p^e_{\bm{\bm{\mu}\neq\bm{0}}}=\frac{1}{16}$ (Take $\theta=\frac{\pi}{4}$ in \eqref{eqn:steane_probs}). When the trivial syndrome is observed, it follows from Example \ref{examp8} that the induced logical operator is $T^\dagger_L=\exp(\imath\frac{\pi}{8}Z_L)$. We then apply a physical representation of the logical $P$hase gate $\overline{P}$ to obtain $\ket{\overline{A}}=P_L T^\dagger_L \ket{\overline{+}}$. 
		In practice, each of the input magic states $\ket{A}$ is noisy. We assume dephasing noise: $\rho \to (1-p)\rho +p Z\rho Z$ with the same probability $p$ of a Pauli $Z$ error for each of the seven physical qubits. The probability of observing the trivial syndrome involves two terms. The first term captures the event that upon observing the trivial syndrome $\bm{\mu}= \bm{0}$, the dephasing error is undetectable. The second term captures the event that upon observing the non-trivial syndrome $\bm{\mu}\neq \bm{0}$, the dephasing error cancels the observed syndrome. The probability of success is given by
% 		$Z$-errors become
		\begin{align}
		P_{\bm{\mu}=\bm{0}} 
		&=p^e_{\bm{\bm{\mu}=\bm{0}}}P(Z\text{-error in } \mathcal{C}_2^\perp) + \sum_{\bm{\mu}\neq \bm{0}}p^e_{\bm{\bm{\mu}}} P(Z\text{-error in } \mathcal{C}_2^\perp+\bm{\mu})\\
		&=\frac{9}{16}\sum_{\bm{v} \in \mathcal{C}_2^\perp}(1-p)^{7-w_H(\bm{v}) }p^{w_H(\bm{v})} +  
		\sum_{\bm{\mu}\neq \bm{0}}\frac{1}{16}\sum_{\bm{v}\in \mathcal{C}_2^\perp+\bm{\mu}}(1-p)^{7-w_H(\bm{v}) }p^{w_H(\bm{v})}\\
 		&=\frac{9}{16}\frac{1}{|\mathcal{C}_2|}\sum_{\bm{v}\in \mathcal{C}_2}(1-2p)^{w_H(\bm{v}) } + \frac{7}{16}\frac{1}{|\mathcal{C}_2|}\sum_{\bm{v}\in \mathcal{C}_2}(-1)^{\bm{v}\bm{e_1}^T}(1-2p)^{w_H(\bm{v}) }\\
% 		&= \frac{1}{128}(9+63(1-2p)^4 + 7 - 7 (1-2p)^4 )\\ 
		&= \frac{1}{16}\left(2+7(1-2p)^4\right). \label{eqn:steane_ps}
		\end{align}
		Note that the $7$ cosets corresponding to non-trivial syndromes have identical weight enumerators.
		\item If we observe the non-trivial syndrome $\bm{\mu}\neq 0$, we declare failure and restart. Upon observing the trivial syndrome, we decode and the output mixed state is
		\begin{equation}\label{eqn:steane_output_state}
		\rho_{out} = \frac{1}{P_{\bm{\mu}=\bm{0}}}(p_{out}^0\ket{A}\bra{A} + p_{out}^1Z\ket{A}\bra{A} Z)
		\end{equation}
		where
		\begin{align}
		p_{out}^0
		&=p^e_{\bm{\bm{\mu}=\bm{0}}}P(Z\text{-error in } \mathcal{C}_1^\perp) + \sum_{\bm{\mu}\neq \bm{0}}p^e_{\bm{\bm{\mu}}} P(Z\text{-error in } \mathcal{C}_1^\perp+\bm{\mu}+\bm{\gamma} \text{ for } \bm{\gamma} \neq \bm{0})\\
		&=\frac{9}{16}\sum_{\bm{v}\in {\mathcal{C}_1^\perp}}(1-p)^{n-w_H(\bm{v})}p^{w_H(\bm{v})} + \sum_{\bm{\mu}\neq 0}\frac{1}{16}\sum_{\bm{v}\in {\mathcal{C}_1^\perp+\bm{\mu}+\bm{1}}}(1-p)^{n-w_H(\bm{v})}p^{w_H(\bm{v})}\\
% 		&=
		&=\frac{1}{32}\left(2+7(1-2p)^3+7(1-2p)^4+2(1-2p)^7\right). \label{eqn:steane_c1_0}
		\end{align}
		The first term captures the event that upon observing the the trivial syndrome $\bm{\mu}=\bm{0}$, the dephasing error acts as a $Z$-stabilizer ($B_{\bm{\mu}=\bm{0}} = \frac{3}{4}\bar{T}^\dagger$). The second captures the event that upon observing the the non-trivial syndrome $\bm{\mu}\neq\bm{0}$, the dephasing error lies in $\mathcal{C}_1^\perp +\bm{\mu}+\bm{\gamma}$ ($B_{\bm{\mu}\neq\bm{0}} = \frac{1}{4}\bar{T}^\dagger\bar{Z}$). In this case, the dephasing error appears as the error correction that maps back to the code space and results in a logical $T^\dagger$ gate. 
		We now write the output error rate $q$ as a function of the initial error rate $p$, and calculate its Taylor expansion at $0$
		\begin{equation}
		q(p) = 1-\frac{p_{out}^0}{P_{\bm{\mu}=\bm{0}}} = \frac{7}{9}p + \frac{14}{81}p^2 + O(p^3).
		\end{equation}
		This implies that the threshold for the initial error rate is $0.1464...$ (the same as \cite{reichardt2005quantum}), while that of the $\llbr 15,1,3\rrbr$ code is $0.1415..$ \cite{bravyi2005universal}. 
% 		\item We compare the distillation scheme with that of puntured RM $[[15,1,3]]$ code in the following plot.
% 		The rate of convergence for $[[7,1,3]]$ is only linear while that of $[[ 15,1,3]]$ is cubic. However, the $[[7,1,3]]$ scheme has lower thershold of the initial error rate $p_{th} = 0.1464$ while that of $[[ 15,1,3]]$ is 0.1415.
% 		\begin{figure}[h!]
% 			\centering
% 			\includegraphics[scale=0.3]{MSD_comparsion.png}
% 			\caption{Output error rate v.s. Initial error rate; $[[ 7,1,3]]$ (Blue) vs $[[ 15,1,3]]$ (Black)}
% 		\end{figure}
	\end{enumerate}
	
% 	\item
	Case 2: Note that probability of success in Case 1 is upper bounded by $9/16=56.25\%$. It is natural to ask whether we may introduce error correction to increase the probability of success. It follows from \eqref{eqn:Steane_Kraus} that we can choose proper corrections based on syndromes ($\bm{\gamma_{\mu}}=\bar{Z}$ for $\bm{\mu}\neq \bm{0}$) to obtain the logical operator $T^\dagger$ with probability $1$ if the physical transversal $T$ is exact. The output error-rate now becomes
% 	Now, we consider the input state $\ket{A}$ is noisy under the dephasing noise with probability $p$ again. Then probability of success $P_S =1$ with the output state in \eqref{eqn:steane_output_state}, where
	 \begin{align}
	    q(p) = 1-p^{0}_{out} =
	   %  p^e_{\bm{\bm{\mu}=\bm{0}}}P(Z\text{-error in } \mathcal{C}_1^\perp) + \sum_{\bm{\mu}}p^e_{\bm{\bm{\mu}\neq\bm{0}}}
	    1- P(Z\text{-error in } \mathcal{C}_1^\perp) = \sum_{\bm{v}\in {\mathcal{C}_1^\perp}}(1-p)^{n-w_H(\bm{v})}p^{w_H(\bm{v})} = \frac{1}{8}\left(1+7(1-2p)^4\right).
	 \end{align}
% 	and 
% 	\begin{align}
% 	    q(p) = 1-p^{0}_{out} =  7p - 21p^2 + + O(p^3).
% 	\end{align}
	The output error rate does not fall below the line $y=x$ in the positive orthant, and we say that the protocol does not converge.  
% 	\item 

	Case 3: We balance Case 1 and Case 2 by implementing error correction for only one of the seven non-trivial syndromes, say $\bm{\mu} = \bm{e_1}$. Although the probability of success increases slightly to
	\begin{align}
	    P_S = P_{\bm{\mu}=\bm{0}}+ P_{\bm{\mu}=\bm{e_1}} = \frac{1}{16}\left(2+7(1-2p)^4\right)+\frac{1}{16}\left(2-(1-2p)^4 \right) = \frac{1}{8}\left(2+3(1-2p)^4\right),
	\end{align}

	the prefactor of the linear term of the output error rate is greater than 1. We conclude that the protocol does not converge. 

The same analysis can be performed for a code that is perfectly preserved by the transversal $T$ gate, such as the $\llbr 15,1,3\rrbr$ code. The analysis provides insight into the trade-off between the probability of success and the fidelity of the output magic states. 
% Generator coefficient framework traces the operator 

        \section{Generator Coefficient Framework for Stabilizer codes}
		\label{sec:stab_gcf}
	We described the generator coefficient framework for CSS code and we now extend it to arbitrary stabilizer codes. We consider a general stabilizer code generated by the matrix
	\begin{align} \label{eqn:stabilizer_matrix}
		G_{\mathcal{S}}=\left[\begin{array}{c c}
		K & 0 \\ \hline
		0 & J\\ \hline 
		 \multicolumn{2}{c}{D}\\
	\end{array} \right],
	\end{align}	
	where $D=(D_x,D_z)$ such that $D_x$ is the $X$-component of $D$ and $D_z$ is the $Z$-component of $D$. We assume that the row space of $D$ contains no non-zero vector $\bm{c}=(\bm{c_X},\bm{c_Z})$ with $\bm{c_X}=\bm{0}$ or $\bm{c_Z}=\bm{0}$.
	Assume the dimensions of $K$, $J$, and $D$ are $n_x, n_z,n_{xz}$ respectively. Then, we have 
	\begin{align}
		\Pi_{\mathcal{S}} = \Pi_{\mathcal{S}_X} \Pi_{\mathcal{S}_Z} \Pi_{\mathcal{S}_{XZ}} ,
	\end{align} 
	where
	\begin{equation}
		\Pi_{\mathcal{S}_X} = \frac{1}{2^{n_x}}  \sum_{\bm{a}\in \mathcal{K}=\langle K \rangle } \epsilon_{(\bm{a},\bm{0})} E(\bm{a},\bm{0}), ~%\text{ and }
% 	\end{equation}
% 	\begin{equation}
	\Pi_{\mathcal{S}_Z} =\frac{1}{2^{n_z}}  \sum_{\bm{b}\in \mathcal{J} = \langle J \rangle } \epsilon_{(\bm{0},\bm{b})} E(\bm{0},\bm{b}), \text{ and }
	\end{equation}
	\begin{equation}
	\Pi_{\mathcal{S}_{XZ}} =\frac{1}{2^{n_{xz}}}  \sum_{(\bm{c},\bm{d})\in \mathcal{D} = \langle D \rangle} \epsilon_{(\bm{c},\bm{d})} E(\bm{c},\bm{d}).
	\end{equation}
	Let $\mathcal{T}\coloneqq \langle K,D_x\rangle$. Then,  $\mathcal{J}\subset \mathcal{T}^\perp \subset \F_2^n$ as described below.
	\begin{center}
		\begin{tikzpicture}
		\node (Z) at (0,0) {$\{ \bm{0} \}$};
		\node (C2) at (0,1) {$\mathcal{C}_2$};
		\node (C1) at (0,2) {$\mathcal{C}_1$};
		\node (F2m) at (0,3) {$\mathbb{F}_2^{n}$};
		
		%\path[draw] (Z) -- (C2) node[midway,left] {{$R_1 t$}} -- (C1) node[midway,left] {{$(R_2-R_1)t$}} -- (F2m) node[midway,left] {{$(1-R_2)t$}};
		\path[draw] (Z) -- (C2)  -- (C1) -- (F2m);
		%\path[draw,<->,black] (0.4,0) -- (0.4,1.5) node [midway,right] {$k$}; 
		
		\node (Zp) at (1.5,0) {$\{ \bm{0} \}$};
		\node (C1p) at (1.5,1) {$\mathcal{C}_1^{\perp}$};
		\node (C2p) at (1.5,2) {$\mathcal{C}_2^{\perp}$};
		\node (F2m) at (1.5,3) {$\mathbb{F}_2^{n}$};
% 		\path[draw] (Zp) -- (C1p) -- (C2p) node[midway,right] {{$k$}} -- (F2m); 
        \path[draw] (Zp) -- (C1p) -- (C2p)  -- (F2m); 
		%\path[draw] (Zp) -- (C1p) node[midway,right] {{$(1-R_2)t$}} -- (C2p) node[midway,right] {{$(R_2-R_1)t$}} -- (F2m) node[midway,right] {{$R_1 t$}};
		%%%%%%%%%%%%%%%%%%%
		\node (Z) at (5.7,0) {$\{ \bm{0} \}$};
		\node (C2) at (5.7,1) {$\mathcal{K}$};
		\node (C1) at (5.7,2) {$\langle J,D_z \rangle^\perp$};
		\node (F2m) at (5.7,3) {$\mathbb{F}_2^{n}$};
		
		\path[draw] (Z) -- (C2)  -- (C1)  -- (F2m); 
		
		\node (Zp) at (7.9,0) {$\{ \bm{0} \}$};
		\node (C1p) at (7.9,1) {$\mathcal{J}$};
		\node (C2p) at (7.9,2) {$\mathcal{T}^\perp =\langle K,D_x\rangle ^\perp$};
		\node (F2m) at (7.9,3) {$\mathbb{F}_2^{n}$};
		
		\path[draw] (Zp) -- (C1p) -- (C2p) -- (F2m); 
		
		\draw [dashed,black] (3.5,0) -- (3.5,3)
		node[at end,right] {{Stabilizer}}
		node[at end,left] {{CSS}};
        	\node (mu) at (2.3,2.5) {$\bm{\mu}$};
			\node (frommu) at (1.9,2.5) {};
			\node (tomu) at (1.5,2.5) {};
			\node (gamma) at (2.3,1.5) {$\bm{\gamma}$};
			\node (fromgamma) at (1.9,1.5) {};
			\node (togamma) at (1.5,1.5) {};
			\draw[->,black] (frommu) to [out=10,in=-10] (tomu);
			\draw[->,black] (fromgamma) to [out=10,in=-10] (togamma);
			\node (mu_n) at (8.6,2.5) {$\bm{\mu}$};
			\node (frommu_n) at (8.2,2.5) {};
			\node (tomu_n) at (7.9,2.5) {};
			\node (gamma_n) at (8.6,1.5) {$\bm{\gamma}$};
			\node (fromgamma_n) at (8.2,1.5) {};
			\node (togamma_n) at (7.9,1.5) {};
			\draw[->,black] (frommu_n) to [out=10,in=-10] (tomu_n);
			\draw[->,black] (fromgamma_n) to [out=10,in=-10] (togamma_n);
		\end{tikzpicture}
	\end{center}
	
	Then \eqref{eqn:Pi_SZ_U_Z} becomes
	\begin{align*}
	\Pi_{\mathcal{S}_Z}U_Z&=\left( \frac{1}{2^{n_z}} \sum_{\bm{b}\in \mathcal{J}} \epsilon_{(\bm{0},\bm{b})} E(\bm{0},\bm{b})\right) \left( \sum_{\bm{v}\in \F_2^n}f(\bm{v})E(\bm{0},\bm{v})\right)\\
% 	&= \frac{1}{2^{n_z}} \sum_{v\in \F_2^n}f(n,v)\sum_{b\in B} \epsilon_{(0,b)} E(0,b\oplus v)\\
% 	&= \frac{1}{2^{n_z}} \sum_{v\in \F_2^n}\epsilon_{(0,v)}f(n,v)\sum_{u\in B+v} \epsilon_{(0,u)} E(0,u)\\
	&= \frac{1}{2^{n_z}} \sum_{\bm{\mu}\in \F_2^n/\mathcal{T}^\perp}\sum_{\bm{\gamma}\in \mathcal{T}^\perp/\mathcal{J}}\left(\sum_{\bm{z} \in \mathcal{J}+\bm{\mu}+\bm{\gamma}}\epsilon_{(\bm{0},\bm{v})}f(\bm{z})\right)\sum_{\bm{u}\in \mathcal{J}+\bm{\mu}+\bm{\gamma}} \epsilon_{(\bm{0},\bm{u})} E(\bm{0},\bm{u}), \numberthis \label{eqn:stab_Pi_SZ_UZ}
	\end{align*}
	and the generator coefficients of $U_Z$ for the stabilizer code $\mathcal{S}$ are given by
	\begin{equation}
	A_{\bm{\mu},\bm{\gamma}}^{\mathcal{S}}\coloneqq \sum_{\bm{z}\in \mathcal{J}+\bm{\mu}+\bm{\gamma}}\epsilon_{(\bm{0},\bm{z})}f(\bm{z}),
	\end{equation}
	where $\bm{\mu}\in \F_2^n/\mathcal{T}^\perp$ and $\bm{\gamma}\in \mathcal{T}^\perp/\mathcal{J}$.
    These generalized generator coefficients inherit the properties described in Theorem \ref{thm:GCs_prop}, that is,
    	\begin{align}\label{eqn:sum_squre_general}
		\sum_{{\bm{\mu}} \in \F_2^n/\mathcal{T}^\perp} \sum_{\bm{\gamma} \in \mathcal{T}^\perp/\mathcal{J}} \overline{A_{\bm{\mu},\bm{\gamma}}^{\mathcal{S}}} A_{\bm{\mu},\bm{\eta} \oplus \bm{\gamma}}^{\mathcal{S}} = 
		\left\{ \begin{array}{lc}
		1 ~~ \text{ if } {\bm{\eta}} = \mathbf{0}, \\
		0~~ \text{ if } {\bm{\eta}} \neq \mathbf{0},
		\end{array}\right.
		\end{align}
		for $\bm{\eta} \in \mathcal{T}^\perp / \mathcal{J}$. 
    Grouping together the projectors $\Pi_{\mathcal{S}_X}$ and $\Pi_{\mathcal{S}_{XZ}}$, we consider the new family of projectors
    \begin{align*}
        \mathcal{L} 
        &\coloneqq \Pi_{\mathcal{S}_{X}}\Pi_{\mathcal{S}_{XZ}}\\
        &= \left(\frac{1}{2^{n_x}}  \sum_{\bm{a}\in \mathcal{K}=\langle K \rangle } \epsilon_{(\bm{a},\bm{0})} E(\bm{a},\bm{0})\right) \left( \frac{1}{2^{n_{xz}}}  \sum_{(\bm{c},\bm{d})\in \mathcal{D} = \langle D \rangle} \epsilon_{(\bm{c},\bm{d})} E(\bm{c},\bm{d})\right)\\
        &= \frac{1}{2^{n_x + n_{xz}}} \sum_{\substack{\bm{a}\in \mathcal{K},\\ (\bm{c},\bm{d})\in \mathcal{D}}} \epsilon_{(\bm{a}\oplus \bm{c})} \imath^{-\bm{a}\bm{d}^T} (-1)^{\bm{d}(\bm{a}*\bm{c})^T} E(\bm{a}\oplus \bm{c},\bm{d}).\numberthis \label{eqn:L_expand}
    \end{align*} 
    For $\bm{\mu} \in \F_2^n / \mathcal{T}^\perp$, we write
    \begin{align}
        \mathcal{L}_{(\bm{\mu})} \coloneqq \left(\frac{1}{2^{n_x}}  \sum_{\bm{a}\in \mathcal{K}=\langle K \rangle } (-1)^{\bm{\mu a}^T} \epsilon_{(\bm{a},\bm{0})} E(\bm{a},\bm{0})\right) \left( \frac{1}{2^{n_{xz}}}  \sum_{(\bm{c},\bm{d})\in \mathcal{D} = \langle D \rangle} (-1)^{\bm{\mu c}^T} \epsilon_{(\bm{c},\bm{d})} E(\bm{c},\bm{d})\right),
    \end{align}
    and note that $\{\mathcal{L}_{(\bm{\mu})}\}_{\bm{\mu} \in \F_2^n / \mathcal{T}^\perp}$ is a resolution of identity.
    
    Replacing the resolution of identity $\{\Pi_{\mathcal{S}_X(\bm{\mu})}\}_{\bm{\mu}\in \F_2^n / \mathcal{C}_2^\perp}$ by $\{\mathcal{L}_{(\bm{\mu})}\}_{\bm{\mu} \in \F_2^n / \mathcal{T}^\perp}$, we conclude that the generator coefficients $\{A_{\bm{\mu},\bm{\gamma}}^{\mathcal{S}}\}_{\bm{\mu}\in \F_2^n/\mathcal{T}^\perp, \bm{\gamma}\in \mathcal{T}^\perp /\mathcal{J}}$ describe the same average logical channel as in \eqref{eqn:kraus_rep} and \eqref{eqn:kraus_ops} since the logical Pauli $Z$ for stabilizer codes can be chosen as $\bm{\gamma} \in \mathcal{T}^\perp / \mathcal{J}$ up to a sign. Based on the description of the average logical channel, we study the conditions for the invariance of a stabilizer code as below. %The necessary and sufficient conditions for the invariance of stabilizer codes are also the same. 
	
	\begin{theorem}\label{thm:preserved_by_Uz_stab}
	     Consider a general stabilizer code defined by \eqref{eqn:stabilizer_matrix}. Consider $\mathcal{T}=\langle K,H_x\rangle $, and we have $\mathcal{J}\subset \mathcal{T}^\perp \subset \F_2^n$.  Then, a $Z$-unitary gate $U_Z=\sum_{\bm{v}\in \F_2^n}f(\bm{v})E(\bm{0},\bm{v})$ preserves $\mathcal{V}(\mathcal{S})$ (i.e. $U_Z\Pi_\mathcal{S} U_Z^\dagger=\Pi_\mathcal{S}$) if and only if 
		\begin{equation}\label{eqn:preserved_by_UZ_stab}
		\sum_{\bm{\gamma}\in \mathcal{T}^\perp/\mathcal{J}} |A_{\bm{0},\bm{\gamma}}^{\mathcal{S}}|^2=1.
% 		\sum_{{\gamma}\in  T^\perp/B}
% 		\sum_{w\in B}\epsilon_{(0,w)}
% 		\left(\sum_{z\in B + \gamma}f(n,z)\overline{f(n,z\oplus w)} \right)=1.
		\end{equation}
	\end{theorem}
	\begin{proof}
% 	\color{blue}
		$\Leftarrow$: We assume \eqref{eqn:preserved_by_UZ_stab} holds and derive $U_Z\Pi_{\mathcal{S}}=\Pi_{\mathcal{S}}U_Z$. It follows from \eqref{eqn:sum_squre_general} that $A_{\bm{\mu},\bm{\gamma}}^{\mathcal{S}}=0$ when $\bm{\mu}\neq \bm{0}$. Then, by \eqref{eqn:stab_Pi_SZ_UZ}, we have
		\begin{equation}\label{eqn:leftdirection1}
		U_Z\Pi_{\mathcal{S}_Z}=\Pi_{\mathcal{S}_Z}U_Z=\frac{1}{2^{n-k_1}}\sum_{\bm{\gamma}\in \mathcal{T}^\perp/\mathcal{J} } A_{\bm{0},\bm{\gamma}}^{\mathcal{S}} \left( \sum_{\bm{u}\in \mathcal{C}_1^\perp+\bm{\gamma}}\epsilon_{(\bm{0},\bm{u})}E(\bm{0},\bm{u})\right).
		\end{equation}
		For any $\bm{\gamma}\in \mathcal{T}^\perp/\mathcal{J}$ and $\bm{u}\in \mathcal{C}_1^\perp+\bm{\gamma} \subset \mathcal{T}^\perp$, we have $E(\bm{0},\bm{u})\mathcal{L}=\mathcal{L}E(\bm{0},\bm{u})$, where $\mathcal{L}=\Pi_{\mathcal{S}_{X}}\Pi_{\mathcal{S}_{XZ}}$. Hence, 
		\begin{align}
		U_Z\Pi_{\mathcal{S}}&=U_Z\Pi_{\mathcal{S}_Z}\mathcal{L} %\\
% 		&= \frac{1}{2^{n-k_1}}\sum_{\bm{\gamma}\in \mathcal{T}^\perp/\mathcal{J}} A_{\bm{0},\bm{\gamma}}^{\mathcal{S}} \left( \sum_{\bm{u}\in \mathcal{C}_1^\perp+\bm{\gamma}}\epsilon_{(\bm{0},\bm{u})}E(\bm{0},\bm{u})\mathcal{L}\right)\\
% 		&=\frac{1}{2^{n-k_1}}\sum_{\bm{\gamma}\in \mathcal{T}^\perp/\mathcal{J}} A_{\bm{0},\bm{\gamma}}^{\mathcal{S}}\left( \sum_{\bm{u}\in \mathcal{C}_1^\perp+\bm{\gamma}}\epsilon_{(\bm{0},\bm{u})}\mathcal{L} E(\bm{0},\bm{u})\right)\\
		= \mathcal{L}U_Z\Pi_{\mathcal{S}_Z}=\mathcal{L}\Pi_{\mathcal{S}_Z}U_Z=\Pi_{\mathcal{S}}U_Z.
		\end{align}
		
		$\Rightarrow$: We assume $U_Z\Pi_{\mathcal{S}}=\Pi_{\mathcal{S}}U_Z$ and show \eqref{eqn:preserved_by_UZ_stab}. The idea is the same as in the proof of Theorem \ref{thm:preserved_by_Uz}, and it remains to show that each term in \eqref{eqn:L_expand} is distinct in order to use the independence of Pauli matrices. 
		Assume $(\bm{a}\oplus\bm{c},\bm{d}) = (\bm{a}' \oplus \bm{c}',\bm{d}')$ for some $\bm{a},\bm{a}' \in \mathcal{K}$ and $(\bm{c},\bm{d}), (\bm{c}',\bm{d}') \in \mathcal{D}$. Then, $\bm{d}=\bm{d}'$ and $\bm{a} \oplus \bm{c} = \bm{a}' \oplus \bm{c}'$. %, which implies that $\bm{a}\oplus \bm{a}' = \bm{c} \oplus \bm{c}'$. 
		Note that $(\bm{c},\bm{d}) \oplus (\bm{c}',\bm{d}') = (\bm{c}\oplus\bm{c}',\bm{0}) \in \mathcal{D}$. Since $J\cap D_x = \{\bm{0}\}$, we have $\bm{c}\oplus\bm{c}'=\bm{0}$, which means $\bm{c}=\bm{c}'$ and $\bm{a}=\bm{a}'$.
	\end{proof}
	
	\begin{theorem}
	\label{thm:opt_CSS_among_non_degen}
	     Consider an $\llbr n,k,d \rrbr $ stabilize code generated by the matrix 
	     $
		G_{\mathcal{S}}=\left[\begin{array}{c c}
		K & 0 \\ \hline
		0 & J\\ \hline 
		 \multicolumn{2}{c}{D}\\
	\end{array} \right]$ that satisfies Theorem \ref{thm:preserved_by_Uz_stab}. Let $\mathcal{J}$ be the space defined by the generator matrix $J$. Assume the minimum weight in $\mathcal{J}$ is at least $d$ (i.e. $\min_{\bm{z}\in \mathcal{J}}w_H(z)\geq d$). Then the CSS code generated by 
	$
		G_{\mathcal{S'}}=\left[\begin{array}{c c}
		K & 0 \\ \hline
		0 & J\\ \hline 
		D_x & 0 \\
	\end{array} \right]$ satisfies Theorem \ref{thm:preserved_by_Uz}. Moreover, the CSS code has parameters $n'=n$, $k'=k$, and the $Z$-distance $d'_Z=\min_{\bm{z}\in\langle K,D_x\rangle^\perp \setminus \mathcal{J}}w_H(\bm{z})\ge d$. 
	\end{theorem}
	\begin{proof}
	    From the construction of $G_{\mathcal{S'}}$, the number of physical qubits does not change ($n'=n$). Also, $k'=k$ follows from the fact that $D_x\cap K=\{\bm{0}\}$. It remains to show that the new $Z$-distance $d'_Z \ge d$.
	    
	    Assume there exists $(\bm{s},\bm{t})\in \mathcal{N}(\mathcal{S'})\setminus \mathcal{S'}$ such that $h(\bm{s},\bm{t})<d$ and $\bm{t}\neq \bm{0}$, where $h$ is the Pauli weight (number of nontrivial Pauli matrices) defined by 
	    \begin{equation}\label{eqn:pauliweight}
	        h(\bm{s},\bm{t})=w_H(\bm{s})+w_H(\bm{t})-w_H(\bm{s}*\bm{t}).
	    \end{equation}
	    
	    Then, $h(\bm{0},\bm{t})<d$ and $\bm{t}\in M^\perp \cap D_x^\perp$, which implies that $(\bm{0},\bm{t})\in \mathcal{N}(S)$. Also by definition, we have $J\cap D_z=\{\bm{0}\}$ and thus $(\bm{0},\bm{t})\in \mathcal{N}(\mathcal{S})\setminus \mathcal{S}$. However, by assumption the distance of $\mathcal{V}(\mathcal{S})$ is $d$ and thus $\mathcal{N}(\mathcal{S})\setminus \mathcal{S}$ has minimum weight $d$, which is a contradiction. Therefore, $d'_Z\ge d$. 
	\end{proof}
	
	\begin{remark}
	    \normalfont
	    \label{rem:opt_CSS_diagonal_gates}
	   % \color{blue}
	    Note that the values of generator coefficients are the same for the $\llbr n,k,d \rrbr $ stabilizer code and the $\llbr n'=n,k'=k,d'_Z \ge d \rrbr $ CSS code. The induced logical operator by $U_Z$ remains the same. It follows from Theorem \ref{thm:opt_CSS_among_non_degen} that given an $\llbr n,k,d \rrbr $ non-degenerate stabilizer code supporting a physical $U_Z=\sum_{\bm{v}\in \F_2^n}f(\bm{v})E(\bm{0},\bm{v})$ quantum (unitary) gate, there exists an equivalent CSS code (since the Pauli expansion of the physical gate $U_Z$ has support only on Pauli $Z$, we only compare the distance $d$ of stabilizer code with the $Z$-distance of the equivalent CSS code) supporting the same operation. Note that a similar argument applies to $U_X=\sum_{\bm{v}\in \F_2^n}f(\bm{v})E(\bm{v},\bm{0})$.
	\end{remark}
	
    \section{Proofs for All Results}
    \subsection{Proof of Lemma \ref{lemma:equ_def_gcs}}
    \label{subsec:proof_eqv_GCs_Rz}
	     Setting $\mathcal{B} = \{\bm{z}\in \mathcal{C}_1^\perp \mid \epsilon_{(\bm{0},\bm{z})} = 1\}$, we have $\mathcal{B}^\perp = \langle \mathcal{C}_1,\bm{y}\rangle $. Setting 
	     \begin{equation}
	         S_p = \sum_{\bm{z}\in \mathcal{B} + \bm{\mu} + \bm{\gamma}} \left(\cos\frac{\theta}{2}\right)^{n-w_H(\bm{z})}\left(-\imath \sin\frac{\theta}{2}\right)^{w_H(\bm{z})},
	     \end{equation} 
	     and 
	     \begin{equation}
	         S_n = \sum_{\bm{z} \in \mathcal{C}_1^\perp + \bm{\mu} + \bm{\gamma}} \left(\cos\frac{\theta}{2}\right)^{n-w_H(\bm{z})}\left(-\imath \sin\frac{\theta}{2}\right)^{w_H(\bm{z})},
	     \end{equation}
	     we may rewrite \eqref{eqn:def_gc} as 
		\begin{equation}
		 (-1)^{(\bm{\mu}\oplus \bm{\gamma})\bm{y}^T} A_{\bm{\mu},\bm{\gamma}}(\theta) = 2S_p -S_n.
% 		2\sum_{\bm{z}\in \mathcal{B} + \bm{\mu} + \bm{\gamma}} \left(\cos\frac{\theta}{2}\right)^{n-w_H(\bm{z})}\left(-\imath \sin\frac{\theta}{2}\right)^{w_H(\bm{z})} 
% 		- \sum_{\bm{z} \in \mathcal{C}_1^\perp + \bm{\mu} + \bm{\gamma}} \left(\cos\frac{\theta}{2}\right)^{n-w_H(\bm{z})}\left(-\imath \sin\frac{\theta}{2}\right)^{w_H(\bm{z})}. 
		\end{equation}
		Since $\mathcal{B}+\bm{\mu} +\bm{\gamma} = \langle \mathcal{B}, \bm{\mu} \oplus \bm{\gamma} \rangle \setminus \mathcal{B}$ and $\mathcal{C}_1^\perp + \bm{\mu} + \bm{\gamma} = \langle \mathcal{C}_1^\perp, \bm{\mu} \oplus \bm{\gamma} \rangle \setminus \mathcal{C}_1^\perp$, we have
		\begin{equation}\label{eqn:sum_of_wt_enu}
		 (-1)^{(\bm{\mu}\oplus \bm{\gamma})\bm{y}^T} A_{\bm{\mu},\bm{\gamma}}(\theta)  = 2(P_{\theta}[\langle \mathcal{B}, \bm{\mu} \oplus \bm{\gamma} \rangle] - P_{\theta}[\mathcal{B}]) - (P_{\theta}[\langle \mathcal{C}_1^\perp, \bm{\mu} \oplus \bm{\gamma} \rangle ] - P_{\theta}[\mathcal{C}_1^\perp]).
		\end{equation}
		We may apply the MacWilliams Identities to obtain
		\begin{align*}
		P_{\theta}[\langle \mathcal{B}, \bm{\mu} \oplus \bm{\gamma} \rangle] 
		&= \frac{1}{|\mathcal{B}^\perp \cap (\bm{\mu} \oplus \bm{\gamma})^\perp|} P_{\mathcal{B}^\perp \cap (\bm{\mu} \oplus \bm{\gamma})^\perp} \left(\cos\frac{\theta}{2} - \imath \sin \frac{\theta}{2}, \cos\frac{\theta}{2} + \imath \sin \frac{\theta}{2}\right) \\ &=\frac{1}{|\mathcal{B}^\perp \cap (\bm{\mu} \oplus \bm{\gamma})^\perp|}\sum_{\bm{z}\in \mathcal{B}^\perp \cap (\bm{\mu} \oplus \bm{\gamma})^\perp} \left(\cos\frac{\theta}{2} - \imath\sin\frac{\theta}{2}\right)^{n-2w_H(\bm{z})} \\
		&= \frac{2}{|\mathcal{B}^\perp|} \sum_{\bm{z}\in \mathcal{B}^\perp \cap (\bm{\mu} \oplus \bm{\gamma})^\perp}\left(e^{-\imath\frac{\theta}{2}}\right)^{n-2w_H(\bm{z})},\numberthis\label{eqn:eqv_simp1}
		\end{align*}
		and similarly
		\begin{align}\label{eqn:eqv_simp2}
		P_{\theta}[\mathcal{B}] 
		%&= \frac{1}{|\mathcal{B}^\perp|} P_{\mathcal{B}^\perp} \left(\cos\frac{\theta}{2} - \imath \sin \frac{\theta}{2}, \cos\frac{\theta}{2} + \imath \sin \frac{\theta}{2}\right) \\ % &=\frac{1}{|B^\perp|}\sum_{z\in B^\perp} \left(\cos\frac{\theta}{2} - \imath\sin\frac{\theta}{2}\right)^{n-2w_H(z)} \\
		&= \frac{1}{|\mathcal{B}^\perp|}\sum_{\bm{z}\in \mathcal{B}^\perp}\left(e^{-\imath\frac{\theta}{2}}\right)^{n-2w_H(\bm{z})}.
		\end{align}
		We combine \eqref{eqn:eqv_simp1} and \eqref{eqn:eqv_simp2} to obtain 
		\begin{align*}
		P_{\theta}[\langle \mathcal{B}, \bm{\mu} \oplus \bm{\gamma}\rangle] - P_{\theta}[\mathcal{B}] 
		&= \frac{2}{|\mathcal{B}^\perp|} \sum_{\bm{z}\in \mathcal{B}^\perp \cap (\bm{\mu} \oplus \bm{\gamma})^\perp}\left(e^{-\imath\frac{\theta}{2}}\right)^{n-2w_H(\bm{z})} - 
		\frac{1}{|\mathcal{B}^\perp|}\sum_{\bm{z}\in \mathcal{B}^\perp}\left(e^{-\imath\frac{\theta}{2}}\right)^{n-2w_H(\bm{z})}\\
		&= \frac{1}{|\mathcal{B}^\perp|} \left( \sum_{\bm{z}\in \mathcal{B}^\perp \cap (\bm{\mu} \oplus \bm{\gamma})^\perp}\left(e^{-\imath\frac{\theta}{2}}\right)^{n-2w_H(\bm{z})}  - \sum_{\bm{z}\in \mathcal{B}^\perp \setminus (\bm{\mu} \oplus \bm{\gamma})^\perp}\left(e^{-\imath\frac{\theta}{2}}\right)^{n-2w_H(\bm{z})} \right)\\
		&= \frac{1}{|\mathcal{B}^\perp|} \sum_{\bm{z}\in \mathcal{B}^\perp}(-1)^{(\bm{\mu} \oplus\bm{\gamma})\bm{z}^T} \left(e^{-\imath\frac{\theta}{2}}\right)^{n-2w_H(\bm{z})}.\numberthis \label{eqn:Mac1}
		\end{align*}
		Similarly,
		\begin{equation}
		P_{\theta}[\langle \mathcal{C}_1^\perp, \bm{\mu} \oplus \bm{\gamma}\rangle] - P_{\theta}[\mathcal{C}_1^\perp] = \frac{1}{|\mathcal{C}_1|} \sum_{\bm{z}\in \mathcal{C}_1} (-1)^{(\bm{\mu} \oplus \bm{\gamma})\bm{z}^T} \left(e^{-\imath\frac{\theta}{2}}\right)^{n-2w_H(\bm{z})}.\label{eqn:Mac2}
		\end{equation}
		Since $\mathcal{B}^\perp \setminus \mathcal{C}_1 = \mathcal{C}_1 +\bm{y}$, it follows from \eqref{eqn:sum_of_wt_enu}, \eqref{eqn:Mac1}, \eqref{eqn:Mac2} that 
		\begin{align*}
		 (-1)^{(\bm{\mu}\oplus \bm{\gamma})\bm{y}^T} A_{\bm{\mu},\bm{\gamma}}(\theta) 
		&= \frac{2}{|\mathcal{B}^\perp|} \sum_{\bm{z}\in \mathcal{B}^\perp}(-1)^{(\bm{\mu} \oplus \bm{\gamma})\bm{z}^T} \left(e^{-\imath\frac{\theta}{2}}\right)^{n-2w_H(\bm{z})} - \frac{1}{|\mathcal{C}_1|} \sum_{\bm{z}\in \mathcal{C}_1} (-1)^{(\bm{\mu} \oplus \bm{\gamma})\bm{z}^T} \left(e^{-\imath\frac{\theta}{2}}\right)^{n-2w_H(\bm{z})} \\
		&= \frac{1}{|\mathcal{C}_1|} \sum_{\bm{z}\in \mathcal{C}_1 + \bm{y}}(-1)^{(\bm{\mu} \oplus \bm{\gamma})\bm{z}^T} \left(e^{-\imath\frac{\theta}{2}}\right)^{n-2w_H(\bm{z})}, \numberthis
		\end{align*}
		which completes the proof.\qed
		
   \subsection{Derivation of \eqref{eqn:kraus_simplify}}
    \label{seubsec:deriv_eqn}
    \begin{align*} 
    U_Z \Pi_{\mathcal{S}} 
	&= U_Z\Pi_{\mathcal{S}_Z}\Pi_{\mathcal{S}_X} \\
	&=\frac{1}{2^{n-k_1+k_2}}\sum_{\bm{\mu}\in \F_2^n/\mathcal{C}_2^\perp}\sum_{\bm{\gamma}\in \mathcal{C}_2^\perp/\mathcal{C}_1^\perp} A_{\bm{\mu},\bm{\gamma}}\left( \sum_{\bm{u}\in \mathcal{C}_1^\perp+\bm{\mu}+\bm{\gamma}}\epsilon_{(\bm{0},\bm{u})}E(\bm{0},\bm{u})\right)\left(\sum_{\bm{a}\in \mathcal{C}_2}\epsilon_{(\bm{a},\bm{0})}E(\bm{a},\bm{0})\right)\\
	&=\frac{1}{2^{n-k_1+k_2}}\sum_{\bm{\mu}\in \F_2^n/\mathcal{C}_2^\perp}\sum_{\bm{\gamma}\in \mathcal{C}_2^\perp/\mathcal{C}_1^\perp} A_{\bm{\mu},\bm{\gamma}}\left(\sum_{\bm{a}\in \mathcal{C}_2}(-1)^{\bm{a}\bm{\mu}^T}\epsilon_{(\bm{a},\bm{0})}E(\bm{a},\bm{0})\right)\left( \sum_{\bm{u}\in \mathcal{C}_1^\perp+\bm{\mu}+\bm{\gamma}}\epsilon_{(\bm{0},\bm{u})}E(\bm{0},\bm{u})\right)\\
	&=\frac{1}{2^{n-k_1}}\sum_{\bm{\mu}\in \F_2^n/\mathcal{C}_2^\perp}\Pi_{\mathcal{S}_X(\bm{\mu})} \left(\sum_{\bm{\gamma}\in \mathcal{C}_2^\perp/\mathcal{C}_1^\perp} A_{\bm{\mu},\bm{\gamma}} \left( \sum_{\bm{u}\in \mathcal{C}_1^\perp+\bm{\mu}+\bm{\gamma}}\epsilon_{(\bm{0},\bm{u})}E(\bm{0},\bm{u})\right)\right),
	\numberthis 
	\end{align*}
	where $\Pi_{\mathcal{S}_X(\bm{\mu})} = \frac{1}{|\mathcal{C}_2|}\sum_{\bm{a}\in \mathcal{C}_2} (-1)^{\bm{a}\bm{\mu}^T} \epsilon_{(\bm{a},\bm{0})} E(\bm{a},\bm{0})$. \qed
	
	\subsection{Derivation of $\theta(\theta_L)$}
	\label{sec:log_rot}
	  Since there is only one logical qubit, $\bm{\gamma}$ is either zero or non-zero. It then follows from (76) and (77) that the effective physical operator corresponding to the syndrome $\bm{\mu}=\bm{0}$ is 
    \begin{align}
    B_{\bm{\mu}=\bm{0}} = A_{\bm{\mu}=\bm{0},\bm{\gamma}=\bm{0}}E(\bm{0},\bm{0}) + A_{\bm{\mu}=\bm{0},\bm{\gamma}\neq\bm{0}}E(\bm{0},\bm{\gamma}\neq \bm{0}).
    \end{align}
    Thus, if we observe the trivial syndrome,
    % (for example, the Steane code has 56.25\% to observe the trivial syndrome after applying $R_Z(\frac{\pi}{4})$),
    then the induced logical portion is 
    \begin{align}
    U_Z^L(\bm{\mu}=\bm{0}) = A_{\bm{\mu}=\bm{0},\bm{\gamma}=\bm{0}}I_L + A_{\bm{\mu}=\bm{0},\bm{\gamma}\neq\bm{0}}Z_L = 
    \begin{bmatrix} A_{\bm{0},\bm{\gamma}=\bm{0}}+A_{\bm{0},\bm{\gamma}\neq\bm{0}} & 0 \\ 0 & A_{\bm{0},\bm{\gamma}=\bm{0}}-A_{\bm{0},\bm{\gamma}\neq\bm{0}}\end{bmatrix}.
    \end{align}
    Since we also assume that one of the pair $(A_{\bm{\mu}=\bm{0},\bm{\gamma}=\bm{0}},A_{\bm{\mu}=\bm{0},\bm{\gamma}\neq\bm{0}})$ is real and the other is pure imaginary, we can consider $U_Z^L(\bm{\mu}=\bm{0})$ as a $Z$-rotation with angle $\theta_L$ up to some logical Pauli $Z_L$:
    \begin{align}
    U_Z^L(\bm{\mu}=\bm{0}) = 	\left\{ 
	\begin{array}{lc}
	\cos(\theta_L/2) I_L +\imath\sin(\theta_L/2)Z_L  =R_Z(\theta_L)& \text{if } A_{\bm{\mu}=\bm{0},\bm{\gamma}=\bm{0}} \text{ is real } \\
    \imath\sin(\theta_L/2)I_L + \cos(\theta_L/2) Z_L =Z_LR_Z(\theta_L) & \text{if } A_{\bm{\mu}=\bm{0},\bm{\gamma}\neq\bm{0}} \text{ is real }
	\end{array}\right., 
    \end{align}
    with 
    % \begin{align} 
    $\theta_L/2 =\tan^{-1}\left(\frac{\sin(\theta_L/2)}{\cos(\theta_L/2)}\right) = \tan^{-1}\left(\frac{\imath A_{\bm{\mu}=\bm{0},\bm{\gamma}\neq\bm{0}}}{A_{\bm{\mu}=\bm{0},\bm{\gamma}=\bm{0}}}\right).$
    % \end{align}
	
    \subsection{Proof of Theorem \ref{thm:GCs_prop}}
    \label{subsec:proof_Kraus_ver}
      It follows from \eqref{eqn:def_gc} that 
		\begin{align*}
		\overline{A_{\bm{\mu},\bm{\gamma}}} A_{\bm{\mu},\bm{\eta} \oplus \bm{\gamma}} &=\left(\sum_{\bm{z}\in \mathcal{C}_1^\perp+\bm{\mu}+\bm{\gamma}}\epsilon_{(\bm{0},\bm{z})}f(\bm{z}) \right)\left(\sum_{\bm{z'}\in \mathcal{C}_1^\perp+\bm{\mu}+\bm{\eta}+\bm{\gamma}}\epsilon_{(\bm{0},\bm{z'})} f(n,\bm{z'})\right)\\
		&=\sum_{\bm{w}\in \mathcal{C}_1^\perp+\bm{\eta}}\epsilon_{(\bm{0},\bm{w})}\left(\sum_{\bm{z}\in \mathcal{C}_1^\perp + \bm{\mu} + \bm{\gamma}}f(\bm{z})\overline{f(\bm{z}\oplus \bm{w})} \right).\numberthis
		\end{align*}
    Then, we have 
		\begin{align*}
		\sum_{\bm{\mu}\in \F_2^n / \mathcal{C}_2^\perp} 
		\sum_{\bm{\gamma} \in \mathcal{C}_2^\perp / \mathcal{C}_1^\perp } 
		\overline{A_{\bm{\mu},\bm{\gamma}}} A_{\bm{\mu},\bm{\eta} \oplus \bm{\gamma}}
		&=\sum_{\bm{\mu}\in \F_2^n / \mathcal{C}_2^\perp} 
		\sum_{{\bm\gamma} \in \mathcal{C}_2^\perp / \mathcal{C}_1^\perp }
		\sum_{\bm{w} \in \mathcal{C}_1^\perp+\bm{\eta}}\epsilon_{(\bm{0},\bm{w})}
		\left(\sum_{\bm{z}\in \mathcal{C}_1^\perp + \bm{\mu} + \bm{\gamma}}f(\bm{z})\overline{f(\bm{z}\oplus \bm{w})} \right)\\
		&=\sum_{\bm{w}\in \mathcal{C}_1^\perp+\bm{\eta}}\epsilon_{(\bm{0},\bm{w})}
		\left( \sum_{\bm{\mu}\in \F_2^n / \mathcal{C}_2^\perp} \sum_{\bm{\gamma} \in \mathcal{C}_2^\perp / \mathcal{C}_1^\perp }
		\sum_{\bm{z}\in \mathcal{C}_1^\perp + \bm{\mu} + \bm{\gamma}}f(\bm{z})\overline{f(\bm{z}\oplus \bm{w})} \right)\\
		&=\sum_{\bm{w}\in \mathcal{C}_1^\perp+\bm{\eta}}\epsilon_{(\bm{0},\bm{w})}
		\left( \sum_{\bm{z}\in \F_2^n}f(\bm{z})\overline{f(\bm{z}\oplus \bm{w})} \right)\\
		&=\left\{\begin{array}{lc}
	    \epsilon_{(\bm{0},\bm{0})} = 1 & \text{ if } \bm{\eta}=\bm{0}\\
    	0 & \text{ if } \bm{\eta}\neq \bm{0}
	    \end{array} \right.,\numberthis
		\end{align*}
		where the last step follows from the fact that $U_Z$ is unitary \eqref{eqn:unitarycoefficients}.    \qed
        
        \subsection{Derivation of \eqref{eqn:later_half}}
        \label{subsec:derv_simplify}
        \begin{align*}
	\Pi_{\mathcal{S}_X(\bm{\mu_0})} U_Z\Pi_{\mathcal{S}_Z} \ket{\phi}
	& = \frac{1}{|\mathcal{C}_2|}\sum_{\bm{a}\in \mathcal{C}_2} (-1)^{\bm{a}\bm{\mu_0}^T} \epsilon_{(\bm{a},\bm{0)}} E(\bm{a},\bm{0}) \sum_{\bm{\mu}\in \F_2^n/\mathcal{C}_2^\perp}\sum_{\bm{\gamma}\in \mathcal{C}_2^\perp/\mathcal{C}_1^\perp} A_{\bm{\mu},\bm{\gamma}} \epsilon_{(\bm{0},\bm{\mu} \oplus \bm{\gamma})}E(\bm{0},\bm{\mu} \oplus \bm{\gamma}) \ket{\phi}\\
	& = \frac{1}{|\mathcal{C}_2|}\sum_{\bm{\mu}}\sum_{\bm{\gamma}} 
	A_{\bm{\mu},\bm{\gamma}} \epsilon_{(\bm{0},\bm{\mu} \oplus \bm{\gamma})}E(\bm{0},\bm{\mu} \oplus \bm{\gamma})  \sum_{\bm{a}\in \mathcal{C}_2} (-1)^{\bm{a}(\bm{\mu}+\bm{\mu_0})^T} \epsilon_{(\bm{a},\bm{0})} E(\bm{a},\bm{0}) \ket{\phi}\\
	& = \frac{1}{|\mathcal{C}_2|}\sum_{\bm{\mu}}\sum_{\bm{\gamma}} A_{\bm{\mu},\bm{\gamma}} \epsilon_{(\bm{0},\bm{\bm{\mu} \oplus \bm{\gamma}})}E(\bm{0},\bm{\mu} \oplus \bm{\gamma})  \sum_{\bm{a}\in \mathcal{C}_2} (-1)^{\bm{a} (\bm{\mu}\oplus\bm{\mu_0})^T}  \ket{\phi},\numberthis \label{eqn:prob_simp_append}
	\end{align*}
	where \eqref{eqn:prob_simp_append} follows from the fact $\epsilon_{(\bm{a},\bm{0})} E(\bm{a},\bm{0})  \in \mathcal{S}$. \qed
	
	\subsection{Proof of Theorem \ref{thm:preserved_by_Uz}}
	\label{subsec:proof_thm_preserved_by_Uz}
	    Recall from \eqref{eqn:Pi_SZ_U_Z} that $U_Z\Pi_{S_\mathcal{Z}} = \Pi_{S_\mathcal{Z}}U_Z$ simplifies to
		\begin{align}
		U_Z\Pi_{\mathcal{S}_Z}
		&=\frac{1}{2^{n-k_1}}\sum_{\bm{\mu}\in \F_2^n/\mathcal{C}_2^\perp}\sum_{\bm{\gamma}\in \mathcal{C}_2^\perp/\mathcal{C}_1^\perp} A_{\bm{\mu},\bm{\gamma}}\left( \sum_{\bm{u}\in \mathcal{C}_1^\perp+\bm{\mu}+\bm{\gamma}}\epsilon_{(\bm{0},\bm{u})}E(\bm{0},\bm{u})\right). %\label{eqn:Pi_SZ_U_Z}
		\end{align}
		
		%		For a $\mu \neq 0$ in $\F_2^n/C_2^\perp$, there exists $1\le i_0\le k_2$ such that $(-1)^{c_{i_0}\mu^T}=-1$. Hence,
		%		\begin{equation}
		%		\frac{I+\nu_{(c_{i_0},0)}E(c_{i_0},0)}{2}\cdot \frac{I+(-1)^{c_{i_0}\mu^T}\nu_{(c_{i_0},0)}E(c_{i_0},0)}{2}=\frac{I+\nu_{(c_{i_0},0)}E(c_{i_0},0)}{2}\cdot\frac{I-\nu_{(c_{i_0},0)}E(c_{i_0},0)}{2}=0.
		%		\end{equation}
		%		\begin{equation}
		%		\Pi_{S_X}\Pi_{S(\mu)}=\prod_{i=1}^{k_2}\frac{I+\nu_{(c_{i},0)}E(c_{i},0)}{2}\prod_{i=1}^{k_2}\frac{I+(-1)^{c_{i}\mu^T}\nu_{(c_{i},0)}E(c_{i},0)}{2}\prod_{j=1}^{n-k_1}\frac{I+\nu_{(0,d_j)}E(0,d_j)}{2}=0, \text{ when $\mu \neq 0$ in $\F_2^n/C_2^\perp$}.
		%		\end{equation}
		
		%		\begin{align}
		%		\Pi_{S_X}2^{n-k_1}\Pi_{S_Z}U(\theta)\Pi_S &=\sum_{\mu\in \F_2^n/C_2^\perp}\left\{\Pi_{S_X}\Pi_{S(\mu)}\sum_{\gamma\in C_2^\perp/C_1^\perp} \left[A_{\mu,\gamma}(\theta)\left( \sum_{u\in C_1^\perp+\gamma+\mu}\epsilon_{(0,u)}E(0,u)\right)\right]\right\}\\
		%		&=\sum_{\gamma\in C_2^\perp/C_1^\perp} \left[A_{0,\gamma}(\theta)\left( \sum_{u\in C_1^\perp+\gamma+\mu}\epsilon_{(0,u)}E(0,u)\right)\right] \label{eqn:general_PiVthetaPi}
		%		\end{align}
		
		$\Leftarrow$: We assume \eqref{eqn:preserved_by_Uz} holds and derive $U_Z\Pi_{\mathcal{S}}=\Pi_{\mathcal{S}}U_Z$. By Theorem \ref{thm:GCs_prop}, we have $A_{\bm{\mu},\bm{\gamma}}=0$ when $\bm{\mu}\neq \bm{0}$. It follows from \eqref{eqn:Pi_SZ_U_Z} that
		\begin{equation}%\label{eqn:leftdirection1}
		U_Z\Pi_{\mathcal{S}_Z}=\Pi_{\mathcal{S}_Z}U_Z=\frac{1}{2^{n-k_1}}\sum_{\bm{\gamma}\in \mathcal{C}_2^\perp/\mathcal{C}_1^\perp} A_{\bm{0},\bm{\gamma}}\left( \sum_{\bm{u}\in \mathcal{C}_1^\perp+\bm{\gamma}}\epsilon_{(\bm{0},\bm{u})}E(\bm{0},\bm{u})\right).
		\end{equation}
		For any $\bm{\gamma}\in \mathcal{C}_2^\perp/\mathcal{C}_1^\perp$ and $\bm{u}\in \mathcal{C}_1^\perp+\bm{\gamma} \subset \mathcal{C}_2^\perp$, we have $E(\bm{0},\bm{u})\Pi_{\mathcal{S}_X}=\Pi_{\mathcal{S}_X}E(\bm{0},\bm{u})$. Hence, 
		\begin{align*}
		U_Z\Pi_{\mathcal{S}}=U_Z\Pi_{\mathcal{S}_Z}\Pi_{\mathcal{S}_X}
% 		&= \frac{1}{2^{n-k_1}}\sum_{\bm{\gamma}\in \mathcal{C}_2^\perp/\mathcal{C}_1^\perp} A_{\bm{0},\bm{\gamma}}\left( \sum_{\bm{u}\in \mathcal{C}_1^\perp+\bm{\gamma}}\epsilon_{(\bm{0},\bm{u})}E(\bm{0},\bm{u})\Pi_{\mathcal{S}_X}\right)\\
		&=\frac{1}{2^{n-k_1}}\sum_{\bm{\gamma}\in \mathcal{C}_2^\perp/\mathcal{C}_1^\perp} A_{\bm{0},\bm{\gamma}}\left( \sum_{\bm{u}\in \mathcal{C}_1^\perp+\bm{\gamma}}\epsilon_{(\bm{0},\bm{u})}\Pi_{\mathcal{S}_X}E(\bm{0},\bm{u})\right)\\\
		&= \Pi_{\mathcal{S}_X}U_Z\Pi_{\mathcal{S}_Z}=\Pi_{\mathcal{S}_X}\Pi_{\mathcal{S}_Z}U_Z=\Pi_{\mathcal{S}}U_Z.\numberthis
		\end{align*}
		
		%		By cancelling the factor $2^{n-k_1}$ and multiplying $\Pi_{S_X}$ on both sides of \eqref{eqn:leftdirection1}, we have
		%		\begin{equation}\label{eqn:leftdirection2}
		%		\Pi_{S_X}\Pi_{S_Z}U(\theta)=\Pi_{S_X}\Pi_S\Pi_{S_Z}U(\theta)\Pi_S.
		%		\end{equation}
		%		By the properties of projectors and the stabilizer group $S$, we have $\Pi_{S_X}\Pi_S=\Pi_S\Pi_{S_X}$, $\Pi_{S_X}\Pi_{S_Z}=\Pi_S$, and $\Pi_S^2=\Pi_S$. By \eqref{eqn:leftdirection2}, we have
		%		\begin{equation}
		%		\Pi_S U(\theta)=\Pi_{S_X}\Pi_{S_Z}U(\theta)=\Pi_{S_X}\Pi_S\Pi_{S_Z}U(\theta)\Pi_S=\Pi_S\Pi_{S_X}\Pi_{S_Z}U(\theta)\Pi_S=\Pi_S^2U(\theta)\Pi_S
		%		\end{equation}
		
		$\Rightarrow$: We assume $U_Z\Pi_{\mathcal{S}}=\Pi_{\mathcal{S}}U_Z$ and show \eqref{eqn:preserved_by_Uz}. 
% 		For any $\mu\in \F_2^n/C_2^\perp$, recall 
% 		\begin{equation}
% 		S_X(\mu)=\left\{ (-1)^{a\mu^T}\epsilon_{(a,0)}E(a,0)| a\in C_2 \right\}=\langle (-1)^{c_i\mu^T}\nu_{(c_i,0)}E(c_i,0)|1\le i\le k_2 \rangle.
% 		\end{equation}
% 		Note that $S_X(0)=S_X$. Let $\Pi_{S_X(\mu)}$ be the projector of $S_X(\mu)$. 
        It follows from \eqref{eqn:kraus_simplify} that
		\begin{align*}
		U_Z\Pi_{\mathcal{S}} &=U_Z\Pi_{\mathcal{S}_Z}\Pi_{\mathcal{S}_X} \\
%		&=\frac{1}{2^{n-k_1+k_2}}\sum_{\mu\in \F_2^n/C_2^\perp}\sum_{\gamma\in C_2^\perp/C_1^\perp} \left[A_{\mu,\gamma}(\theta)\left( \sum_{u\in C_1^\perp+\gamma+\mu}\epsilon_{(0,u)}E(0,u)\right)\left(\sum_{a\in C_2}\epsilon_{(a,0)}E(a,0)\right)\right]\\
%		&=\frac{1}{2^{n-k_1+k_2}}\sum_{\mu\in \F_2^n/C_2^\perp}\sum_{\gamma\in C_2^\perp/C_1^\perp} \left[A_{\mu,\gamma}(\theta)\left(\sum_{a\in C_2}(-1)^{a\mu^T}\epsilon_{(a,0)}E(a,0)\right)\left( \sum_{u\in C_1^\perp+\gamma+\mu}\epsilon_{(0,u)}E(0,u)\right)\right]\\
		&=\frac{1}{2^{n-k_1}}\sum_{\bm{\mu}\in \F_2^n/\mathcal{C}_2^\perp}\left(\Pi_{\mathcal{S}_X(\bm{\mu})}\sum_{\bm{\gamma}\in \mathcal{C}_2^\perp/\mathcal{C}_1^\perp} A_{\bm{\mu},\bm{\gamma}}\left( \sum_{\bm{u}\in \mathcal{C}_1^\perp+\bm{\gamma}+\bm{\mu}}\epsilon_{(\bm{0},\bm{u})}E(\bm{0},\bm{u})\right)\right)
		= \Pi_{\mathcal{S}}U_Z.
		\numberthis \label{eqn:rightdirection1}
		\end{align*}
%		On the other hand, by \eqref{eqn:Pi_SZ_U_Z} 
%		\begin{equation}
%		\Pi_{S} U(\theta)=\Pi_{S_X}\Pi_{S_Z}U(\theta)=\frac{1}{2^{n-k_1}}\sum_{\mu\in \F_2^n/C_2^\perp}\left\{\Pi_{S_X}\sum_{\gamma\in C_2^\perp/C_1^\perp} \left[A_{\mu,\gamma}(\theta)\left( \sum_{u\in C_1^\perp+\gamma+\mu}\epsilon_{(0,u)}E(0,u)\right)\right]\right\}.\\
%		\end{equation}
		Pairwise orthogonality of projectors implies $\Pi_{\mathcal{S}_X(\bm{\mu})}\Pi_{\mathcal{S}_X(\bm{\mu'})}=0$ when $\bm{\mu} \neq \bm{\mu'}$ in $\F_2^n/\mathcal{C}_2^\perp$. Hence, for any $\bm{\mu_0}\in \F_2^n/\mathcal{C}_2^\perp\setminus \{\bm{0}\}$, we have 
% 		\begin{equation}\label{eqn:rightdirection2}
% 		\Pi_{S_X(\mu_0)}(\Pi_{S} R_Z(\theta))=\frac{1}{2^{n-k_1+k_2}}\sum_{\mu\in \F_2^n/C_2^\perp}\left\{\Pi_{S_X(\mu_0)}\Pi_{S_X}\sum_{\gamma\in C_2^\perp/C_1^\perp} \left[A_{\mu,\gamma}(\theta)\left( \sum_{u\in C_1^\perp+\gamma+\mu}\epsilon_{(0,u)}E(0,u)\right)\right]\right\}=0.
% 		\end{equation}
% 		By \eqref{eqn:rightdirection1} and \eqref{eqn:rightdirection2}, 
        we have $0=\Pi_{\mathcal{S}_X(\bm{\mu_0})}\Pi_{\mathcal{S}_X}\Pi_{\mathcal{S}_Z} U_Z=\Pi_{\mathcal{S}_X(\bm{\mu_0})}(\Pi_{\mathcal{S}}U_Z) =  \Pi_{\mathcal{S}_X(\bm{\mu_0})}(U_Z\Pi_{\mathcal{S}} )$, which implies that 
		\begin{align*}
		0&=\frac{1}{2^{n-k_1}}
		\sum_{\bm{\mu}\in \F_2^n/\mathcal{C}_2^\perp}\left(\Pi_{\mathcal{S}_X(\bm{\mu_0})}\Pi_{\mathcal{S}_X(\bm{\mu})}
		\sum_{\bm{\gamma}\in \mathcal{C}_2^\perp/\mathcal{C}_1^\perp} A_{\bm{\mu},\bm{\gamma}}
		\left( \sum_{\bm{u}\in \mathcal{C}_1^\perp+\bm{\gamma}+\bm{\mu}}\epsilon_{(\bm{0},\bm{u})}E(\bm{0},\bm{u})\right)\right)\\
		&=\frac{1}{2^{n-k_1}}\Pi_{S_X(\bm{\mu_0})}
		\sum_{\bm{\gamma}\in \mathcal{C}_2^\perp/\mathcal{C}_1^\perp} A_{\bm{\mu_0},\bm{\gamma}}
		\left( \sum_{\bm{u}\in \mathcal{C}_1^\perp+\bm{\gamma}+\bm{\mu_0}}
		\epsilon_{(\bm{0},\bm{u})}E(\bm{0},\bm{u})\right)\\
		&=\frac{1}{2^{n-k_1}}\left(\frac{1}{2^{k_2}}\sum_{\bm{a}\in \mathcal{C}_2} (-1)^{\bm{a}\bm{\mu_0}^T}\epsilon_{(\bm{a},\bm{0})}E(\bm{a},\bm{0})\right) 
		\left(\sum_{\bm{\gamma}\in \mathcal{C}_2^\perp/\mathcal{C}_1^\perp} A_{\bm{\mu_0},\bm{\gamma}}
		\left( \sum_{\bm{u}\in \mathcal{C}_1^\perp+\bm{\gamma}+\bm{\mu_0}}\epsilon_{(\bm{0},\bm{u})}E(\bm{0},\bm{u})
		\right)\right)\\
		&= \frac{1}{2^{n-k_1+k_2}}
		\sum_{\bm{\gamma}\in \mathcal{C}_2^\perp/\mathcal{C}_1^\perp}
		\sum_{\bm{u}\in \mathcal{C}_1^\perp+\bm{\gamma}+\bm{\mu_0}}\sum_{\bm{a}\in\mathcal{C}_2} A_{\bm{\mu_0},\bm{\gamma}} (-1)^{\bm{a}\bm{\mu_0}^T}\imath^{\bm{a}\bm{\mu_0}^T}\epsilon_{(\bm{a},\bm{u})}E(\bm{a},\bm{u}).\numberthis
		\end{align*}
		Since Pauli matrices are linear independent, we have $A_{\bm{\mu_0},\bm{\gamma}}=0$ for all $\bm{\mu}\in\F_2^n/\mathcal{C}_2^\perp\setminus \{\bm{0} \}$ and all $\bm{\gamma}\in \mathcal{C}_2^\perp /\mathcal{C}_1^\perp$, and $\eqref{eqn:preserved_by_Uz}$ holds. \qed
	
	\subsection{Proof of Lemma \ref{lemma:two_div_conds_qfd}}
	\label{subsec:proof_lemma_qfd_div}
	
    		$\Rightarrow$: Assume \eqref{eqn:div_cond_qfd} holds for all $\bm{v_1}, \bm{v_2} \in \mathcal{C}_1+\bm{y}$ such that $\bm{v_1} \oplus \bm{v_2} \in \mathcal{C}_2$. Then, \eqref{eqn:first_div_cond_qfd} is satisfied. Let $\bm{v_1},\bm{v_2}\in (\mathcal{C}_1 +\bm{y})/(\mathcal{C}_2 +\bm{y})$ and $\bm{v_1}\oplus \bm{v_2} \in \mathcal{C}_2$. Then we can write $\bm{v_1} = \bm{u_1}+\bm{w}+\bm{y}$ and $\bm{v_2} = \bm{u_2}+\bm{w}+\bm{y}$ for $\bm{u_1},\bm{u_2}\in \mathcal{C}_2$ and $\bm{w}\in \mathcal{C}_1 /\mathcal{C}_2$. We simplify \eqref{eqn:div_cond_qfd} as
    		\begin{align}
    		2^l & \mid   (\bm{u_1}+\bm{w}+\bm{y}) R (\bm{u_1}+\bm{w}+\bm{y})^T - (\bm{u_2}+\bm{w}+\bm{y}) R (\bm{u_2}+\bm{w}+\bm{y})^T \label{eqn:simpfy_div_cond_qfd1}\\
    		2^l & \mid   \left((\bm{u_1}+\bm{y}) R (\bm{u_1}+\bm{y})^T - (\bm{u_2}+\bm{y}) R (\bm{u_2}+\bm{y})^T \right)+ 2\left((\bm{u_1}+\bm{y}) R \bm{w}^T - (\bm{u_1}+\bm{y}) R \bm{w}^T\right) \label{eqn:simpfy_div_cond_qfd2}\\
    		2^l & \mid   2(\bm{u_1}-\bm{u_2}) R \bm{w}^T, \label{eqn:simpfy_div_cond_qfd3}
    		\end{align} 
    		since $\bm{u_1}+\bm{y}, \bm{u_2}+\bm{y} \in \mathcal{C}_2 +\bm{y}$. Thus, \eqref{eqn:sec_div_cond_qfd} is also satisfied.
    		
    		$\Leftarrow$: We simply reverse \eqref{eqn:simpfy_div_cond_qfd1}, \eqref{eqn:simpfy_div_cond_qfd2}, and \eqref{eqn:simpfy_div_cond_qfd3}. \qed
    		%Since the simplifications in  \eqref{eqn:simpfy_div_cond_qfd1}, \eqref{eqn:simpfy_div_cond_qfd2}, \eqref{eqn:simpfy_div_cond_qfd3} are true reversely, it follows from the same argument that the sufficient direction holds.
    		
    \subsection{Proof of Theorem \ref{thm:div_cond_RZ}}
    \label{subsec:proof_div_Rz}
    The proof idea is the same as that of Theorem \ref{thm:div_cond_qfd}
	We take $U_Z = R_Z\left(\frac{\pi}{p}\right)$ and simplify \eqref{eqn:A_mu,gamma} using \eqref{eqn:preserved_by_Uz}: 
	\begin{align*}
	1
	&=\sum_{\bm{\gamma}\in \mathcal{C}_2^\perp/\mathcal{C}_1^\perp} \left|A_{\bm{0},\bm{\gamma}}\left(\frac{\pi}{p}\right)\right|^2\\
	&=\sum_{\bm{\gamma}\in \mathcal{C}_2^\perp /\mathcal{C}_1^\perp}\frac{1}{|\mathcal{C}_1|^2} \sum_{\bm{z_1},\bm{z_2}\in \mathcal{C}_1+\bm{y}} (-1)^{\bm{\gamma} (\bm{z_1}\oplus\bm{z_2})^T} \left(e^{\imath\frac{\pi}{p}}\right)^{w_H(\bm{z_1})-w_H(\bm{z_2})}.\numberthis
	\end{align*}
	Setting $\bm{w}=\bm{z_1}\oplus\bm{z_2}$ and $\bm{z}=\bm{z_2}$, we obtain
	\begin{align*}
	1
	& = \frac{1}{|\mathcal{C}_1|^2} \sum_{\bm{w}\in \mathcal{C}_1} %(-1)^{\bm{y}\bm{w}^T}
	\sum_{\bm{z}\in \mathcal{C}_1+\bm{y}} \left(e^{\imath\frac{\pi}{p}}\right)^{w_H(\bm{w}\oplus \bm{z})-w_H(\bm{z})}  \sum_{\bm{\gamma} \in \mathcal{C}_2^\perp /\mathcal{C}_1^\perp}(-1)^{\bm{\gamma} \bm{w}^T}\\
	&= \frac{1}{|\mathcal{C}_1|^2} \frac{|\mathcal{C}_1|}{|\mathcal{C}_2|} \sum_{\bm{w}\in \mathcal{C}_2} %(-1)^{\bm{y}\bm{w}^T}
	\sum_{\bm{z}\in \mathcal{C}_1+\bm{y}} \left(e^{\imath\frac{\pi}{p}}\right)^{w_H(\bm{w}\oplus\bm{z})-w_H(\bm{z})}  \\
	&=\frac{1}{|\mathcal{C}_1||\mathcal{C}_2|} \sum_{\bm{w}\in \mathcal{C}_2} \sum_{\bm{z}\in \mathcal{C}_1+\bm{y}} \left(e^{\imath\frac{\pi}{p}}\right)^{w_H(\bm{w})-2w_H(\bm{w}*\bm{z})}, \numberthis \label{eqn:div_last_RZ}
	\end{align*}
	Note that \eqref{eqn:div_last_RZ} implies every term in the double sum is equal to $1$, which completes the proof.
	%Thus, we need all the rotation along the same direction to make each term in the summation of the right hand side contribute 1 so as to equalize the left hand side. 
	
	\subsection{Proof of Lemma \ref{lemma:simplify}}
	\label{subsec:proof_connect_lem_1}
		It follows from \eqref{eqn:A_mu,gamma} that
		\begin{equation} 
		|A_{\bm{0},\bm{\gamma}}(\theta)|^2= \frac{1}{|\mathcal{C}_1|} \sum_{\bm{w}\in \mathcal{C}_1} (-1)^{ \bm{\gamma} \bm{w}^T} s_{\bm{w}}, % \frac{1}{|\mathcal{C}_1|}\sum_{z\in \mathcal{C}_1+y} \left(e^{\imath\theta}\right)^{w_H(w)-2w_H(w*z)}
		\end{equation}
		where 
		\begin{equation}
		s_{\bm{w}} \coloneqq \frac{1}{|\mathcal{C}_1|}\sum_{\bm{z}\in \mathcal{C}_1+\bm{y}} \left(e^{\imath\theta}\right)^{w_H(\bm{w})-2w_H(\bm{w}*\bm{z})}.
		\end{equation}
		Then
		\begin{align*}
		\sum_{\bm{\gamma} \in \mathcal{C}_2^\perp / \mathcal{C}_1^\perp} |A_{\bm{0},\bm{\gamma}}(\theta)|^2 
		&= \frac{1}{|\mathcal{C}_1|} \sum_{\bm{\gamma} \in \mathcal{C}_2^\perp / \mathcal{C}_1^\perp} \left(\sum_{\bm{w}\in \mathcal{C}_2} (-1)^{\bm{\gamma}\bm{w}^T} s_{\bm{w}} + \sum_{\bm{w}\in \mathcal{C}_1\setminus \mathcal{C}_2} (-1)^{\bm{\gamma} \bm{w}^T} s_{\bm{w}}\right)\\
		&=  \frac{1}{|\mathcal{C}_1|} \sum_{\bm{\gamma} \in \mathcal{C}_2^\perp / \mathcal{C}_1^\perp} \sum_{\bm{w}\in \mathcal{C}_2}s_{\bm{w}} + \frac{1}{|\mathcal{C}_1|}\sum_{\bm{w}\in \mathcal{C}_1\setminus \mathcal{C}_2} \sum_{\bm{\gamma} \in \mathcal{C}_2^\perp / \mathcal{C}_1^\perp} (-1)^{\bm{\gamma} \bm{w}^T}s_{\bm{w}} \\
		&= \frac{1}{|\mathcal{C}_1|} \frac{|\mathcal{C}_1|}{|\mathcal{C}_2|} \sum_{\bm{w}\in \mathcal{C}_2}s_{\bm{w}} = \frac{1}{|\mathcal{C}_2|}\sum_{\bm{w}\in \mathcal{C}_2}s_{\bm{w}},\numberthis \label{eqn:connect_lem}
		\end{align*}
		where the last step follows from the fact for any $\bm{w}\in \mathcal{C}_1 \setminus \mathcal{C}_2$, $\sum_{\bm{\gamma} \in \mathcal{C}_2^\perp / \mathcal{C}_1^\perp} (-1)^{\bm{\gamma} \bm{w}^T} = 0$.
		Thus, \eqref{eqn:connect_lem} equals $1$ 
% 		\eqref{eqn:preserved_by_Uz} holds
		%$\sum_{\bm{\gamma} \in \mathcal{C}_2^\perp /\mathcal{C}_1^\perp } \left|A_{\bm{0},\bm{\gamma}}(\theta)\right|^2 =1$
		if and only if $s_{\bm{w}}=1$ for all $\bm{w}\in \mathcal{C}_2 $. Note that $s_{\bm{0}} = 1$, and for all non-zero $\bm{w}$, we have
		\begin{align*}
		s_{\bm{w}} 
		&= \frac{1}{|\mathcal{C}_1|} \sum_{\bm{z}\in \mathcal{C}_1} \left(e^{\imath \theta}\right)^{w_H(\bm{w})-2w_H(\bm{w}*(\bm{z}\oplus \bm{y}))} \\
		%= \frac{|\mathcal{K}_{\bm{w}}|}{|\mathcal{C}_1|} \sum_{\bm{v}\in \mathcal{C}_1/\mathcal{K}_{\bm{w}}} \left(e^{\imath \theta}\right)^{w_H(\bm{w})-2w_H(\bm{w}*(\bm{z}\oplus \bm{y}))}.
		&=\frac{1}{|\mathcal{D}_{\bm{w}}|} \sum_{\bm{v}\in \mathcal{D}_{\bm{w}}}\left(e^{\imath \theta}\right)^{w_H(\bm{w}*(\bm{v}\oplus \bm{y}))}\\
		&=\frac{1}{|\mathcal{D}_{\bm{w}}|}\sum_{\bm{x}\in \mathcal{D}_{\bm{w}}+\bm{w}*\bm{y}} \left(e^{\imath\theta}\right)^{w_H(\bm{w})-2w_H(\bm{x})}.\numberthis
		\end{align*}
		%Since $|\mathcal{D}_{\bm{w}}| = |\mathcal{C}_1 / \mathcal{K}_{\bm{w}}| = \frac{|\mathcal{C}_1|}{|\mathcal{K}_{\bm{w}}|}$, we have
% 		\begin{equation}
% 		s_{\bm{w}} = \frac{1}{|\mathcal{D}_{\bm{w}}|} \sum_{\bm{v}\in \mathcal{D}_{\bm{w}}}\left(e^{\imath \theta}\right)^{w_H(\bm{w}*(\bm{v}\oplus \bm{y}))} = \frac{1}{|\mathcal{D}_{\bm{w}}|}\sum_{\bm{x}\in \mathcal{D}_{\bm{w}}+\bm{w}*\bm{y}} \left(e^{\imath\theta}\right)^{w_H(\bm{w})-2w_H(\bm{x})}.
% 		\end{equation}
		Thus, $\sum_{\bm{\gamma}\in \mathcal{C}_2^\perp /\mathcal{C}_1^\perp } \left|A_{\bm{0},\bm{\gamma}}(\theta)\right|^2 =1$ if and only if \eqref{eqn:D_w} holds 
		%$ \frac{1}{|\mathcal{D}_{\bm{w}}|}\sum_{\bm{x}\in \mathcal{D}_{\bm{w}}+\bm{w}*\bm{y}} \left(e^{\imath\theta}\right)^{w_H(\bm{w})-2w_H(\bm{x})} =1$ 
		for all non-zero $\bm{w}\in \mathcal{C}_2 $. \qed

	\subsection{Proof of Lemma \ref{lemma:space_equal}}
	\label{subsec:proof_connect_lem_2}
		We first show that $\mathcal{D}_{\bm{w}} + \bm{w}*\bm{y} \subseteq \mathrm{proj}_{\bm{w}}(\tilde{\mathcal{Z}}_{\bm{w}}^\perp) + \bm{y}'$. %, and then prove that they have the same size. 
		Let $\bm{z}\in \mathcal{C}_1$. Then, $\bm{w}*\bm{z}\oplus \bm{w}*\bm{y} \in \mathcal{D}_{\bm{w}} + \bm{w}*\bm{y}$. Let $\bm{v}\in \mathcal{Z}_{\bm{w}} \subseteq \mathcal{C}_1^\perp $. We observe
		\begin{equation}
		\left(\bm{w}*(\bm{z}\oplus \bm{y}) \oplus \bm{y}'\right)*\bm{v} = \bm{z}*\bm{w}*\bm{v} \oplus \bm{y}*\bm{w}*\bm{v} \oplus \bm{y}'*\bm{v} = \bm{z}*\bm{v} \oplus \bm{y}*\bm{v} \oplus \bm{y}'*\bm{v},
		\end{equation}
		where the last step follows from $\mathrm{supp}(\bm{x})\subseteq \mathrm{supp}(\bm{w})$. Since $\bm{x}\in \mathcal{C}_1^\perp$ and $\bm{z}\in \mathcal{C}_1$, $w_H(\bm{z}*\bm{v}) = 0 \bmod 2$. We consider two cases. If $\bm{v}\in \mathcal{B}_{\bm{w}} \subseteq \mathcal{Z}_{\bm{w}}$, then $w_H(\bm{y}*\bm{v}) = 0 \bmod 2$ and $w_H(\bm{y}'*\bm{v}) = 0 \bmod 2$. Otherwise, $\bm{v}\in \mathcal{Z}_{\bm{w}}\setminus \mathcal{B}_{\bm{w}}$. Then $w_H(\bm{y}*\bm{v}) = 1 \bmod 2$ and $w_H(\bm{y}'*\bm{v}) = 1 \bmod 2$. For both cases, $w_H(	\left(\bm{w}*(\bm{z}\oplus \bm{y}) \oplus \bm{y}'\right)*\bm{v}) = 0 \bmod 2$. Thus, $	\bm{w}*(\bm{z}\oplus \bm{y}) \oplus \bm{y}' \in \mathrm{proj}_{\bm{w}}(\tilde{\mathcal{Z}}_{\bm{w}}^\perp)$, which implies that $\bm{w}*(\bm{z}\oplus \bm{y}) \in \mathrm{proj}_{\bm{w}}(\tilde{\mathcal{Z}}_{\bm{w}}^\perp)+\bm{y}'$. Then, we have $\mathcal{D}_{\bm{w}} + \bm{w}*\bm{y} \subseteq \mathrm{proj}_{\bm{w}}(\tilde{\mathcal{Z}}_{\bm{w}}^\perp) + \bm{y}'$.
		
		It remains to show that $|\mathcal{D}_{\bm{w}}| = | \mathrm{proj}_{\bm{w}}(\tilde{\mathcal{Z}}_{\bm{w}}^\perp) |$. We observe that $\mathcal{D}_{\bm{w}}=\mathcal{C}_1\big|_{\bm{1}-\bm{w}} = (\mathcal{C}_1^\perp\big|_{\bm{1}-\bm{w}})^\perp$. Thus, $\dim(\mathcal{D}_{\bm{w}}) = w_H(\bm{w}) - d_{\bm{w}} = \dim(\mathcal{Z}_{\bm{w}}^\perp) = \dim(\mathrm{proj}_{\bm{w}}(\tilde{\mathcal{Z}}_{\bm{w}}^\perp) ),$ which completes the proof. \qed
% 		Let $d_{\bm{w}} = \dim(\mathcal{Z}_{\bm{w}})$ and $d_{\mathcal{C}} = \dim(\mathcal{C}_1^\perp)$. Let $\mathcal{Z}_{\bm{w}} = \langle (\bm{a_1},\bm{0}), \dots, (\bm{a_{d_w}},\bm{0})\rangle$ and $\mathcal{C}_1^\perp = \langle (\bm{a_1},\bm{0}), \dots, (\bm{a_{d_w}},\bm{0}), (\bm{c_1},\bm{b_1}),$ $ \dots, (\bm{c_{d_C-d_w}},\bm{d_{d_C-d_w}}) \rangle$. 
% 		\begin{equation}
% 		G_{\mathcal{C}_1^\perp}=\left[\begin{array}{c|c}
% 		\bm{a_1} & \bm{0} \\ 
% 		\vdots & \vdots\\
% 		\bm{a_{d_w}} & \bm{0} \\
% 		\hline
% 		\bm{c_1} & \bm{b_1} \\
% 		\vdots & \vdots\\
% 		\bm{c_{d_C-d_w}} &  \bm{b_{d_C-d_w}} 
% 		\end{array} \right].
% 		\end{equation}
% 		It follows from $|\mathcal{D}_{\bm{w}}| = \frac{|\mathcal{C}_1|}{|\mathcal{K}_{\bm{w}}|}$ that 
% 		\begin{equation}
% 		\dim(\mathcal{D}_{\bm{w}}) = n-d_C-\dim(\mathcal{K}_{\bm{w}}). 
% 		\end{equation}
% 		Assume $\mathcal{T} = \{\bm{b_1},\dots, \bm{b_{dC-d_w}}\}$ is linear dependent, then there exists a subset $\mathcal{T}'\subset \mathcal{T}$ summing to zero. However, combined with the prefixes $\bm{c}$'s, $(\bm{c_i},\bm{d_i})$ should belong to $\mathcal{Z}_{\bm{w}}$ by definition. Thus, $\mathcal{T}$ must be linear independent and 
% 		\begin{equation}
% 		\dim(\mathcal{K}_{\bm{w}}) = n-w_H(\bm{w})-(d_C-d_{\bm{w}}).
% 		\end{equation}
% 		Thus, 
% 		\begin{equation}\label{eqn:size}
% 		\dim(\mathcal{D}_{\bm{w}}) = n- d_C - n + w_H(\bm{w}) + d-d_{\bm{w}} = w_H(\bm{w}) - d_{\bm{w}} = \dim(\mathcal{Z}_{\bm{w}}^\perp) = \dim(\mathrm{proj}_{\bm{w}}(\tilde{\mathcal{Z}}_{\bm{w}}^\perp) ), 
% 		\end{equation}
% 		which completes the proof.
	
	\subsection{Proof of Lemma \ref{lemma:Mac1}}
	\label{subsec:proof_connect_lem_3}
		We rewrite \eqref{eqn:sec} as
		\begin{equation}
		2\sum_{\bm{v}\in \mathcal{B}_{\bm{w}}} \left(\imath \tan\theta \right)^{w_H(\bm{v})} -\sum_{\bm{v}\in \mathcal{Z}_{\bm{w}}} \left(\imath \tan\theta \right)^{w_H(\bm{v})}= \left(\sec\theta\right)^{w_H(\bm{w})},
		\end{equation}
		and rearrange to obtain
		\begin{equation}
		2\sum_{\bm{v}\in \mathcal{B}_{\bm{w}}}\left(\cos\theta\right)^{w_H(\bm{w})-w_H(\bm{v})}\left(\sin\theta\right)^{w_H(\bm{v})} 
		- \sum_{\bm{v}\in \mathcal{Z}_{\bm{w}}}\left(\cos\theta\right)^{w_H(\bm{w})-w_H(\bm{v})}\left(\sin\theta\right)^{w_H(\bm{v})} = 1.
		% 2P[B_w] - P[Z_w] = 1
		\end{equation}
		We apply the MacWilliams Identities to $P_{2\theta}[\mathcal{B}_{\bm{w}}]$ and $P_{2\theta}[\mathcal{Z}_{\bm{w}}]$ ($P_{\theta}[\mathcal{C}]$ is deifned in \eqref{eqn:key_sub_Mac} for any angle $\theta$ and linear code $\mathcal{C}$) to obtain
		\begin{equation}\label{eqn:Mac3}
		\frac{2}{|\mathcal{B}_{\bm{w}}^\perp|} \sum_{\bm{z}\in \mathcal{B}_{\bm{w}}^\perp} \left(e^{\imath \theta}\right)^{w_H(\bm{w})-2w_H(\bm{z})} - \frac{1}{|\mathcal{Z}_{\bm{w}}^\perp|} \sum_{\bm{z}\in \mathcal{Z}_{\bm{w}}^\perp} \left(e^{\imath \theta}\right)^{w_H(\bm{w})-2w_H(\bm{z})} = 1.
		\end{equation}
		Since $|\mathcal{B}_{\bm{w}}^\perp| = 2|\mathcal{Z}_{\bm{w}}^\perp|$,  $\mathcal{B}_{\bm{w}}^\perp = \mathrm{proj}_{\bm{w}}(\tilde{\mathcal{B}}_{\bm{w}}^\perp) $, and $\mathcal{Z}_{\bm{w}}^\perp = \mathrm{proj}_{\bm{w}}(\tilde{\mathcal{Z}}_{\bm{w}}^\perp)$, we obtain %Then \eqref{eqn:Mac3} becomes
		\begin{equation}
		\frac{1}{\left|\mathrm{proj}_{\bm{w}}(\tilde{\mathcal{Z}}_{\bm{w}}^\perp)\right|}
		\sum_{\bm{v}\in \mathrm{proj}_{\bm{w}}(\tilde{\mathcal{Z}}_{\bm{w}}^\perp) + \bm{y'}} \left(e^{i\theta}\right)^{w_H(\bm{w})-2w_H(\bm{v})} = 1,
		\end{equation}
		which completes the proof. \qed

\end{document}